\PassOptionsToPackage{table}{xcolor}
\documentclass[a4paper,onecolumn,11pt,unpublished]{quantumarticle}
\pdfoutput=1

\usepackage[T1]{fontenc}
\usepackage[utf8]{inputenc}
\usepackage[english]{babel}

\usepackage{amsmath}
\usepackage{amssymb}
\usepackage{makecell}
\usepackage{amsthm}
\usepackage{mathtools}
\usepackage{microtype}

\newcommand{\ket}[1]{\lvert #1\rangle}
\newcommand{\bra}[1]{\langle #1\rvert}

\usepackage{graphicx}
\usepackage{tikz}
\usepackage{quantikz}
\usetikzlibrary{
    arrows.meta,
    shadows,
    calc,
    positioning,
    decorations.pathreplacing,
    fit,
    backgrounds
}
\usepackage{pdflscape}
\usepackage{array}
\usepackage{booktabs}
\usepackage{multirow}
\usepackage{tabularx}
\usepackage{float}
\usepackage{algorithm}
\usepackage{algpseudocode}
\usepackage{enumitem}
\newskip\PrismOriginalFootinsSkip
\PrismOriginalFootinsSkip=\skip\footins
\AddToHookNext{shipout/after}{\global\skip\footins=\PrismOriginalFootinsSkip}
\usepackage{csquotes}

\usepackage[
  backend=biber,
  style=apa,
  sorting=nyt,
  doi=true,
  url=true,
  isbn=false,
  eprint=false,
  giveninits=false,
  uniquename=false
]{biblatex}

\usepackage{comment}

\usepackage{subcaption}
\usepackage{hyperref}

\hypersetup{
    allcolors=quantumviolet,
    pdfsubject={Master's Thesis in Quantitative Finance},
    pdfkeywords={Credit Valuation Adjustment, quantum amplitude estimation, quantum finance, noisy quantum hardware}
}

\makeatletter
\let\PrismOriginalContentsline\contentsline
\newcommand{\PrismPageTopContentsline}[4]{%
    \PrismOriginalContentsline{#1}{#2}{#3}{page.#3}%
}
\newcommand{\PrismTableOfContents}{%
    \begingroup
    \setcounter{tocdepth}{2}%
    \small
    \color{quantumviolet}%
    \hypersetup{colorlinks=true,linkcolor=quantumviolet}%
    \let\contentsline\PrismPageTopContentsline
    \tableofcontents
    \endgroup
}
\makeatother

\newtheorem{proposition}{Proposition}

\algrenewcommand\algorithmiccomment[1]{\hfill\(\triangleright\) #1}

\title{A Noise-Aware Quantum Algorithm for Credit Valuation Adjustments on Real Quantum Hardware}

\author[1]{Guillem Borràs Espert}
\email{guiboes@alumni.uv.es}

\author[2]{Francisco Gómez Casanova}
\email{fgomezc@bbva.com}

\author[3]{Luis de Pedro Sánchez}
\email{luis.depedro@uam.es}

\author[4]{Senaida Hernández Santana}
\email{senaida.hernandezsantana@bbva.com}

\author[4]{Pablo Serrano Molinero}
\email{pablo.serrano.molinero@bbva.com}

\affil[1]{Universitat de València, València}

\affil[2]{BBVA Corporate \& Investment Banking, Madrid, Spain}
\affil[3]{Universidad Autónoma de Madrid, Madrid, Spain}
\affil[4]{BBVA Quantum, Madrid, Spain}

\date{July 2026}

\begin{document}

\maketitle

\begin{abstract}
Credit Valuation Adjustments (CVA) require repeated risk-neutral expectation estimation, making them a natural test bed for quantum amplitude estimation, whose coherent amplification can in principle reduce Monte Carlo sampling cost. Whether this advantage survives realistic financial encoding and noisy hardware remains open. We develop an end-to-end, noise-aware quantum workflow for CVA, covering market calibration, discretisation, oracle construction, hardware execution and error-budget analysis. The model combines a correlated two-asset exposure with discount and default factors, encoded through a QCBM-based joint time--market distribution and controlled payoff rotations. We introduce contrast-aware Bayesian iterative quantum amplitude estimation (CABIQAE), which incorporates experimentally calibrated Grover-contrast loss into Bayesian inference and circuit-depth selection. Hardware-calibrated experiments show that CABIQAE exploits the limited amplification available on current devices more effectively than noise-agnostic alternatives and achieves a much lower classical post-processing runtime than the noise-aware BAE baseline. The analysis further decomposes the total CVA error into statistical, encoding, discretisation and hardware contributions. The full CVA oracle remains limited by circuit depth and discretisation resolution.
\end{abstract}


\begin{center}
\begin{minipage}{0.9\textwidth}
\centering\footnotesize\itshape
\textbf{Institutional and hardware access statement.} This work was undertaken solely for academic purposes and does not necessarily reflect the views of BBVA. Access to the IBM Quantum hardware used in this study was made possible through the support of Basque Quantum (BasQ).
\end{minipage}
\end{center}

\clearpage
\PrismTableOfContents
\clearpage

\section{Introduction}
\subsection{Motivation and objectives}

Credit Valuation Adjustment (CVA) emerged as a central component of derivative valuation after the global financial crisis, when large losses revealed that the market value of derivatives could not be separated from the credit quality of the counterparties involved. The financial risk underlying this adjustment is counterparty credit risk, understood in the regulatory terminology of the European Banking Authority as the risk that a counterparty defaults before the final settlement of the transaction cash flows \parencite{EBA_CCR}. In derivatives portfolios, such a default generates an economic loss only when the portfolio has positive value for the surviving party at the time of default. CVA accounts for this effect by adjusting the risk-free value of a derivative portfolio for the expected loss associated with possible counterparty default.

From a quantitative perspective, CVA is a prototypical example of a financial
quantity whose accurate computation depends on repeated estimation of
risk-neutral expectations. In realistic derivative portfolios, the exposure
depends on several market factors and netting effects. This makes Monte Carlo simulation the
standard numerical tool for CVA computation, but also makes high-precision
estimation computationally expensive, since the sampling error decreases only
as \(\mathcal{O}(N_{\mathrm{MC}}^{-1/2})\). This slow convergence is precisely
the point at which quantum computing becomes relevant for this work. Certain quantum algorithms allow an expectation value, such as the discretised CVA functional considered here, to be reformulated as a probability-estimation problem. In this formulation, the target quantity is encoded into the probability that a quantum circuit returns a specified measurement outcome. Once the financial quantity has been
cast into this form, a quantum procedure can estimate that probability more
efficiently than direct sampling by using coherent amplification before
measurement. The central algorithmic primitive used for this purpose is called
quantum amplitude estimation. In the ideal model, it reduces the query
complexity\footnote{Throughout this work, a \emph{query} denotes one use of
the relevant computational primitive: in the quantum setting, the corresponding
quantum circuit, or its inverse when applicable; in the classical Monte Carlo
setting, the corresponding sampler. The terms \emph{ideal} and \emph{noiseless}
are used interchangeably to denote models or simulations with no hardware
noise.} for estimating an expectation to accuracy
\(\varepsilon\) from the classical Monte Carlo scaling
\(\mathcal{O}(\varepsilon^{-2})\) to \(\mathcal{O}(\varepsilon^{-1})\), giving a
quadratic improvement in the dependence on the target accuracy. This makes CVA
a natural candidate for studying whether quantum-enhanced Monte Carlo methods
can become useful for quantitative finance.

However, turning this asymptotic promise into a practical financial workflow is
far from automatic. A quantum CVA algorithm must encode a joint distribution
over time and market states, evaluate the positive exposure of a portfolio,
include discount and default factors, and estimate the resulting amplitude
under severe hardware constraints. The first circuit-level proposal by
\textcite{Alcazar2021CVA} established this direction, but considered a simple
single-option instance and assessed the algorithm through numerical simulation
and resource estimation rather than execution on real quantum hardware.
Moreover, the final estimation step remains a central bottleneck. Once the CVA
quantity has been embedded into a quantum circuit, the algorithm must extract
the corresponding numerical value from repeated quantum measurements. In
ideal settings, this can be accelerated by using Grover\footnote{Here and
throughout this work, Grover amplification refers to the repeated application
of the reflection-based operator introduced in Grover's quantum search
algorithm \parencite{Grover1996}, which rotates amplitude from the undesired
subspace into the target subspace. An overview of this operator is also given
in Appendix~\ref{appsection:qae}.} amplification, which increases the
information gained per circuit execution. In practice, however, methods that
are statistically powerful in ideal or simulated settings may require deep
amplified circuits or significant classical inference overhead, limiting their
scalability on near-term devices.

\subsection{Contributions and scope}

The main objective of this work is to assess the practical viability of
quantum CVA on real quantum hardware as a complete end-to-end pipeline. In
particular, we investigate whether the theoretical quadratic advantage of
quantum amplitude estimation over Monte Carlo can survive the combined effects
of realistic financial modelling, quantum encoding, finite discretisation and
hardware noise. This general objective is developed through three specific
subgoals. First, we extend the quantum CVA framework developed by
\textcite{Alcazar2021CVA} from a single European call option to a more
realistic multi-asset netting set composed of several derivatives written on
correlated underlyings and calibrated from market data. Second, we study the
amplitude estimation stage in detail and propose a noise-aware Bayesian
iterative algorithm designed to make the estimation procedure more compatible
with the limitations of near-term quantum hardware. Third, we validate the
complete workflow not only in ideal simulation, but also on real IBM Quantum
hardware through a hardware-replay methodology that enables robust statistical
comparisons under experimentally observed noise. The aim is not to claim an
end-to-end runtime advantage, but to provide a realistic and reproducible
assessment of where the quantum advantage can enter, which components currently
limit it, and how hardware-aware amplitude estimation can improve the practical
scalability of quantum CVA.

The algorithmic claim is deliberately restricted. CABIQAE is not presented as a proof of asymptotic superiority over BIQAE or over other QAE variants. Rather, it is designed for the regime observed in the hardware experiments of this work, where Grover amplification remains useful only over a finite contrast window. The relevant question is therefore not whether CABIQAE is universally better as an amplitude-estimation method, but whether it is better matched to the depth-dependent loss of contrast encountered by the trained CVA circuits considered here.

\subsection{Literature review}
\label{chap:literature_review}

Quantum computing has attracted increasing attention in finance because several
core tasks in the field can be formulated as sampling, integration,
optimisation or learning problems. Broad reviews of quantum finance identify
portfolio optimisation, derivative pricing, risk analysis and machine learning
as the main application areas, while also emphasising that the strongest
complexity guarantees typically arise in Monte Carlo-type estimation rather than
in heuristic near-term optimisation \parencite{OrusMugelLizaso2019,
Egger2020QuantumFinance,Bouland2020QuantumFinance}. This is particularly
relevant for quantitative finance, where many pricing and risk-management
problems reduce to the estimation of expectations under a risk-neutral measure.

The algorithmic primitive behind most quantum Monte Carlo proposals is Quantum
Amplitude Estimation (QAE). The original amplitude amplification and estimation
framework of \textcite{Brassard2000QAE} provides a quadratic improvement over
classical sampling in the ideal query model. \textcite{Montanaro2015} later
generalised this perspective by showing that the mean output of a bounded-variance randomised algorithm can be estimated with near-quadratic quantum
speedup. These results motivated several financial applications.
\textcite{WoernerEgger2019QuantumRisk} applied QAE to risk measures such as
Value-at-Risk and Conditional Value-at-Risk, while
\textcite{Stamatopoulos2020OptionPricing} developed quantum circuits for
vanilla, multi-asset and path-dependent option pricing. In parallel,
\textcite{An2021QuantumMLMC} studied quantum-accelerated multilevel Monte Carlo
methods for stochastic differential equations, with direct applications to
mathematical finance.

Within valuation adjustments, the main reference is the pioneering contribution of
\textcite{Alcazar2021CVA}, which introduced the first
circuit-level quantum algorithm for CVA estimation. Their construction starts from the discretised CVA
formula, prepares a joint time--price probability distribution, encodes payoff,
discount and default-probability functions through controlled rotations, and
then estimates the resulting observable using a Bayesian amplitude estimation
method based on engineered likelihood functions. This work is important because
it goes beyond black-box oracle assumptions and studies concrete circuit
components. Nevertheless, the benchmark instance is essentially a single
European call, and the assessment is based on numerical simulation and resource
estimation rather than real-hardware execution. This leaves open a practical question that is central to the present work: even if the CVA functional can be encoded as a quantum amplitude, the useful information extracted from amplified circuits may degrade much faster on hardware than the ideal query model suggests. In that case, the amplitude-estimation routine must be evaluated not only by its ideal query complexity, but also by how it behaves when the Grover oscillation loses contrast with depth. In addition, the amplitude-estimation routine adopted in that work is essentially robust amplitude estimation (RAE) implemented through engineered likelihood functions, which introduces
a further scalability concern on the classical side. As discussed by
\textcite{Ramoa2025bayesianquantum}, RAE generalises the fixed Grover
reflections, originally equal to \(\pi\), into arbitrary rotations. This changes
the classical inference problem from optimising over a single amplification
choice to working with \(2m\) continuous parameters, where \(m\) is the number
of circuit layers or Grover applications. For CVA, where each circuit evaluation 
is already costly because it involves non-negligible quantum subroutines, this 
enlarged optimisation space can become a major practical bottleneck.

A related theoretical direction was developed by \textcite{Han2022QuantumCVA},
who analysed multi-option portfolio pricing and
CVA within a quantum Monte Carlo framework. Their results provide
query-complexity guarantees under bounded-variance assumptions and clarify when
CVA-like quantities can inherit a quadratic speedup from amplitude estimation.
This complements the construction of \textcite{Alcazar2021CVA}, but it does not replace the need for an
end-to-end circuit-level and hardware-aware implementation: the practical
question remains how to construct the state-preparation and payoff-encoding
blocks, how to control the discretisation error, and how to run amplitude
estimation when real devices destroy useful quantum information beyond a finite
circuit depth. The present work should therefore be read as complementary to such theoretical query-complexity analyses. It does not attempt to prove a new asymptotic complexity theorem for CVA. Instead, it studies which parts of the ideal quantum Monte Carlo advantage survive after the CVA functional has been discretised, encoded into trained quantum circuits, transpiled to hardware-compatible layouts, and subjected to experimentally measured amplification loss.

The amplitude estimation layer is therefore decisive. The canonical QAE
algorithm requires the use of the quantum Fourier transform (QFT) and deep controlled-Grover circuits,
which are unsuitable for current noisy devices. QFT-free variants such as
iterative quantum amplitude estimation (IQAE) \parencite{Grinko2021IQAE} reduce circuit depth while retaining the
quadratic speedup up to logarithmic factors. More recently, Bayesian approaches
have improved the statistical efficiency of amplitude estimation. Bayesian amplitude estimation (BAE) introduces
a noise-aware Bayesian formulation capable of adapting to imperfect
hardware \parencite{Ramoa2025bayesianquantum}. Bayesian iterative quantum amplitude estimation (BIQAE), proposed by
\textcite{Li2026BIQAE}, injects Bayesian inference into the IQAE
structure, replacing confidence intervals with credible intervals and showing
improved sample complexity relative to several QAE baselines. Importantly, \textcite{Li2026BIQAE} also highlight that general Bayesian approaches based on unrestricted Monte
Carlo posterior propagation can incur substantial classical overhead; conjugate
Bayesian updates, as in Beta-BIQAE, are introduced precisely to retain Bayesian
efficiency while improving scalability. 

This work is positioned at the intersection of these gaps. On the financial side, it extends circuit-level quantum CVA beyond the single-option benchmark by constructing a calibrated two-asset netting set with correlated equity dynamics, finite-grid exposure tables, discount factors and default probabilities. On the implementation side, it moves beyond purely ideal simulation by validating the amplitude-estimation layer with real IBM Quantum hardware data and a hardware-replay methodology. On the algorithmic side, it treats amplitude estimation as a separate practical bottleneck and introduces CABIQAE as a contrast-aware Bayesian iterative method for the specific regime in which amplified circuits lose Grover contrast with depth.

The novelty is therefore not a claim that CABIQAE is theoretically superior to BIQAE, nor that the full CVA pipeline already achieves an end-to-end quantum runtime advantage. The contribution is narrower and more testable: the work shows how a realistic CVA workflow can be decomposed into discretisation, encoding, statistical and hardware components, and it demonstrates that incorporating experimentally calibrated contrast loss into the Bayesian amplitude-estimation loop makes the estimator better suited to the noisy amplified circuits available in the present hardware regime.

\section{Theoretical framework}
\label{sec:theoretical_framework}
\subsection{Classical framework}
\label{sec:classical_cva_valuation}

\subsubsection{Risk-neutral multi-asset dynamics}
\label{subsec:risk_neutral_multi_asset_dynamics}

Let $(\Omega,\mathcal{F},\mathbb{F},\mathbb{P})$ be a complete filtered
probability space, with $\mathbb{F}=(\mathcal{F}_t)_{t\in[0,T]}$
satisfying the usual conditions and $T<\infty$. By the First Fundamental
Theorem of Asset Pricing \parencite{HARRISON1979381,HARRISON1981215}, absence
of arbitrage guarantees the existence of a unique equivalent martingale
measure $\mathbb{Q}\sim\mathbb{P}$ under market completeness; all
valuations are henceforth performed under $\mathbb{Q}$, with
$\mathbb{E}^{\mathbb{Q}}_t[\,\cdot\,]:=\mathbb{E}^{\mathbb{Q}}
[\,\cdot\mid\mathcal{F}_t]$. Risk premia are transferred from the drift
via Girsanov's theorem \parencite{karatzas1991brownian}. The risk-free curve
is represented by a flat continuously compounded rate $r$, yielding
discount factor
\begin{equation}\label{eq:discount_factor}
    D(t,u):=\exp[-r(u-t)], \qquad 0\leq t\leq u\leq T.
\end{equation}

The exposure model is a multi-asset lognormal extension of the
Black--Scholes--Merton framework, incorporating continuous dividend
yields, deterministic time-dependent volatilities, and instantaneous
correlation. Let $\mathbf{S}_t=(S_t^{(1)},\ldots,S_t^{(d)})^\top$ be
the vector of $d$ equity underlyings, with asset $k$ paying dividend
yield $q_k$. The volatility of asset $k$ is modelled as a
deterministic piecewise-constant function over the implied-volatility
maturity grid $0=T_0<T_1<\cdots<T_J$, with bucket values
\begin{equation}\label{eq:piecewise_forward_vol}
    \sigma_{k,j}
    =
    \left(
        \frac{
            T_{j+1}\bigl(\sigma^{\mathrm{imp}}_k(T_{j+1})\bigr)^2
            -
            T_j\bigl(\sigma^{\mathrm{imp}}_k(T_j)\bigr)^2
        }{T_{j+1}-T_j}
    \right)^{1/2},
\end{equation}
chosen so that $\int_0^{T_j}\sigma_k^2(u)\,du =
T_j\bigl(\sigma^{\mathrm{imp}}_k(T_j)\bigr)^2$, reproducing the
market-implied total variance at each quoted maturity. Under
$\mathbb{Q}$, the asset dynamics are
\begin{equation}\label{eq:multi_asset_sde}
    dS_t^{(k)}
    =
    (r-q_k)S_t^{(k)}\,dt
    +
    \sigma_k(t)S_t^{(k)}\,d\widetilde{W}_t^{(k)},
    \qquad k=1,\ldots,d,
\end{equation}
with instantaneous correlations
$d\widetilde{W}_t^{(k)}\,d\widetilde{W}_t^{(l)}=\rho_{kl}\,dt$.
The symmetric positive definite correlation matrix
$\boldsymbol{\rho}=(\rho_{kl})$ admits the Cholesky factorisation
$\boldsymbol{\rho}=LL^\top$, yielding the representation
\begin{equation}\label{eq:correlated_brownian_cholesky}
    d\widetilde{\mathbf{W}}_t = L\,d\mathbf{W}_t,
\end{equation}
where $\mathbf{W}$ is a vector of independent standard Brownian
motions under $\mathbb{Q}$.

Applying Itô's formula to $\log S_t^{(k)}$ gives the exact transition
\begin{equation}\label{eq:multi_asset_transition_solution}
    S_t^{(k)}
    =
    S_s^{(k)}
    \exp\!\left[
        (r-q_k)(t-s)
        -\tfrac{1}{2}\int_s^t\sigma_k^2(u)\,du
        +\int_s^t\sigma_k(u)\,d\widetilde{W}_u^{(k)}
    \right], \quad 0\leq s<t\leq T.
\end{equation}
Since $\sigma_k(\cdot)$ is deterministic, the conditional log-return is
Gaussian, and it is convenient to define the root-mean-square volatility
\begin{equation}\label{eq:rms_vol}
    \bar{\sigma}_k(s,t)
    :=
    \left(
        \frac{1}{t-s}\int_s^t\sigma_k^2(u)\,du
    \right)^{1/2},
\end{equation}
under which the marginal law of $S_t^{(k)}$ coincides with that of a
constant-volatility Black--Scholes--Merton model with parameter
$\bar{\sigma}_k(s,t)$. This is
\[
    \log S_t^{(k)} \mid S_s^{(k)}
    \sim
    \mathcal{N}
    \left(
        \log S_s^{(k)}
        +
        (r-q_k)(t-s)
        -
        \frac{1}{2}\bar{\sigma}_k^2(s,t)(t-s),
        \;
        \bar{\sigma}_k^2(s,t)(t-s)
    \right),
\]
so that \(S_t^{(k)}\mid S_s^{(k)}\) is lognormally distributed with
\begin{equation}\label{eq:expected_value_GBM}
        \mu_{S,k}(s,t)
    :=
    \mathbb{E}\!\left[S_t^{(k)}\mid S_s^{(k)}\right]
    =
    S_s^{(k)}e^{(r-q_k)(t-s)},
\end{equation}
and
\begin{equation}\label{eq:vol_GBM}
    \eta_{S,k}(s,t)
    :=
    \operatorname{Std}\!\left[S_t^{(k)}\mid S_s^{(k)}\right]
    =
    S_s^{(k)}e^{(r-q_k)(t-s)}
    \sqrt{
        e^{\bar{\sigma}_k^2(s,t)(t-s)}-1
    }.
\end{equation}
On a monitoring interval $(t_i,t_{i+1}]$ contained within bucket
$(T_j,T_{j+1}]$, the exact simulated transition is
\begin{equation}\label{eq:exact_step}
S_{t_{i+1}}^{(k)}
=
S_{t_i}^{(k)}
\exp\left[
\left(r-q_k-\tfrac{1}{2}\sigma_{k,j}^2\right)\Delta t_i
+\sigma_{k,j}\sqrt{\Delta t_i}
\sum_{l=1}^{d}L_{kl}Z_l^{(i)}
\right],
\end{equation}
where $\Delta t_i=t_{i+1}-t_i$ and
$\mathbf{Z}^{(i)}\overset{\mathrm{iid}}{\sim}\mathcal{N}(\mathbf{0},I_d)$
across monitoring dates. If $(t_i,t_{i+1}]$ intersects several volatility
buckets, the transition is instead generated from the integrated
log-return covariance matrix
\begin{equation}
\Gamma_i^{kl}
:=
\rho_{kl}
\int_{t_i}^{t_{i+1}}
\sigma_k(u)\sigma_l(u)\,du,
\end{equation}
with marginal integrated variances
$V_{k,i}:=\Gamma_i^{kk}=\int_{t_i}^{t_{i+1}}\sigma_k^2(u)\,du$.
Writing $\Gamma_i=L_iL_i^\top$, the exact transition becomes
\begin{equation}
S_{t_{i+1}}^{(k)}
=
S_{t_i}^{(k)}
\exp\left[
(r-q_k)\Delta t_i
-\tfrac{1}{2}V_{k,i}
+
\sum_{l=1}^{d}(L_i)_{kl}Z_l^{(i)}
\right].
\end{equation}
Hence intra-step volatility changes are integrated exactly and no Euler
discretisation bias is incurred. Empirically estimated correlation matrices that fail positive definiteness
are regularised via a minimal spectral shift prior to Cholesky
factorisation.

\subsubsection{Classical portfolio valuation}
\label{sec:classical_valuation}

For an $\mathcal{F}_T$-measurable payoff $X_T$, the default-free
arbitrage-free value at time $t\leq T$ is
\begin{equation}\label{eq:risk_neutral_pricing}
    V_t = \mathbb{E}^{\mathbb{Q}}_t\!\left[D(t,T)X_T\right].
\end{equation}
Since the conditional distribution of $S_T\mid\mathcal{F}_t$ under
deterministic volatility is equivalent to that of a constant-volatility
model with parameter $\bar{\sigma}(t,T)$, the standard
Black--Scholes--Merton pricing formulas apply with this substitution.
Defining
\[
    d_\pm(t,T,K)
    :=
    \frac{
        \log(S_t/K)+\left(r-q\pm\tfrac{1}{2}\bar{\sigma}^2(t,T)\right)(T-t)
    }{\bar{\sigma}(t,T)\sqrt{T-t}},
\]
the time-$t$ values of a European call, put and forward are
\begin{align}
    c(t,T,K) &= S_t e^{-q(T-t)}\Phi(d_+) - Ke^{-r(T-t)}\Phi(d_-),
    \label{eq:call_price}\\
    p(t,T,K) &= Ke^{-r(T-t)}\Phi(-d_-) - S_t e^{-q(T-t)}\Phi(-d_+),
    \label{eq:put_price}\\
    f(t,T,K) &= S_t e^{-q(T-t)} - Ke^{-r(T-t)}.
    \label{eq:forward_price}
\end{align}

Let \(\mathcal{N}\) denote the netting set. Each trade \(\ell\in\mathcal{N}\) references an
underlying asset \(a_\ell\in\{1,\ldots,d\}\), has maturity \(T_\ell\),
position sign \(\phi_\ell\in\{-1,+1\}\), and default-free value
\(V_\ell(S_t^{(a_\ell)},t)\). The mark-to-market value of the netting set is
\begin{equation}\label{eq:portfolio_value}
    V(t)
    :=
    \Pi_{\mathcal{N}}(\mathbf{S}_t,t)
    =
    \sum_{\ell\in\mathcal{N}}
    \phi_\ell\,V_\ell\!\left(S_t^{(a_\ell)},t\right)
    \mathbf{1}_{\{t\leq T_\ell\}},
\end{equation}
with positive and negative parts
\begin{equation}\label{eq:positive_negative_exposure}
    V^+(t):=\max(V(t),0), \qquad V^-(t):=\min(V(t),0).
\end{equation}
The quantity $V^+(\mathbf{S}_t,t)$ constitutes the relevant exposure
for unilateral counterparty credit risk: a positive portfolio value
implies exposure to counterparty default, whereas a negative value
represents an obligation to the counterparty.

\subsubsection{Default modelling and the Credit Valuation Adjustment formula}
\label{sec:cva_model}

Let \(\tau\) denote the counterparty default time and let \(R_{\mathrm{CVA}}\in[0,1]\)
be the recovery rate. We work in a reduced-form credit-risk framework,
where default is modelled through an intensity process calibrated to CDS
market quotes. Technical details on filtrations and intensity-based default modelling can be found, for instance,
in \textcite{BieleckiJeanblancRutkowski2004Hedging} and \textcite{JeanblancLeCam2009Immersion}.

Using the notation introduced in
\eqref{eq:portfolio_value}--\eqref{eq:positive_negative_exposure},
the credit-adjusted value of the netting set on the event \(\{\tau>t\}\)
and under a default-free close-out convention is
\begin{equation}
    V^*(t)
    =
    \mathbb{E}^{\mathbb{Q}}_t
    \!\left[
        D(t,T)X_T\mathbf{1}_{\{\tau>T\}}
        +
        D(t,\tau)
        \bigl(
            R_{\mathrm{CVA}} V^+(\tau)+V^-(\tau)
        \bigr)
        \mathbf{1}_{\{t<\tau\leq T\}}
    \right].
    \label{eq:risky_value_main}
\end{equation}
The close-out payment can be written as
\[
    R_{\mathrm{CVA}} V^+(\tau)+V^-(\tau)
    =
    V(\tau)-(1-R_{\mathrm{CVA}})V^+(\tau).
\]
Substituting this identity into \eqref{eq:risky_value_main} yields
\begin{align}
    V^*(t)
    &=
    \mathbb{E}^{\mathbb{Q}}_t
    \!\left[
        D(t,T)X_T\mathbf{1}_{\{\tau>T\}}
        +
        D(t,\tau)V(\tau)
        \mathbf{1}_{\{t<\tau\leq T\}}
    \right]
    \notag\\
    &\quad
    -
    (1-R_{\mathrm{CVA}})
    \mathbb{E}^{\mathbb{Q}}_t
    \!\left[
        D(t,\tau)V^+(\tau)
        \mathbf{1}_{\{t<\tau\leq T\}}
    \right].
    \label{eq:risky_value_decomposition_main}
\end{align}
The first expectation in \eqref{eq:risky_value_decomposition_main} is the
default-free value \(V(t)\): if no default occurs before maturity, the
terminal payoff is received directly; if default occurs before maturity,
the default-free continuation value \(V(\tau)\) is paid at the default
time. Therefore, by the tower property,
\begin{equation}
    V^*(t)
    =
    V(t)-\mathrm{CVA}(t),
\end{equation}
where
\begin{equation}
    \mathrm{CVA}(t)
    =
    \mathbf{1}_{\{\tau>t\}}(1-R_{\mathrm{CVA}})
    \mathbb{E}^{\mathbb{Q}}_t
    \!\left[
        D(t,\tau)V^+(\tau)
        \mathbf{1}_{\{t<\tau\leq T\}}
    \right]
    \label{eq:cva_general_main}
\end{equation}
denotes the credit valuation adjustment. It represents
the risk-neutral expected discounted loss due to the possible default of the
counterparty before maturity. Since losses arise only when the portfolio has
positive value to the investor, the adjustment depends on the positive exposure
\(V^+(\tau)\), weighted by the loss-given-default factor \((1-R_{\mathrm{CVA}})\) and by the
event that default occurs before maturity. Therefore, the risky value \(V^*(t)\)
is the default-free value \(V(t)\) reduced by the expected loss induced by
counterparty credit risk.
For \(0\leq t\leq u\), define the conditional survival probability
\begin{equation}
    P(t,u)
    :=
    \mathbb{Q}(\tau>u\mid \tau>t,\mathcal{F}_t).
    \label{eq:survival_probability_definition}
\end{equation}
In the deterministic-intensity setting considered in this work,
\begin{equation}
    P(t,u)
    =
    \exp\!\left(
        -\int_t^u \lambda(s)\,ds
    \right).
    \label{eq:survival_probability_compact}
\end{equation}
Since \(u\mapsto P(t,u)\) is non-increasing, the positive default measure
over future time is \(-dP(t,u)\). Under the usual no-wrong-way-risk
approximation adopted in this work, conditional on the market information available at time
\(t\), the exposure process and the default event are assumed independent.
Thus default enters the CVA calculation only through marginal default
probabilities, and \eqref{eq:cva_general_main} becomes
\begin{equation}
    \mathrm{CVA}(t)
    =
    -\mathbf{1}_{\{\tau>t\}}(1-R_{\mathrm{CVA}})
    \int_t^T
        D(t,u)\,
        \mathrm{EE}(t,u)\,
        dP(t,u),
    \label{eq:cva_continuous_compact}
\end{equation}
where
\begin{equation}
    \mathrm{EE}(t,u)
    :=
    \mathbb{E}^{\mathbb{Q}}_t
    \!\left[
        V^+(u)
    \right]
    =
    \mathbb{E}^{\mathbb{Q}}_t
    \!\left[
        V^+(\mathbf{S}_u,u)
    \right],
    \qquad u\geq t.
    \label{eq:expected_exposure_compact}
\end{equation}

The intensity curve is obtained from CDS market quotes by imposing the
standard par-spread condition, namely that the present value of the premium
leg equals the present value of the protection leg
\parencite{HullWhite2000}. For a CDS maturity \(T_{\mathrm{CDS}}=t_n\), the
discrete premium and protection legs satisfy
\begin{equation}
    s(0,T_{\mathrm{CDS}})
    \sum_{i=1}^{n}
        \Delta_i D(0,t_i)P(0,t_i)
    =
    (1-R_{\mathrm{CDS}})
    \sum_{i=1}^{n}
        D(0,t_i)
        \bigl(
            P(0,t_{i-1})-P(0,t_i)
        \bigr),
    \label{eq:cds_par_spread_compact}
\end{equation}
where \(s(0,T_{\mathrm{CDS}})\) is the quoted CDS spread,
\(R_{\mathrm{CDS}}\) is the recovery rate assumed in the CDS market,
\(D(0,t_i)\) is the risk-free discount factor and \(\Delta_i\) is the
year fraction associated with the premium payment over
\((t_{i-1},t_i]\). 

On a monitoring grid \(t=t_0<t_1<\cdots<t_M=T\), define the conditional
default probability over the interval \((t_{i-1},t_i]\), as seen from
valuation time \(t\), by
\begin{equation}
    \Delta q(t_i)
    :=
    P\!\left(t,\max(t,t_{i-1})\right)-P(t,t_i).
    \label{eq:default_increment_compact}
\end{equation}
The discretised CVA is then approximated by
\begin{equation}
    \mathrm{CVA}(t)
    \approx
    \mathbf{1}_{\{\tau>t\}}(1-R_{\mathrm{CVA}})
    \sum_{i=1}^{M}
        D(t,t_i)\,
        \mathrm{EE}(t,t_i)\,
        \Delta q(t_i).
    \label{eq:cva_discrete_compact}
\end{equation}

Given \(N_{\mathrm{MC}}\) independent simulated paths
\(\{\mathbf{S}^{(k)}\}_{k=1}^{N_{\mathrm{MC}}}\), generated
conditionally from time \(t\), the Monte Carlo estimator of the expected
exposure at a future monitoring date \(t_i>t\) is
\begin{equation}
    \widehat{\mathrm{EE}}(t,t_i)
    =
    \frac{1}{N_{\mathrm{MC}}}
    \sum_{k=1}^{N_{\mathrm{MC}}}
    V^+(\mathbf{S}^{(k)}_{t_i},t_i).
    \label{eq:ee_mc_compact}
\end{equation}
The corresponding discretised CVA estimator is
\begin{equation}
    \widehat{\mathrm{CVA}}_{\mathrm{MC}}(t)
    =
    \mathbf{1}_{\{\tau>t\}}(1-R_{\mathrm{CVA}})
    \sum_{i=1}^{M}
        D(t,t_i)\,
        \widehat{\mathrm{EE}}(t,t_i)\,
        \Delta q(t_i).
    \label{eq:cva_mc_compact}
\end{equation}
Conditional on the monitoring grid and on the calibrated deterministic
discount and default curves, this estimator is unbiased for the discretised
quantity in \eqref{eq:cva_discrete_compact}. Moreover, under the log-normal
equity model and the Black--Scholes pricing formulas used for calls, puts
and forwards, the positive exposure has finite second moment. Hence the
central limit theorem yields the standard Monte Carlo root mean square error (RMSE) rate
\[
    \mathrm{RMSE}
    =
    \mathcal{O}
    \!\left(
        N_{\mathrm{MC}}^{-1/2}
    \right).
\]
This rate provides the classical benchmark against which the quantum
amplitude estimation pipeline is compared.

\subsection{Quantum framework}
\label{sec:quantum_framework}
This section shows how the discretised CVA functional can be encoded as the success probability of a quantum circuit. This reformulation allows quantum amplitude estimation, the central quantum primitive of this work reviewed in Appendix~\ref{appsection:qae}, to accelerate the CVA calculation relative to classical Monte Carlo sampling in terms of query complexity. Throughout this work, unless explicitly stated otherwise, the term quantum advantage is used only in this restricted sense. It does not, by itself, imply an end-to-end runtime advantage for a practical CVA calculation. Such a runtime comparison also depends on the cost of state preparation,
controlled rotations, circuit depth, hardware noise, error mitigation or
correction overheads, and queueing time. A complete runtime analysis is beyond the scope of this work; for an in-depth discussion in the context of quantum CVA, we refer the reader to the runtime analysis of \textcite[Section~5]{Alcazar2021CVA}.

We assume familiarity with the standard notation and basic principles of
quantum computation, including qubits, quantum states, unitary circuits,
measurement and projective observables. These concepts are not reviewed here, as they are standard material in quantum computation and are thoroughly covered in the canonical textbook by \textcite{NielsenChuang2010}. 

\subsubsection{Quantum encoding of the Credit Valuation Adjustment}
\label{subsec:quantum_cva_encoding}
A quantum circuit acts on finite-dimensional registers and therefore cannot
represent an unbounded continuum of price values exactly. We thus begin by
truncating the continuous price domain of each underlying \(S^{(k)}\) appearing in the netting set to a
finite interval. In practice, we define
\[
    \mathcal{D}_k
    =
    \left[
    \max\{\mu_{S,k}(0,T) - 3\eta_{S,k}(0,T),0\},
    \mu_{S,k}(0,T) + 3\eta_{S,k}(0,T)
    \right],
\]
clipping the endpoint
at zero because prices are strictly positive under the GBM model. 

After this truncation, the time and market variables must be represented on
finite discrete grids whose cardinalities are compatible with qubit registers.
Let \(\{t_i\}_{i=1}^{M}\), with \(M=2^m\),\footnote{The number of monitoring dates is chosen to be a power of two so that the time index can be encoded exactly in an $m$-qubit register. This avoids unused computational basis states and makes the finite grid directly match the cardinality of the quantum register.} denote the exposure monitoring
dates, and for each underlying \(S^{(k)}\), partition the
truncated domain \(\mathcal{D}_k\) into \(N_k=2^{n_k}\) disjoint bins
\(\{B^{(k)}_{j_k}\}_{j_k=0}^{N_k-1}\). Writing
\(B^{(k)}_{j_k}=[b^{(k)}_{j_k},b^{(k)}_{j_k+1})\), each bin is represented by
its left endpoint, $s^{(k)}_{j_k}=b^{(k)}_{j_k}$,
with the last bin closed on the right\footnote{The left-endpoint convention is
used to avoid introducing an additional arithmetic step in the finite-grid
evaluation of the exposure function.}. A multi-index
\(\mathbf{j}=(j_1,\ldots,j_d)\) then identifies the Cartesian product bin
\(B_{\mathbf{j}} = B^{(1)}_{j_1} \times \cdots \times B^{(d)}_{j_d}\), so
that each pair \((t_i,B_{\mathbf{j}})\) represents one cell of the
discretised time--market grid. The associated quantum representation uses
\(m\) qubits for the time register and
\(n := \sum_{k=1}^{d} n_k\)
qubits for the asset registers, with computational basis state
\[
    \ket{i}_T\ket{\mathbf{j}}_S
    :=
    |i\rangle_T\,|j_1\rangle_{S_1}\cdots|j_d\rangle_{S_d}
\]
encoding the market-grid cell corresponding to monitoring date \(t_i\) and
multi-asset price bin \(B_{\mathbf{j}}\).

For notational clarity, and without loss of generality for the encoding
construction, we fix the valuation date to be the current date, denoted by
\(t_{\mathrm{val}}=0\). All quantities appearing below are therefore understood
as time-zero, risk-neutral quantities, conditional on the information available
at the valuation date and, when relevant, on survival up to that date. Under this 
convention, and before the CVA can be embedded into a quantum amplitude estimation routine,
its continuous and semi-discrete forms must first be represented on this
finite time--market grid. Starting from the discrete-time approximation
\eqref{eq:cva_discrete_compact}, the expected exposure at each monitoring
date is replaced by a finite sum over the market-grid bins, giving
\[
    \mathrm{EE}_{\Delta}(0,t_i) := \mathrm{EE}_{\Delta}(t_i)
    =
    \sum_{\mathbf{j}}
    \mathbb{Q}_{\mathcal{D}}
    \!\left(
        \mathbf{S}_{t_i}\in B_{\mathbf{j}}
    \right)
    V^+_{i,\mathbf{j}},
\]
and therefore
\begin{equation}
   \mathrm{CVA}_{\Delta}(t=0) := \ \mathrm{CVA}_{\Delta}
    =
    (1-R_{\mathrm{CVA}})
    \sum_{i=1}^{M}
    D(0,t_i)\,\Delta q(t_i)
    \sum_{\mathbf{j}}
    \mathbb{Q}_{\mathcal{D}}
    \!\left(
        \mathbf{S}_{t_i}\in B_{\mathbf{j}}
    \right)
    V^+_{i,\mathbf{j}}.
    \label{eq:cva_discrete_grid}
\end{equation}
where the subscript \(\Delta\) collectively refers to the time discretisation
and to the market-grid bin widths.

At this stage the probabilities entering
the exposure are still conditional market probabilities at each fixed date
\(t_i\); however, the quantum circuit used for amplitude estimation must
prepare a single coherent superposition over the full time--market register,
and it is therefore necessary to reorganise them into a single discrete
probability law over the Cartesian product of the monitoring dates and the
market-grid bins. To that end, following the construction of
\textcite{Alcazar2021CVA}, we introduce the joint probability tensor.
For each monitoring date $t_i$, the risk-neutral density
$f_{\mathbf{S}_{t_i}}$ of the $d$-dimensional price vector is defined over
$\mathbb{R}^d_{+}$, whereas the qubit register can only represent prices
within the finite domain $\mathcal{D}=\bigcup_{\mathbf{j}}B_{\mathbf{j}}$.
The mass of $f_{\mathbf{S}_{t_i}}$ outside $\mathcal{D}$ is accommodated by
conditioning on the event $\{\mathbf{S}_{t_i}\in\mathcal{D}\}$, which amounts
to renormalising the density by

\[
    \mathcal{Z}_i
    \;=\;
    \int_{\mathcal{D}}
    f_{\mathbf{S}_{t_i}}(\mathbf{s})\,d\mathbf{s} \;\in\; (0,1].
\]
so that the discrete probability assigned to bin $B_{\mathbf{j}}$ at date $t_i$
is the conditional bin probability
\[
    \mathbb{Q}_{\mathcal{D}}\!\left(\mathbf{S}_{t_i}\in B_{\mathbf{j}}\right)
    \;=\;
    \frac{1}{\mathcal{Z}_i}
    \int_{B_{\mathbf{j}}}
    f_{\mathbf{S}_{t_i}}(\mathbf{s})\,d\mathbf{s}.
\]
For the purpose of state preparation, we choose a uniform marginal
$\pi_i = 1/M$ over the monitoring dates. The joint probability tensor is therefore defined as
\begin{equation}
    \mathcal{P}_{i,\mathbf{j}}
    \;:=\;
    \frac{1}{M\,\mathcal{Z}_i}
    \int_{B_{\mathbf{j}}}
    f_{\mathbf{S}_{t_i}}(\mathbf{s})\,d\mathbf{s},
    \qquad
    \sum_{i=1}^{M}\sum_{\mathbf{j}}
    \mathcal{P}_{i,\mathbf{j}} = 1,
    \label{eq:joint_probability_tensor}
\end{equation}
which generalises the two-dimensional time--price matrix of
\textcite{Alcazar2021CVA} to a $(d+1)$-dimensional tensor encoding the full
dependence structure of the correlated multi-asset dynamics.

With the joint probability tensor in hand, we can now construct the quantum
circuit that encodes the CVA. The first building block is the
state-preparation unitary\footnote{The term ``unitary''
refers to the linear operator acting on the Hilbert space of the quantum
registers. Its physical implementation is a quantum circuit whose net action,
in the ideal noiseless model, realises that unitary transformation. Thus
\(G_{\mathcal{P}}\) denotes both the target state-preparation transformation
and, by standard abuse of notation, the circuit implementing it.} \(G_{\mathcal{P}}\), defined by its action on the computational zero state as
\begin{equation}
    G_{\mathcal{P}}\,|0\rangle^{\otimes(m+n)}
    \;=\;
    \sum_{i=1}^{M}\sum_{\mathbf{j}}
    \sqrt{\mathcal{P}_{i,\mathbf{j}}}\;
    |i\rangle_T\,|j_1\rangle_{S_1}\cdots|j_d\rangle_{S_d}.
    \label{eq:state_preparation_multi}
\end{equation}
This unitary coherently encodes the joint
time--market distribution in the amplitudes of the quantum state: by Born's
rule, measuring the register in the computational basis returns the grid point
\((i,\mathbf{j})\) with probability
\(\mathcal{P}_{i,\mathbf{j}}\). With correlated
underlyings, this distribution is generally non-factorisable across asset
registers, so the resulting state may be entangled across the time and market
registers. The state-preparation ansatz circuit must therefore capture these
cross-register dependencies; its architecture is described in
Section~\ref{subsec:methodology}.

Once the market state has been coherently prepared, the remaining scalar
quantities in \eqref{eq:cva_discrete_grid} must be converted into valid
probabilities before they can be encoded in ancilla amplitudes. We therefore
introduce classical normalisation constants \(C_v,C_p,C_q>0\) and define
\begin{equation}\label{eq:scaling_cts_definition}
    \widetilde{v}_{i,\mathbf{j}} := \frac{V^+_{i,\mathbf{j}}}{C_v}, \qquad
    \widetilde{p}_i             := \frac{D(0,t_i)}{C_p},            \qquad
    \widetilde{q}_i             := \frac{\Delta q(t_i)}{C_q},
\end{equation}
with \(C_v,C_p,C_q\) chosen so that all three rescaled quantities lie in
\([0,1]\). The original financial scale is restored classically at the end of
the computation through the factor \(C_vC_pC_q\).

The rescaled quantities are then encoded into three dedicated ancilla
qubits\footnote{An ancilla qubit is an auxiliary qubit appended to the
computational register, typically initialised in \(\ket{0}\), and used either as
workspace or as a readout qubit for a particular computational feature. Such
auxiliary degrees of freedom are also needed to embed the desired computation
into a reversible, hence unitary, operation. In the present construction, the
ancilla qubits store, through their amplitudes, the rescaled quantities later
combined into the marked success event. See, e.g., \textcite[Section~4.2]{NielsenChuang2010}} by controlled-rotation unitaries
\(R_f \in \{R_v,R_p,R_q\}\). For a rescaled quantity \(\widetilde f\), the
corresponding controlled rotation acts on an ancilla initialised in
\(\ket{0}\) as
\[
    R_f\,\ket{0}
    \;=\;
    \sqrt{1-\widetilde{f}}\,\ket{0}
    \;+\;
    \sqrt{\widetilde{f}}\,\ket{1},
\]
with the value of \(\widetilde f\) selected by the relevant time and/or market
register.

The discount-factor and default-probability rotations, $R_p$ and $R_q$,
condition exclusively on the time register, since $D(0,t_i)$ and $\Delta q(t_i)$
are deterministic functions of $i$.
The exposure rotation $R_v$, by contrast, must condition on the full pair
$(i,\mathbf{j})$, because the netting-set positive exposure
\[
    V^+_{i,\mathbf{j}}
    \;=\;
    \max\!\left\{
        \sum_{\ell\in\mathcal{N}}
        \phi_\ell\,
        V_\ell\!\left(s^{(a_\ell)}_{j_{a_\ell}},t_i\right)
        \mathbf{1}_{\{t_i\leq T_\ell\}},\;0
    \right\}
\]
is a non-linear, non-separable function of the joint price vector
$\mathbf{s}_{\mathbf{j}}=(s^{(1)}_{j_1},\ldots,s^{(d)}_{j_d})$.
The circuit implementations of $R_p$, $R_q$, and $R_v$ are detailed in
Section~\ref{subsec:methodology}.

The full encoding operator is $A:=R_q R_p R_v G_{\mathcal{P}}$, acting on
$m + n + 3$ qubits, and a comprehensive schematic summarising
the joint state preparation and the controlled rotations for the exposure,
discount, and default factors is provided in
Figure~\ref{fig:quantum_cva_encoding_multiasset}. To see explicitly how this circuit
encodes the CVA, let
$|\xi\rangle = A\,|0\rangle^{\otimes(m+n+3)}$ and let
$\Pi_{111} := I \otimes |111\rangle\langle 111|$ denote the projector onto the
subspace where all three ancilla qubits are in state $|1\rangle$. Writing
the prepared state as
\begin{equation}
\begin{aligned}
|\xi\rangle
={}&
\sum_{i=1}^{M}\sum_{\mathbf{j}}
\sqrt{\mathcal{P}_{i,\mathbf{j}}}\,
|i\rangle|\mathbf{j}\rangle \\
&{}\otimes
\left(
\sqrt{1-\widetilde v_{i,\mathbf{j}}}|0\rangle
+
\sqrt{\widetilde v_{i,\mathbf{j}}}|1\rangle
\right) \\
&{}\otimes
\left(
\sqrt{1-\widetilde p_i}|0\rangle
+
\sqrt{\widetilde p_i}|1\rangle
\right) \\
&{}\otimes
\left(
\sqrt{1-\widetilde q_i}|0\rangle
+
\sqrt{\widetilde q_i}|1\rangle
\right),
\end{aligned}
\end{equation}
projecting onto the all-ones ancilla subspace yields
\begin{equation}
\begin{aligned}
a_{\mathrm{CVA}}
&:=
\langle \xi|\Pi_{111}|\xi\rangle \\
&=
\sum_{i=1}^{M}\sum_{\mathbf{j}}
\mathcal{P}_{i,\mathbf{j}}\,
\left|
\sqrt{\widetilde v_{i,\mathbf{j}}}
\sqrt{\widetilde p_i}
\sqrt{\widetilde q_i}
\right|^2 \\
&=
\sum_{i=1}^{M}\sum_{\mathbf{j}}
\mathcal{P}_{i,\mathbf{j}}\,
\widetilde v_{i,\mathbf{j}}\,
\widetilde p_i\,
\widetilde q_i .
\end{aligned}
\label{eq:cva_amplitude}
\end{equation}

Substituting
$\mathcal{P}_{i,\mathbf{j}} = \pi_i\,\mathbb{Q}_{\mathcal{D}}
(\mathbf{S}_{t_i}\in B_{\mathbf{j}})$ with $\pi_i=1/M$, and
restoring the original scale by inverting \eqref{eq:scaling_cts_definition} one obtains
\begin{align*}
a_{\mathrm{CVA}}
&=
\sum_{i=1}^{M}\sum_{\mathbf{j}}
\frac{1}{M}\,
\mathbb{Q}_{\mathcal{D}}
\!\left(
\mathbf{S}_{t_i}\in B_{\mathbf{j}}
\right)
\frac{V^+_{i,\mathbf{j}}}{C_v}
\frac{D(0,t_i)}{C_p}
\frac{\Delta q(t_i)}{C_q} \\
&=
\frac{1}{M C_v C_p C_q}
\sum_{i=1}^{M}\sum_{\mathbf{j}}
\mathbb{Q}_{\mathcal{D}}
\!\left(
\mathbf{S}_{t_i}\in B_{\mathbf{j}}
\right)
V^+_{i,\mathbf{j}}D(0,t_i)\Delta q(t_i) .
\end{align*}

Hence
\begin{align*}
M C_v C_p C_q\,a_{\mathrm{CVA}}
&=
\sum_{i=1}^{M}\sum_{\mathbf{j}}
\mathbb{Q}_{\mathcal{D}}
\!\left(
\mathbf{S}_{t_i}\in B_{\mathbf{j}}
\right)
V^+_{i,\mathbf{j}}D(0,t_i)\Delta q(t_i) \\
&=
\sum_{i=1}^{M}
D(0,t_i)\Delta q(t_i)
\sum_{\mathbf{j}}
\mathbb{Q}_{\mathcal{D}}
\!\left(
\mathbf{S}_{t_i}\in B_{\mathbf{j}}
\right)
V^+_{i,\mathbf{j}} .
\end{align*}

Since
\begin{equation*}
\mathrm{EE}_{\Delta}(t_i)
=
\sum_{\mathbf{j}}
\mathbb{Q}_{\mathcal{D}}
\!\left(
\mathbf{S}_{t_i}\in B_{\mathbf{j}}
\right)
V^+_{i,\mathbf{j}},
\end{equation*}
the discrete CVA satisfies
\begin{equation}
\begin{aligned}
\mathrm{CVA}_{\Delta}
&=
(1-R_{\mathrm{CVA}})
\sum_{i=1}^{M}
D(0,t_i)\Delta q(t_i)\,\mathrm{EE}_{\Delta}(t_i) \\
&=
(1-R_{\mathrm{CVA}})
\sum_{i=1}^{M}
D(0,t_i)\Delta q(t_i)
\sum_{\mathbf{j}}
\mathbb{Q}_{\mathcal{D}}
\!\left(
\mathbf{S}_{t_i}\in B_{\mathbf{j}}
\right)
V^+_{i,\mathbf{j}} \\
&=
M(1-R_{\mathrm{CVA}})C_vC_pC_q\,a_{\mathrm{CVA}} .
\end{aligned}
\label{eq:cva_from_amplitude}
\end{equation}

\begin{figure}[htbp]
    \centering
    \includegraphics[width=1.0\textwidth]{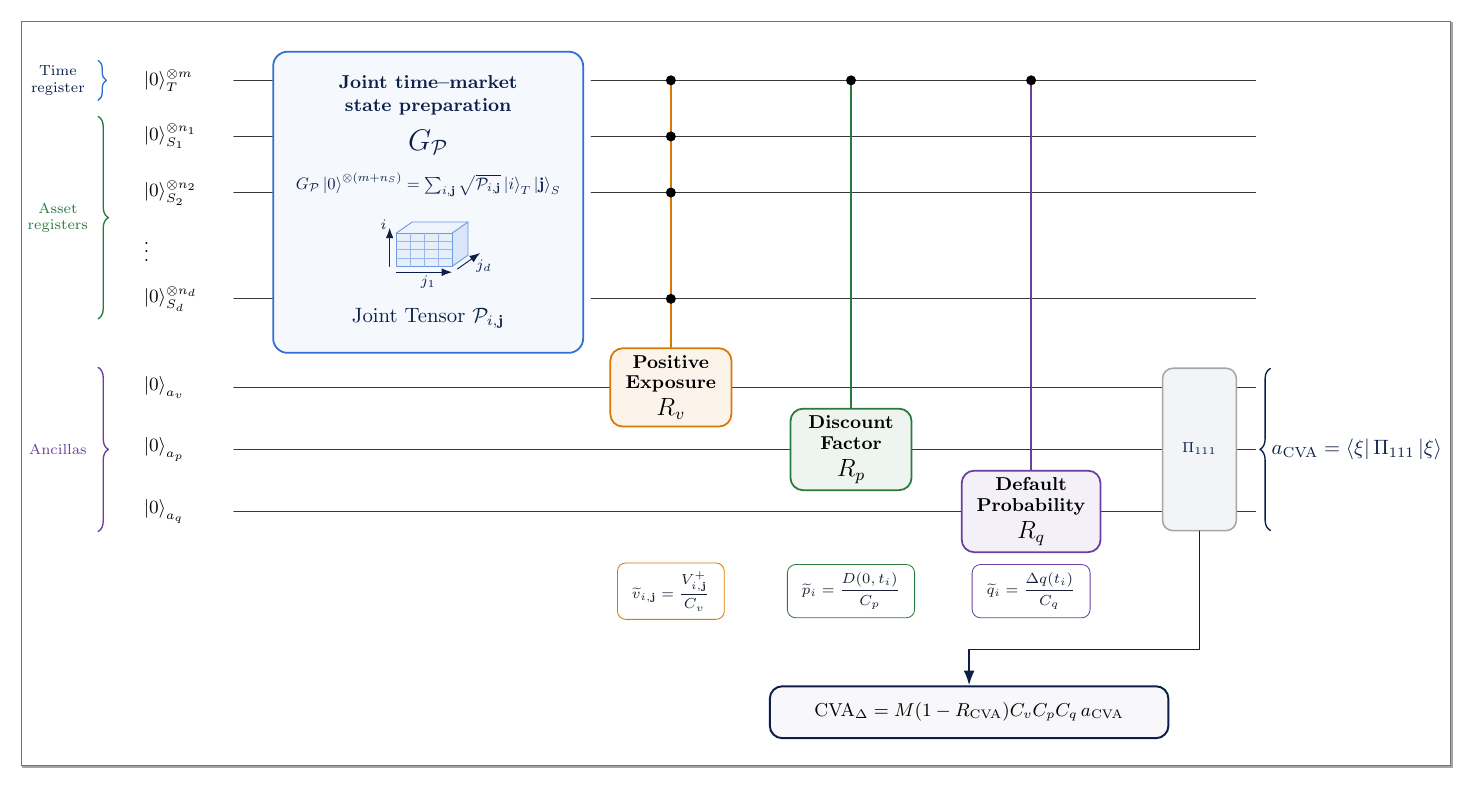}
    \caption{Quantum encoding of the discretised multi-asset CVA. The unitary
$G_{\mathcal P}$ prepares a coherent superposition over the full
time--market grid, while $R_v$, $R_p$, and $R_q$ encode the positive exposure,
discount factor, and default probability into three dedicated ancillas. The
projector $\Pi_{111}$ selects the event in which all three ancillas are measured
in state $\ket{1}$.}
    \label{fig:quantum_cva_encoding_multiasset}
\end{figure}

Consequently, the discretised CVA computation is no longer treated as a
classical expectation to be sampled directly. Once the circuit \(A\) has been
constructed, the target quantity is encoded in the probability of the marked
ancilla event of
\(\ket{\xi}=A\ket{0}^{\otimes(m+n+3)}\), and the financial CVA is
recovered by
\[
    \mathrm{CVA}_{\Delta}
    =
    M(1-R_{\mathrm{CVA}})C_vC_pC_q\,a_{\mathrm{CVA}} .
\]

Although written as the projector expectation
\(\langle \xi|\Pi_{111}|\xi\rangle\), this quantity is equivalently
\[
    a_{\mathrm{CVA}}
    =
    \langle \xi|\Pi_{111}|\xi\rangle
    =
    \|\Pi_{111}\ket{\xi}\|^2 .
\]
Thus \(a_{\mathrm{CVA}}\) is the total probability mass of the marked
subspace. The label \(111\) does not identify a single basis state of the full
\(m+n+3\)-qubit register; it fixes only the three ancilla qubits,
so the marked subspace is spanned by all
\(|i\rangle|\mathbf{j}\rangle|111\rangle\). Equivalently,
\[
    \ket{\xi}
    =
    \sqrt{1-a_{\mathrm{CVA}}}\,\ket{\xi_{\mathrm{bad}}}
    +
    \sqrt{a_{\mathrm{CVA}}}\,\ket{\xi_{\mathrm{good}}},
    \qquad
    \ket{\xi_{\mathrm{good}}}
    =
    \frac{\Pi_{111}\ket{\xi}}{\sqrt{a_{\mathrm{CVA}}}} .
\]
In this sense, estimating the projector expectation is equivalent to
estimating the squared amplitude of the good component, which is the quantity
accessed by quantum amplitude estimation.

\subsubsection{Contrast-aware Bayesian iterative quantum amplitude estimation}
\label{subsec:cabiqae}

Once the CVA functional is encoded into the marked probability
$a_{\mathrm{CVA}}$, the remaining task is to estimate this amplitude from
circuit measurements. As reviewed in Appendix~\ref{appsection:qae}, QAE reduces
the ideal sampling scaling from $\mathcal{O}(\varepsilon^{-2})$ to
$\mathcal{O}(\varepsilon^{-1})$ queries. The central algorithmic contribution of
this work is contrast-aware Bayesian iterative quantum amplitude estimation
(CABIQAE), an estimator developed in this work to make this quantum advantage
operational in the hardware-limited regime. The key idea is to treat the
experimentally observed loss of Grover contrast not as an external imperfection,
but as part of the inference problem itself. CABIQAE therefore links hardware
calibration, Bayesian amplitude inference and circuit-depth selection, so that
amplified circuits are chosen and interpreted according to the statistical
information they can still provide as their depth increases.

Among QFT-free methods,\footnote{QFT-based QAE relies on quantum phase
estimation and therefore requires deep controlled-Grover circuits, which are
not suitable for NISQ devices because noise limits reliably executable circuit
depth \parencite{Preskill2018NISQ}.} we take the recently developed
Beta-BIQAE \parencite{Li2026BIQAE} as the ideal-regime reference. It preserves
the iterative structure of IQAE \parencite{Grinko2021IQAE}, replaces confidence
intervals with Bayesian credible intervals, and propagates posterior
information across Grover depths, achieving the lowest query complexity among
the QFT-free methods benchmarked by \textcite{Li2026BIQAE}. This makes
Beta-BIQAE the natural starting point for the present work: its conjugate
Bayesian update is lightweight and statistically efficient, while the
hardware-sensitive part of the method is concentrated in the ideal observation
model used for amplified circuits.

That observation model assumes that each application of the Grover iterate $Q$
coherently rotates the state within the good--bad subspace (see
Appendix~\ref{appsection:qae},
eq.~\eqref{eq:amplified_probability_appendix}), so that querying $Q^k A$ and
measuring the good subspace yields a success probability that oscillates as a
function of the Grover power $k$. On real hardware, this oscillation is damped
by decoherence and gate errors: beyond a device-dependent depth, additional
Grover applications add query cost and circuit depth without a proportional
gain in information. CABIQAE keeps the Beta-BIQAE Bayesian transport structure,
but replaces the ideal likelihood by a calibrated contrast-decay model and uses
the same calibration to bias the scheduler towards Grover depths that remain
informative.

\paragraph{Noise-aware observation model.}

The Grover structure reviewed in Appendix~\ref{appsection:qae} shows that,
after $k$ Grover applications, the success probability depends on $a$ only
through the latent angle
\begin{equation}
    \vartheta := \arcsin\sqrt{a} \in [0,\pi/2],
    \label{latent_theta_def}
\end{equation}
and equals $\sin^2(K\vartheta)$ with $K = 2k+1$
(eq.~\eqref{eq:amplified_probability_appendix}). This angular
reparameterisation is therefore inherited directly from the geometry of
amplitude amplification, not introduced artificially: the Grover iterate
rotates the prepared state by $2\vartheta$ per application, making $\vartheta$
the natural coordinate for inference.

Thus, in the ideal case the success probability at Grover power $k$ is
\begin{equation}
    q(\vartheta,k) = \sin^2(K\vartheta).
    \label{eq:cabiqae_ideal_probability}
\end{equation}
CABIQAE replaces this by the hardware observation model
\begin{equation}
    p_{\mathrm{obs}}(\vartheta,k)
    =
    b + c(k)\bigl(q(\vartheta,k) - b\bigr),
    \qquad
    c(k) = c_0\exp\!\left(-\frac{2k+1}{\tau_{\mathrm{c}}}\right),
    \label{eq:cabiqae_observation_model}
\end{equation}
where \(b\) is the asymptotic baseline probability towards which the signal
decays as the circuit depth grows, \(c(k)\) is the hardware contrast at Grover
power \(k\), and \(\tau_{\mathrm{c}}\) is the effective contrast-decay scale,
measured in amplification depth \(K=2k+1\). This model extends the effective
contrast model used in BAE by \textcite{Ramoa2025bayesianquantum}
in two ways. First, the baseline \(b\) is not fixed to \(1/2\), but is allowed
to represent the circuit- and observable-dependent noise floor approached by
deep amplified circuits. This is important beyond single-qubit or balanced
readout settings, since the asymptotic probability of a structured good state
need not be \(1/2\). Second, the prefactor \(c_0\) allows for an initial
contrast mismatch already at the shallowest amplified circuit, before the
exponential decay with Grover depth is applied. Together, these two extensions
are particularly relevant for the deeper and more structured circuits that
arise in financial applications. When \(b=1/2\), \(c_0=1\), and
\(\tau_{\mathrm c}\to\infty\), the ideal observation model is recovered; for
finite \(\tau_{\mathrm c}\), the oscillation amplitude shrinks exponentially
with depth, so that the marginal information gain of deeper circuits diminishes
and eventually vanishes.

CABIQAE is implemented as a noise-aware Beta-BIQAE loop. At each stage \(t\),
the algorithm keeps a Grover power \(k_t\), an identifiable interval
\(I_t\subset[0,\pi/2]\) for the latent angle \(\vartheta\), and a Beta prior
over the observed success probability
\(p_{\mathrm{obs}}(\vartheta,k_t)\). A batch of measurements of
\(Q^{k_t}A\) is used to update this prior through the
binomial likelihood induced by the hardware observation model. The resulting
posterior is then pulled back to latent-angle space and converted into a
credible interval for the target amplitude \(a=\sin^2\vartheta\).

If the amplitude interval has half-width below the prescribed tolerance, the
algorithm stops. Otherwise, CABIQAE selects a new Grover power by enforcing
three conditions: the next amplified angle must remain identifiable within a
single branch, the amplification depth must respect the calibrated contrast
scale, and the candidate depth should maximise a posterior-averaged
Fisher-information score per circuit query. When the Grover power changes, the
current latent posterior is transported through the new observation model and
approximated again by a Beta prior. In this way, CABIQAE uses the same Bayesian transport structure as Beta-BIQAE,
but makes the admissible Grover depths depend on the statistical sensitivity
left after contrast loss.

This mechanism determines how the empirical results should be interpreted. In
the high-contrast limit, the likelihood used by CABIQAE collapses to the ideal
amplified model, and the procedure behaves as a Beta-BIQAE-type iterative
estimator. In the hardware-calibrated regime studied here, by contrast, the
useful amplification range is jointly constrained by the target precision and by
the contrast length $\tau_{\mathrm c}$: once the observed probability is pulled
close to the baseline (b), additional Grover applications contribute little
statistical information despite increasing the query cost and circuit depth.
This interpretation is consistent with the benchmarking approach followed by
\textcite{Ramoa2025bayesianquantum} for noise-aware Bayesian amplitude
estimation, where the relevant scaling of BAE is assessed through observed
cost--accuracy curves rather than through a hardware-independent theorem for the
noisy adaptive Bayesian setting. In the same spirit, the complexity evidence for
CABIQAE in this work is empirical and implementation-dependent, based on the
observed error--query behaviour under calibrated contrast loss.

The comparisons reported in the results section are therefore cost--accuracy
comparisons within this finite-contrast regime. CABIQAE is expected to be most
useful precisely when some amplification remains informative, but blindly
increasing the Grover depth would lead to low-contrast measurements and poorly
calibrated posterior updates. This is the operational regime encountered in the
hardware-replay experiments below.

The complete pseudocode and the definitions of the auxiliary modules
\textsc{BayesianUpdate}, \textsc{ComputeCRI}, \textsc{PreparePrior} and
\textsc{FindNextK} are given in Appendix~\ref{app:cabiqae_modules}.
\section{Data and methodology}
\subsection{Market data}

All market inputs are obtained from \textit{LSEG Workspace}. The dataset includes
historical equity-index closes, option-implied volatility surfaces, EUR OIS (overnight indexed swap) discount factors, CDS par spreads and deterministic dividend yields. These inputs
calibrate the equity dynamics, discount curve and default-probability curve used
in the CVA calculation. The full data windows, tenors, interpolation inputs and
retained numerical values are reported in
Appendix~\ref{appsec:classical_market_inputs}.

\subsection{Methodology}
\label{subsec:methodology}
\subsubsection{Classical benchmark}
\label{subsec:classical_benchmark}
Given simulated paths
\(\{\mathbf S^{(n)}_{t_i}\}\),
generated with the exact lognormal transition in~\eqref{eq:exact_step}, each trade in
the netting set is repriced at every exposure date using the Black--Scholes--Merton
formulas in Section~\ref{sec:classical_valuation}. Netting is applied before taking the
positive part of the portfolio value. The piecewise-constant hazard-rate curve is obtained via sequential bootstrapping of
\eqref{eq:cds_par_spread_compact} across quoted CDS maturities.
The continuous-underlying Monte Carlo benchmark is obtained by evaluating
\eqref{eq:cva_mc_compact} at \(t=0\), with standard errors computed from the
corresponding pathwise discounted loss variables. This benchmark, denoted by
\(\widehat{\mathrm{CVA}}^{\mathrm{cont}}_{\mathrm{MC}}\), is used as the main
classical reference.

The same paths are then projected onto the finite time--market grid introduced in
Section~\ref{subsec:quantum_cva_encoding}. For each monitoring date \(t_i\), the
empirical conditional probabilities \(\mathbb{Q}_\mathcal{D}(\mathbf{S}_{t_i}\in B_{\mathbf{j}})\) are obtained by counting
the simulated states falling in each joint price bin \(B_{\mathbf j}\), after truncation and
renormalisation of the price domain. Together with uniform time weights \(w_i=1/M\),
these probabilities define the joint probability tensor as in \eqref{eq:joint_probability_tensor}.

On the same grid, the implementation evaluates the positive-exposure table
\(V^+_{i,\mathbf{j}}\), the discount vector $D(0,t_i)$, the default-increment vector
\(\Delta q(t_i)\), and the scaling constants \(C_v,C_p,C_q\). These objects are the
classical numerical realisation of the quantities encoded by the quantum circuit.

Finally, the finite time-market grid CVA in \eqref{eq:cva_discrete_grid} is
computed from the same discretised inputs later encoded by the quantum circuit and
stored as the direct classical reference for the quantum pipeline. We denote this value by
\(\widehat{\mathrm{CVA}}_{\Delta}^{\mathrm{tab}}\). It is the appropriate benchmark
for evaluating the trained quantum encoding, because it uses exactly the same
probability tensor, exposure table, discount vector and default-increment vector as the
quantum circuit.

By contrast,
\(\widehat{\mathrm{CVA}}^{\mathrm{cont}}_{\mathrm{MC}}\) is used to monitor the
classical approximation error introduced by time monitoring, truncation and finite price
binning. A grid-convergence diagnostic is also performed by increasing the number of asset-grid
bits and comparing the resulting sequence of finite-grid CVA values with the continuous
Monte Carlo benchmark.\footnote{The number of time and asset-register qubits changes
across the hardware-compatible instances. For this reason, the main text describes the
procedure generically.}

The detailed market inputs, numerical parameters, grid specifications and additional computational controls are reported in Appendix~\ref{appsec:classical_market_inputs}.

\subsubsection{Noise-aware amplitude estimation algorithm validation}
\label{subsec:cabiqae_test}

Before using CABIQAE in the final CVA calculation, we first isolate the
amplitude-estimation layer in a reduced validation experiment. The purpose is to
compare the estimators in the regime relevant for the final CVA experiment:
finite-shot inference from amplified circuits whose useful information may
decrease with Grover depth because of hardware-induced contrast loss. The
comparison therefore tracks not only final estimation error and query count, but
also the amplification depths selected by each method and the associated
classical post-processing time.

The study is performed on the reduced three-qubit unitary \(A^{\mathrm{test}}\)
shown in Figure~\ref{fig:three_qubit_parametrised_circuit}. The circuit
preserves the state-preparation, entangling and objective-rotation structure of
the final CVA construction, while remaining shallow enough for noisy simulation
and hardware execution. All gate parameters are fixed arbitrarily except for the
final objective-register rotation \(\varphi\). In the ideal validation regime,
different choices of \(\varphi\) generate several target amplitudes; each
independent trajectory uses one fixed choice. For that trajectory, the success
subspace is defined by the third qubit being in state \(\ket{1}\). The
corresponding true target amplitude is therefore
\begin{equation}
    a_{\mathrm{true}}^{(r)}
    =
    \bra{0}
    \bigl(A^{\mathrm{test}}(\varphi^{(r)})\bigr)^{\dagger}
    \Pi_{\mathrm{2}}
    A^{\mathrm{test}}(\varphi^{(r)})
    \ket{0},
    \qquad
    \Pi_{\mathrm{2}}
    =
    I_{q_0}\otimes I_{q_1}\otimes \ket{1}\!\bra{1}_{q_2}.
    \label{eq:cabiqae_test_target_amplitude}
\end{equation}

The adaptive
estimation algorithm is then executed on this instance, generating the
corresponding trajectory. Different Grover powers within the same trajectory
generate different amplified probabilities \(p_k\), not new target amplitudes. The
reference amplitude \(a_{\mathrm{true}}^{(r)}\) is obtained by
noiseless statevector simulation.\footnote{By noiseless statevector simulation
we mean exact classical simulation of the pure quantum state evolved by the
circuit, without finite-shot sampling noise, gate noise, decoherence or
readout errors. This simulation model is implemented using Qiskit's
Statevector and Aer statevector simulators \parencite{QiskitStatevector,QiskitAerNoiseModels}.}

CABIQAE is compared against two QFT-free amplitude estimation baselines:
BIQAE \parencite{Li2026BIQAE} and BAE \parencite{Ramoa2025bayesianquantum}. BIQAE is
selected as the strongest ideal-regime benchmark considered here, due to its
state-of-the-art proven query complexity; however, since it does not model
hardware-induced contrast loss, its performance is expected to deteriorate under
noisy amplified circuits. Empirically demonstrating this limitation, and
testing whether a contrast-aware correction can overcome it, is one of the
novel contributions of this study. BAE is selected as the natural noise-aware
competitor: it is Bayesian, adaptive, competitive in the ideal regime, and
designed to remain effective under noisy executions. Together, these baselines
provide a stringent test of whether CABIQAE can retain near-ideal QAE efficiency
while improving robustness in the hardware-relevant regime.

The complexity measure is the number of circuit queries to
\(A^{\mathrm{test}}\), denoted by \(N_q\). Because the algorithms are adaptive, one independent run generates a
trajectory of intermediate estimates indexed by the accumulated query cost.
Within such a trajectory, a stage denotes a block of shots collected at a fixed
Grover power. Therefore, if stage \(s\) uses \(n_s\) shots at power \(k_s\), the
accumulated query cost after \(S\) stages is
\begin{equation}
    N_q(S)
    =
    \sum_{s=1}^{S} n_s(2k_s+1),
    \label{eq:cabiqae_testing_query_cost}
\end{equation}
The estimate recorded after stage \(S\) is therefore associated with query cost
\(N_q(S)\). The final cost of the run is
\(N_q=N_q(S_{\mathrm{final}})\).
We also record classical wall-clock runtime, since Bayesian methods
incur non-negligible posterior-update, resampling, interval-inversion,
and scheduler-optimisation costs.

For each recorded trajectory point \((r,s)\), the estimator
\(\widehat a^{(r)}_{s}\)\footnote{For IQAE, BIQAE and CABIQAE, this is the
midpoint of the current amplitude interval; for BAE, it is the posterior
mean produced by the sequential Monte Carlo update \parencite{Ramoa2025bayesianquantum}.
Thus, each method is evaluated using its native point estimator.}
is the point estimate stored at query cost \(N^{(r)}_{q}(s)\), with relative error
\(\epsilon^{(r)}_{rel,s}=|\widehat a^{(r)}_s-a_{\mathrm{true}}^{(r)}|/a_{\mathrm{true}}^{(r)}\). Since adaptive
methods update and terminate at different query costs, trace points are
grouped separately for each algorithm into logarithmic \(N_q\)-bins. For a
bin \(B_j\), the plotted point is
\[
    \left(
    \operatorname{median}_{(r,s)\in B_j} N_{q}^{(r)}(s),
    \operatorname{median}_{(r,s)\in B_j}
    \epsilon^{(r)}_{rel,s}
    \right).
\]
\noindent
We use median errors as robust trajectory summaries; by Appendix~\ref{appsection:qae}, they preserve the usual ideal exponents. The uncertainty bands for each point are obtained by bootstrap resampling of this
median.\footnote{Bootstrap resampling is used because each bin typically
contains a limited number of independent runs, making dispersion summaries such
as the standard deviation a poor indicator of uncertainty for a median-based
curve. The bootstrap directly estimates the sampling variability of the plotted
median under the available finite ensemble.}

We also plot the reference curves
\begin{equation}
    g_{\mathrm{QAE}}(N_q)=y_0\frac{N_0}{N_q},
    \qquad
    g_{\mathrm{MC}}(N_q)=y_0\sqrt{\frac{N_0}{N_q}},
    \label{eq:cabiqae_testing_reference_lines}
\end{equation}
as visual guides for Heisenberg\footnote{The Heisenberg-limited regime in the context of quantum amplitude estimation is briefly reviewed in Appendix~\ref{appsection:qae}.} and Monte Carlo scaling, respectively. Here \(N_0\) is the
smallest plotted representative query cost, and \(y_0\) is the value at
\(N_0\) of the log--log fit through the plotted binned points, following the
intercept convention of \textcite[Appendix~G]{Ramoa2025bayesianquantum}.

The main-text comparison considers two execution regimes. In the ideal
regime, the adaptive algorithms sample each requested Grover power \(k\)
from the Bernoulli model in \eqref{eq:cabiqae_ideal_probability}, so that the
only randomness is finite-shot binomial noise. In the hardware regime, the IBM
backend is used once to measure the experimentally realised response of the
transpiled circuits. After readout calibration, which is explained below, this produces mitigated
probabilities \(\widehat p_{\mathrm{hw}}(1\mid k)\) for the admissible Grover
powers. These probabilities define a hardware-replay model: whenever BAE,
BIQAE or CABIQAE requests shots at \(k\), counts are generated as
\[
    n_1(k)\sim
    \operatorname{Binomial}
    \left(N_{\mathrm{shots}},\widehat p_{\mathrm{hw}}(1\mid k)\right).
\]
The hardware-replay methodology therefore separates the expensive hardware characterisation
step from the statistical evaluation of the estimators. It enables many
independent algorithmic trajectories to be generated from the same measured
hardware response, without executing every adaptive trajectory live on the QPU.

The number of target amplitudes differs between regimes. In the ideal regime,
the validation trajectories use several amplitudes generated by different
values of \(\varphi\). In the hardware-replay regime reported here, the
algorithms are run against a single hardware-calibrated amplitude. This is a
scope choice of the hardware experiment, not a property of the replay mechanism:
changing the hardware amplitude would require executing and calibrating a new
amplified-circuit family for the new \(\varphi\), which would make the required
QPU time prohibitive.

\begin{figure}[htbp]
    \centering
    \includegraphics[width=0.7\textwidth]{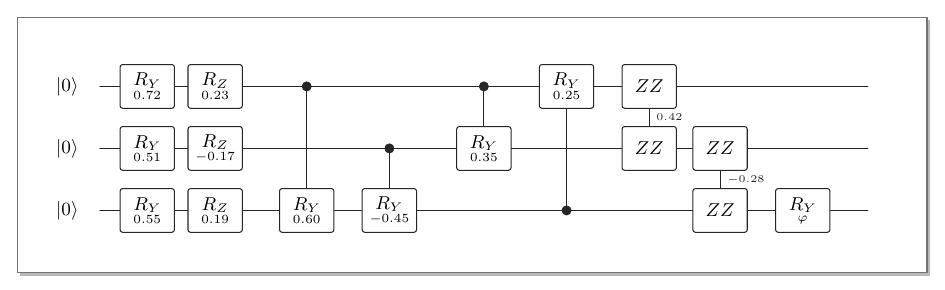}
    \caption{Three-qubit circuit $A^{\mathrm{test}}$ used to validate CABIQAE. The final rotation angle $\varphi$ is left free to generate different target amplitudes.}
    \label{fig:three_qubit_parametrised_circuit}
\end{figure}

The contrast scale \(\tau_{\mathrm c}\) is calibrated
empirically\footnote{This is a supervised, circuit-specific calibration: the ideal
probabilities are used only to fit the effective noise parameters, not as
inputs to the estimator. Although one could approximate them from IBM backend
calibration data, this would require additional modelling assumptions.
Auxiliary calibration circuits could also introduce transfer bias because
their depth, layout, and gate structure would differ from those of the target
circuit. We therefore calibrate the target circuit family directly, as a
pragmatic oracle-assisted characterisation rather than a fully blind
end-to-end protocol.}, for the single hardware-calibrated amplitude. With
\(\vartheta_0=\arcsin\sqrt{a_{\mathrm{true}}}\) as in
\eqref{latent_theta_def}, the calibration varies the Grover power \(k\).
Writing \(q_k=q(\vartheta_0,k)\), \(K=2k+1\), and denoting by
\(\widehat p_k\) the observed probability $p_{\mathrm{obs}}( \vartheta_0,k)$, the contrast sample is
\[
    \widehat c_k
    =
    \frac{\widehat p_k-b}{q_k-b}.
\]
Calibration points with unstable denominators or non-physical contrasts are
discarded, and the remaining values are fitted to the exponential contrast
model through
\[
    \log \widehat c_k
    =
    \beta_0+\beta_1 K+\varepsilon_k,
    \qquad
    \widehat\tau_{\mathrm c}=-1/\widehat\beta_1
    \quad(\widehat\beta_1<0).
\]
If no valid negative-slope fit is obtained, the calibration is rejected and the
ideal observation model is used for that hardware instance. Otherwise,
\(\widehat\tau_{\mathrm c}\) is used both in the likelihood correction
\eqref{eq:cabiqae_observation_model} and in the noise-aware selection of Grover
depths.

Readout mitigation is estimated from objective-register calibration circuits.
With
\[
    r_0=\widehat p(1\mid0),\qquad
    r_1=\widehat p(1\mid1),\qquad
    \Delta_{\mathrm{read}}=r_1-r_0,
\]
the raw probability is corrected, when the readout calibration is stable, as
\[
    \widehat p^{\,\mathrm{mit}}_k
    =
    \Pi_{[0,1]}
    \left(
        \frac{\widehat p^{\mathrm{\,raw}}_k-r_0}{\Delta_{\mathrm{read}}}
    \right).
\]
The binomial standard error of the raw estimate,
\[
    \widehat\sigma^{\mathrm{\,raw}}_{p,k}
    =
    \sqrt{
    \frac{
    \widehat p^{\mathrm{\,raw}}_k
    \bigl(1-\widehat p^{\mathrm{\,raw}}_k\bigr)}
    {n_k}
    },
\]
is propagated through the affine readout correction and then through the
contrast estimator:
\[
    \widehat\sigma^{\mathrm{\,mit}}_{p,k}
    \simeq
    \frac{\widehat\sigma^{\mathrm{\,raw}}_{p,k}}{|\Delta_{\mathrm{read}}|},
    \qquad
    \widehat\sigma_{c,k}
    =
    \frac{\widehat\sigma^{\mathrm{\,mit}}_{p,k}}{|q_k-b|}.
\]
Hardware calibration points are retained only when the denominator is stable,
the contrast is positive and unsaturated, and
\(\widehat c_k\) is statistically resolved from zero relative to
\(\widehat\sigma_{c,k}\).

\subsubsection{Quantum pipeline}
\label{subsec:quantum_pipeline}
\paragraph{Quantum encoding and circuit training.}
In the numerical implementation considered in this work, the state register
contains two time qubits and two qubits per underlying for a two-underlying
exposure grid. Thus, \(m=2\) and \(n=4\), giving a total of
\(m+n=6\) qubits. Although the overall pipeline is not restricted to six qubits, the variational
ansatz circuits considered in this implementation are designed for this
six-qubit register. Their architectures are selected to satisfy the connectivity constraints of the target hardware while retaining sufficient expressive capacity for the required distributional encoding, with a more detailed discussion of the hardware-aware design choices provided in Appendix~\ref{app:hardware}. Extending
the same ansatz families to finer discretisations or higher-dimensional state
registers would generally require redesigning the circuit architectures, and
would not necessarily preserve the same hardware-friendly structure.\footnote{
In this context, an ansatz denotes a parametrised family of quantum circuits. Training amounts to optimising its free rotation angles so that the circuit
prepares a state whose measurement distribution approximates the target
distribution. Further details on the hardware topology, connectivity
constraints and backend selection are provided in Appendix~\ref{app:hardware}.}

Once the cardinalities of the quantum registers have been fixed, the
classically constructed finite-grid objects
\(\mathcal{P}_{i,\mathbf{j}}, \tilde{v}_{i,\mathbf{j}}, \tilde{p}_i\) and
\(\tilde{q}_i\) are taken as inputs to the quantum pipeline. Writing \(x=(i,\mathbf{j})\in\mathcal X\) for a computational basis state of
this register, the target distribution to be encoded is
\[
    P_{\mathrm{target}}(x) := \mathcal{P}_{i,\mathbf{j}},
    \qquad
    \sum_{x\in\mathcal X}P_{\mathrm{target}}(x)=1 .
\]

The state-preparation unitary \(G_{\mathcal P}\) is implemented approximately
using a Quantum Circuit Born Machine (QCBM) \parencite{Liu2018QCBM}. More precisely,
we use  a parametrised circuit
\(G_{\theta}\), whose architecture is shown in
Figure~\ref{fig:qcbm_heavyhex6_two_layer_pattern}. The probability distribution
generated by measuring this circuit in the computational basis is its Born
distribution,
\[
    P_{\boldsymbol{\theta}}(x)
    :=
    \left|
        \langle x|G_{\theta}|0^{\otimes(m+n)}\rangle
    \right|^2 .
\]
Training the QCBM consists of optimising the parameters
\(\theta\) so that \(P_{\theta}\) approximates the
target finite-grid distribution \(\mathcal P\). Following \textcite{Alcazar2021CVA}, the QCBM parameters are trained by minimising
the regularised cross-entropy between the target distribution and the
distribution generated by the circuit,
\[
    \mathcal L_{\mathrm{QCBM}}(\theta)
    =
    -\sum_{x\in\mathcal X}
        P_{\mathrm{target}}(x)
        \log\!\left(P_\theta(x)+\epsilon_{\mathrm{num}}\right),
\]
where \(\epsilon_{\mathrm{num}}>0\) is a small numerical constant introduced to
avoid logarithmic instabilities when a generated probability is close to zero. In the
statevector setting, \(P_\theta\) is obtained exactly from the simulated quantum
state. In finite-shot or noisy executions, \(P_\theta\) is replaced by the
empirical frequency distribution
\(\widehat P_{\theta}^{(N_{\mathrm{shots}})}\) induced by circuit
measurements, with zero-frequency outcomes handled through the same additive
floor \(\epsilon_{\mathrm{num}}\).

\begin{figure}[htbp]
    \centering
    \includegraphics[width=0.6\textwidth]{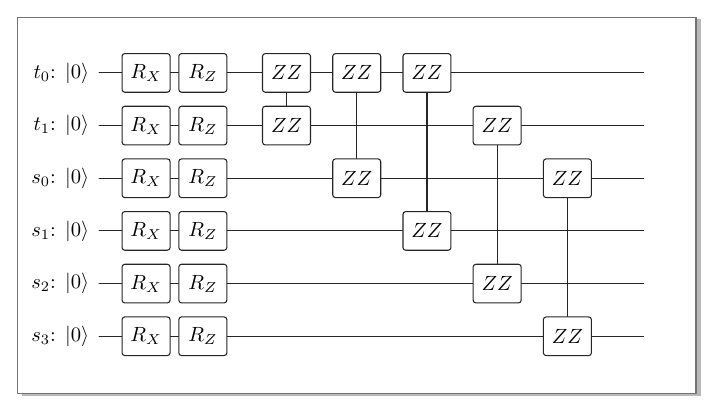}
    \caption{{
Hardware-compatible QCBM ansatz used for the six-qubit state-preparation block. Each variational block combines local single-qubit rotations 
with parametrised \(R_{ZZ}\) entanglers placed along a connectivity pattern chosen 
to respect the target heavy-hex hardware constraints. 
}}
    \label{fig:qcbm_heavyhex6_two_layer_pattern}
\end{figure}

Although the optimisation is carried out through the above loss, the quality of
the trained QCBM is assessed in terms of the numerically regularised
Kullback--Leibler diagnostic from the target distribution to the generated
distribution,
\[
    \operatorname{KL}_{\epsilon}\!\left(P_{\mathrm{target}}\|P_\theta\right)
    =
    \sum_{x:P_{\mathrm{target}}(x)>0}
    P_{\mathrm{target}}(x)
    \log
    \left(
        \frac{P_{\mathrm{target}}(x)}
             {P_\theta(x)+\epsilon_{\mathrm{num}}}
    \right).
\]
This quantity measures the distributional discrepancy between the finite-grid
classical benchmark and the probability law prepared by the quantum circuit,
with the numerical floor preventing singular contributions when the generated
probability is close to zero. 

The remaining quantum subroutines correspond to the controlled rotations
\(R_v\), \(R_p\), and \(R_q\), which encode the finite-grid quantities
\(\widetilde v_{i,\mathbf{j}}\), \(\widetilde p_i\), and \(\widetilde q_i\),
respectively, as described in Section~\ref{subsec:quantum_cva_encoding}.
Following the quantum CVA construction of \textcite{Alcazar2021CVA},
these rotations are implemented through Controlled Rotations Circuit Ansatz
(CRCA) blocks.

\begin{figure}[htbp]
    \centering
    \includegraphics[width=0.6\textwidth]{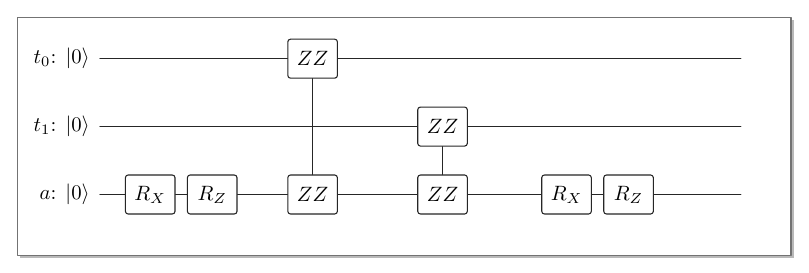}
    \caption{Native-tree CRCA ansatz used for the time-dependent controlled rotations $R_p$ and $R_q$.
    }
\label{fig:crca_default_probabilities_native_tree_ansatz}
\end{figure}

The purpose of each CRCA block is to approximate the deterministic map that
assigns to each computational basis state \(x\) a rescaled quantity
\(f_{\mathrm{target}}(x)\in[0,1]\). For a generic target function
\(f_{\mathrm{target}}:\mathcal X_f\to[0,1]\), the trainable circuit induces an approximation
\(F_\phi(x)\) to \(f_{\mathrm{target}}(x)\), where \(F_\phi(x)\) denotes the probability of
measuring the corresponding ancilla in state \(|1\rangle\) conditional on the
input \(x\). Equivalently, the controlled rotation implements
\[
    |0\rangle
    \longmapsto
    \sqrt{1-F_\phi(x)}\,|0\rangle
    +
    \sqrt{F_\phi(x)}\,|1\rangle .
\]
The parameters \(\phi\) are obtained by minimising the mean-square discrepancy
over the relevant finite grid,
\[
    \mathcal L_f(\phi)
    =
    \frac{1}{|\mathcal X_f|}
    \sum_{x\in\mathcal X_f}
    \left(F_\phi(x)-f_{\mathrm{target}}(x)\right)^2 .
\]

The discount-factor and default-probability rotations \(R_p\) and \(R_q\), implemented through the parametrised circuits \(R_p(\phi_p)\) and \(R_q(\phi_q)\), respectively, use the same native-tree CRCA topology shown in Figure~\ref{fig:crca_default_probabilities_native_tree_ansatz}. They differ only in the ancilla qubits on which they act and in the target functions used during training,
\[
    f_p(i)=\widetilde p_i,
    \qquad
    f_q(i)=\widetilde q_i .
\]
Both rotations condition exclusively on the time register, since the discount
factor and the default-probability increment are deterministic functions of the
monitoring date.

By contrast, the exposure rotation \(R_v\), implemented as $R_v(\phi_v)$, must condition on the full
time-market register, with target function
\[
    f_v(i,\mathbf{j})=\widetilde v_{i,\mathbf{j}} .
\]
This makes \(R_v(\phi_v)\) the most demanding controlled-rotation block, because the
positive exposure is generally non-linear and non-separable across the time and
market coordinates. The ansatz used for this exposure encoder is represented in
Figure~\ref{fig:crca_heavy_hex_star_two_subcircuits}.

\begin{figure}[htbp]
    \centering
    \includegraphics[width=1.0\textwidth]{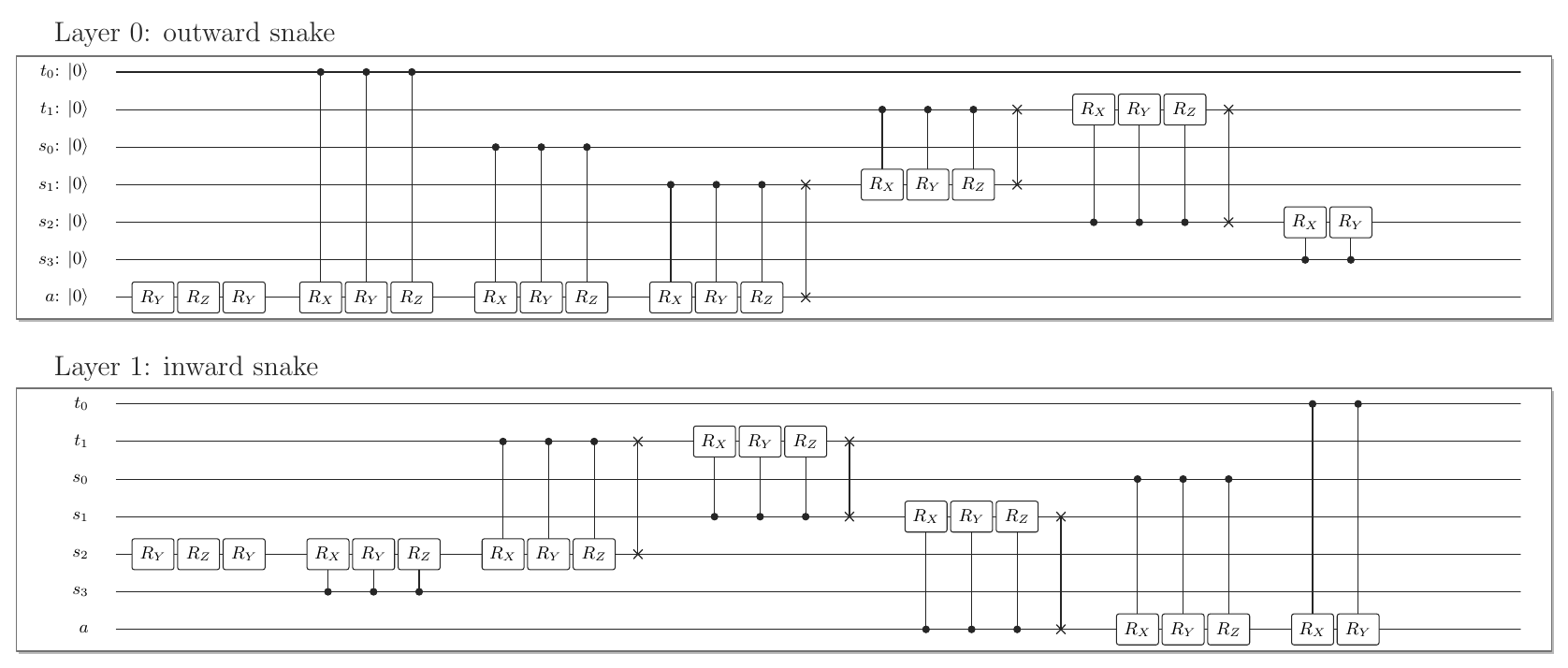}
    \caption{CRCA used for the exposure rotation $R_v$. The displayed outward and inward
snake layers route information between the computational register and the exposure ancilla while preserving a
hardware-aware entangling structure.
    }
\label{fig:crca_heavy_hex_star_two_subcircuits}
\end{figure}

All circuit diagrams shown in this section for the QCBM and CRCA blocks
represent a single variational layer, namely $L=1$, of the corresponding ansatz. The full
trainable circuits used in the experiments are obtained by composing several
repetitions of these layers, with independent parameters in each repetition.

\paragraph{CVA reconstruction and amplitude estimation stage.} \label{par:noiseless_cva}
After training, the four circuit blocks are assembled into the full CVA
encoding unitary
\[
    A_{\Theta}
    =
    R_q(\phi_q^\star)\,
    R_p(\phi_p^\star)\,
    R_v(\phi_v^\star)\,
    G_{\theta^\star},
\]
where \(\Theta=(\theta^\star,\phi_v^\star,\phi_p^\star,\phi_q^\star)\) is the array containing the trained parameters.

Defining
\[
    \widehat{\mathcal P}_{i,\mathbf{j}}
    :=
    P_{\theta^\star}(i,\mathbf{j}),
    \qquad
    \widehat v_{i,\mathbf{j}}
    :=
    F_{\phi_v^\star}(i,\mathbf{j}),
    \qquad
    \widehat p_i
    :=
    F_{\phi_p^\star}(i),
    \qquad
    \widehat q_i
    :=
    F_{\phi_q^\star}(i),
\]
in the noiseless statevector setting, the success probability of the trained
CVA circuit is
\[
\begin{aligned}
    \widehat a_{\mathrm{SV}}(\Theta)
    &=
    \langle 0^{\otimes(m+n+3)}|
    A_{\Theta}^{\dagger}
    \Pi_{111}
    A_{\Theta}
    |0^{\otimes(m+n+3)}\rangle                                    \\
    &=
    \sum_{i,\mathbf{j}}
        \widehat{\mathcal P}_{i,\mathbf{j}}\,
        \widehat v_{i,\mathbf{j}}\,
        \widehat p_i\,
        \widehat q_i .
\end{aligned}
\]

The corresponding encoded CVA is recovered through
\[
    \widehat{\mathrm{CVA}}_{\mathrm{SV}}(\Theta)
    =
    M(1-R_{\mathrm{CVA}})C_vC_pC_q\,
    \widehat a_{\mathrm{SV}}(\Theta).
\]

The amplitude estimation experiment uses the same trained CVA encoding unitary
\(A_{\Theta}\), with no further modification of the encoder. The marked
subsystem is the three-ancilla register, and the good event is the bitstring
\(111\).  An amplitude estimation routine returns
\(\widehat{a}_{\mathrm{AE}}(\Theta)\), which is converted back to monetary CVA
units as
\[
    \widehat{\mathrm{CVA}}_{\mathrm{AE}}(\Theta)
    =
    M(1-R_{\mathrm{CVA}})C_vC_pC_q\,
    \widehat{a}_{\mathrm{AE}}(\Theta).
\]

CABIQAE and BIQAE are then compared with direct circuit sampling (DCS) on this common trained
amplitude, using the same query-cost metric \(N_q\). The latter baseline
corresponds to repeated unamplified sampling of the circuit success bit, not to
a classical Monte Carlo simulation of the financial model. In the ideal regime, measurement outcomes
are sampled from the exact Grover amplification law \eqref{eq:cabiqae_ideal_probability}
equivalently from the exact statevector probability. In the hardware-replay
regime, amplified circuits are first executed on the selected backend for a
fixed grid of Grover powers. The resulting empirical success probabilities are
then used to replay the adaptive trajectories of CABIQAE and BIQAE under
identical hardware response data, following the calibration methodology of
Section~\ref{subsec:cabiqae_test}. This avoids running every adaptive trajectory
directly on quantum hardware, which would require a prohibitive amount of
quantum compute time, while still preserving the experimentally observed
depth-dependent response of the device. For the hardware runs, we use the IBM Quantum backends described in
Appendix~\ref{app:hardware}. The CVA amplitude estimation experiments further
employ Q-CTRL's Performance Management Qiskit Function to reduce the
effect of hardware errors \parencite{QCTRLPerformanceManagement2026}.
Implementation details are provided in the project
repository \parencite{AuthorRepo2026}.

\section{Experimental results}
\subsection{Proposed algorithm validation results}
\label{subsec:noise_aware_ae_results}

All validation results use \(R=300\) independent trajectories, as described in
Section~\ref{subsec:cabiqae_test}. In the ideal regime these trajectories cover
several target amplitudes, whereas the hardware-replay regime uses a single
hardware-calibrated target amplitude. These numerical controls and the
algorithmic hyperparameters are reported in
Table~\ref{tab:ae_hyperparameters} of
Appendix~\ref{appsec:classical_market_inputs}. Table~\ref{tab:reduced_instance_summary_compact}
and Figure~\ref{fig:ae_reduced_main_results} summarise the reduced instance
results in the ideal and hardware-replay regimes, with resource requirements
reported in Appendix~\ref{appsec:amplified_resource_tables}, Table \ref{tab:atest_amplified_logical_transpiled_resources}. The direct circuit sampling (DCS) curve corresponds to unamplified sampling of
the original circuit, and is included as the circuit-level analogue of a Monte
Carlo probability baseline. In noisy regimes, unlike a purely classical Monte
Carlo estimator, it also includes hardware-induced effects.

In the noiseless regime, Figure~\ref{fig:ae_reduced_main_results} (top) shows the expected separation between amplified amplitude estimation and direct sampling. BAE, BIQAE and CABIQAE follow the \(O(1/N_q)\) reference trend at large query budgets, while DCS follows the slower \(O(1/\sqrt{N_q})\) Monte Carlo behaviour. The three amplified routines remain close throughout most of the trajectory: CABIQAE attains the lowest median final error, \(8.02\times10^{-5}\), BIQAE uses the smallest median query budget among the amplified routines, and all three remain well below the DCS error level. The log--log fits in Table~\ref{tab:empirical_error_scaling_fit_metrics} confirm this separation, with fitted exponents close to the ideal amplified rate for BAE, BIQAE and CABIQAE,
\[
    \beta=-0.96\pm0.02,\qquad
    \beta=-0.93\pm0.04,\qquad
    \beta=-0.92\pm0.04,
\]
respectively, whereas DCS gives \(\beta=-0.51\pm0.01\), consistent with standard Monte Carlo sampling.

The hardware-replay regime in Figure~\ref{fig:ae_reduced_main_results} (bottom) is qualitatively different. The noise-naive BIQAE curve stops improving once amplified circuits enter the high-\(K\) noisy regime and settles near a visible noise floor, in line with the warning by \textcite[p.~15]{Li2026BIQAE} that large Grover depths on near-term devices may prevent the predicted gains from being realised. This behaviour is also consistent with the calibrated contrast scale in Table~\ref{tab:combined_calibration_summary}: the exponential contrast model predicts substantial signal degradation when the amplification factor reaches \(K\simeq\tau_c\simeq34\). DCS shows a related degradation, although it remains below BIQAE over the reported range.

By contrast, BAE and CABIQAE continue reducing the error over a broader query interval, though with degraded effective convergence. Their fitted exponents soften from near-\(N_q^{-1}\) behaviour to
\[
    \beta=-0.68\pm0.03,
    \qquad
    \beta=-0.69\pm0.03,
\]
respectively, while still improving faster than hardware-replay DCS, \(\beta=-0.30\pm0.05\). This average hides two regimes: CABIQAE, more clearly than BAE, preserves an amplified-improvement window from roughly \(N_q\simeq10^3\) to just below \(10^4\), where additional queries still exploit informative moderate-depth Grover circuits. Beyond this window, the curves flatten or deteriorate as queries increasingly come from low-contrast high-\(K\) circuits, and the advantage over Monte Carlo is lost. BIQAE, being noise-naive, collapses to \(\beta=-0.25\pm0.04\) because it continues treating deep amplified circuits as informative.

\begin{table}[H]
\centering
\normalsize
\setlength{\tabcolsep}{7pt} 
\renewcommand{\arraystretch}{1.3} 

\caption{
Reduced test-instance summary for the ideal and hardware-replay regimes.
All errors are reported as the median relative error of the final estimates across trajectories. The quantities
$N_q^{50}$, $t_{\mathrm{cl}}^{50}$, and $K_{50}$ denote, respectively,
the median final $A_{\mathrm{test}}$-query cost, the median classical wall-clock runtime,
and the median maximum Grover amplification factor reached during an
independent run. The coverage column (Cov.) reports the empirical fraction
of runs for which the true amplitude lies inside the final interval returned
by the algorithm. Classical runtimes refer only to classical post-processing
and exclude quantum execution and queueing times.
}
\label{tab:reduced_instance_summary_compact}

\begin{tabular}{@{} l c c c c c @{}}
\toprule[1.2pt]
\multicolumn{6}{c}{\textbf{Noiseless}} \\
\midrule[0.8pt]
Algorithm & $\epsilon^{50}$ & $N_q^{50}$ & $t_{\mathrm{cl}}^{50}$ & Cov. & $K_{50}$ \\
\midrule
BAE     & $9.40 \times 10^{-5}$          & $6.88 \times 10^4$ & $17.4$                  & $0.96$ & $1531$ \\
BIQAE   & $1.07 \times 10^{-4}$          & $3.88 \times 10^4$ & $3.91$                  & $0.99$ & $1383$ \\
CABIQAE & $8.02 \times 10^{-5}$ & $5.20 \times 10^4$ & $4.84$                  & $0.92$ & $1911$ \\
DCS     & $1.91 \times 10^{-3}$          & $8.10 \times 10^4$ & $7.04 \times 10^{-3}$   & --     & --     \\
\bottomrule[1.2pt]
\end{tabular}

\vspace{0.8em} 

\begin{tabular}{@{} l c c c c c @{}}
\toprule[1.2pt]
\multicolumn{6}{c}{\textbf{Hardware-replay}} \\
\midrule[0.8pt]
Algorithm & $\epsilon^{50}$ & $N_q^{50}$ & $t_{\mathrm{cl}}^{50}$ & Cov. & $K_{50}$ \\
\midrule
BAE     & $1.77 \times 10^{-3}$ & $4.80 \times 10^4$ & $45.72$                 & $0.82$  & $33$  \\
BIQAE   & $1.99 \times 10^{-2}$ & $5.87 \times 10^4$ & $1.61$                  & $0.015$ & $213$ \\
CABIQAE & $1.82 \times 10^{-3}$ & $4.95 \times 10^4$ & $1.25$                  & $0.724$ & $23$  \\
DCS     & $1.56 \times 10^{-2}$ & $6.55 \times 10^4$ & $2.99 \times 10^{-5}$   & --      & --    \\
\bottomrule[1.2pt]
\end{tabular}
\end{table}

Overall, Table~\ref{tab:reduced_instance_summary_compact} shows that BAE attains a slightly smaller hardware-replay median final error than CABIQAE, \(1.77\times10^{-3}\) versus \(1.82\times10^{-3}\), at a comparable median query budget. The practical difference is the classical and circuit-depth cost: CABIQAE uses shallower amplified circuits, \(K_{50}=23\) versus \(33\) for BAE, and reduces the median classical post-processing runtime from \(45.72\) seconds to \(1.25\) seconds. The runtime panels in Appendix~\ref{appsec:ae_extended_results}, Figure~\ref{fig:ae_validation_final_error_density}, show the same separation at the trajectory level: BAE occupies a much slower classical-runtime region for essentially the same final-error scale. Thus, although BAE gives the lowest final median error in the reduced hardware-replay benchmark, CABIQAE is the more suitable routine for repeated noisy hardware executions under realistic runtime constraints.

BAE's classical post-processing overhead can become a practical bottleneck in
larger noisy amplitude estimation tasks. Table~\ref{tab:reduced_instance_summary_compact}
and Figure~\ref{fig:ae_validation_final_error_density} show that CABIQAE reaches
essentially the same hardware-replay final-error scale as BAE with a much lower
classical runtime. These results therefore motivate using CABIQAE in the full CVA
experiment, where accuracy must be balanced against quantum query cost,
classical runtime and interval reliability.

\begin{figure}[H]
    \centering
    \begin{subfigure}{0.8\textwidth}
        \centering
        \includegraphics[width=\linewidth]{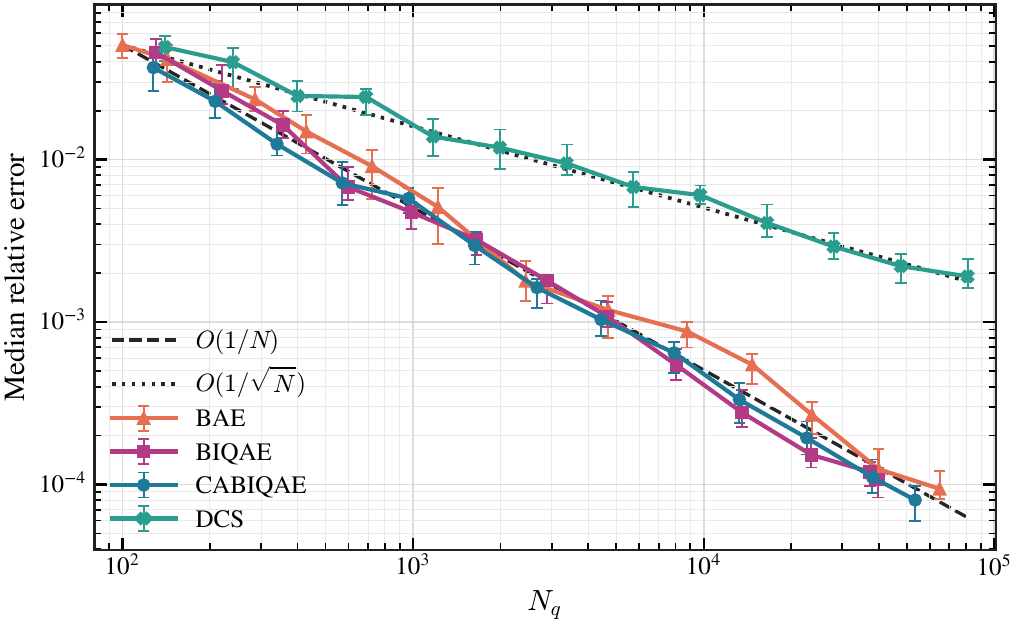}
    \end{subfigure}
    
    \vspace{0.5cm} 
    
    \begin{subfigure}{0.8\textwidth}
        \centering
        \includegraphics[width=\linewidth]{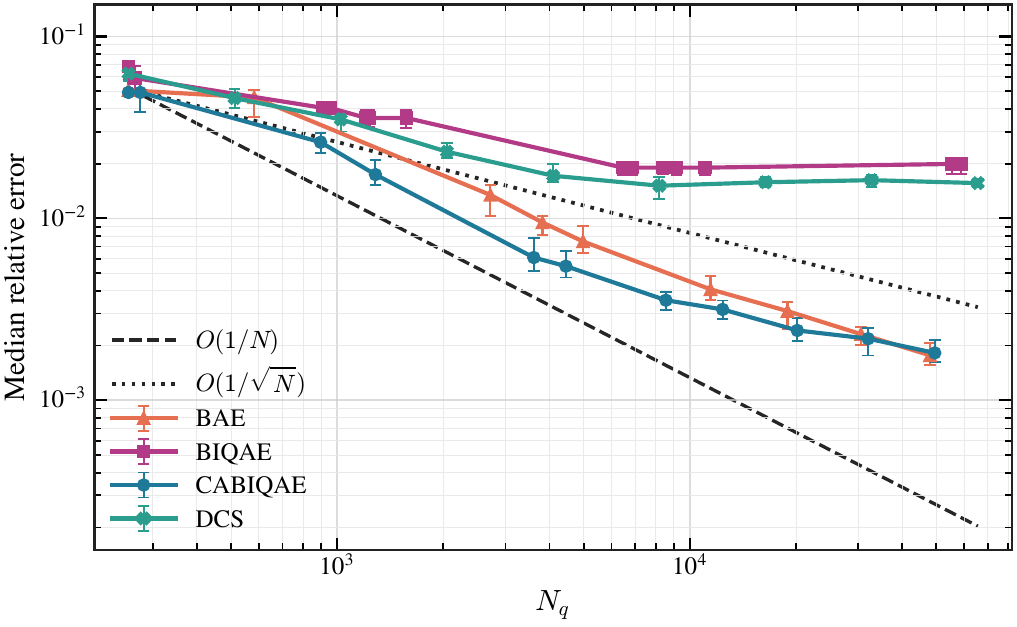}
    \end{subfigure}
    
    \caption{
    Reduced test-instance amplitude estimation results in the noiseless (top) and
    hardware-replay (bottom) regimes.
    }
    \label{fig:ae_reduced_main_results}
\end{figure}

\begin{table}[H]
\centering
\normalsize
\setlength{\tabcolsep}{3.4pt}
\renewcommand{\arraystretch}{1.12}
\caption{
Empirical query-scaling fits of the median relative error using
$\log \varepsilon=\alpha+\beta\log N_q$. Values near $\beta=-1$ indicate
amplified QAE-like scaling, while $\beta\simeq-1/2$ corresponds to Monte Carlo
sampling. Uncertainties are standard errors; RMSE denotes the fit root mean square error.
}
\label{tab:empirical_error_scaling_fit_metrics}
\begin{tabular}{@{} l c c c c @{}}
\toprule[1.2pt]
\multicolumn{5}{c}{\textbf{Noiseless}} \\
\midrule[0.8pt]
Algorithm & $\beta$ & $\alpha$ & RMSE & $R^2$ \\
\midrule
BAE     & $-0.96\pm 0.02$ & $1.47\pm 0.19$ & 0.15 & 0.994 \\
BIQAE   & $-0.93\pm 0.04$ & $0.89\pm 0.29$ & 0.35 & 0.980 \\
CABIQAE & $-0.92\pm 0.04$ & $0.80\pm 0.29$ & 0.35 & 0.982 \\
DCS     & $-0.51\pm 0.01$ & $-0.53\pm 0.08$ & 0.10 & 0.995 \\
\bottomrule[1.2pt]
\end{tabular}

\vspace{0.8em}

\begin{tabular}{@{} l c c c c @{}}
\toprule[1.2pt]
\multicolumn{5}{c}{\textbf{Hardware-replay}} \\
\midrule[0.8pt]
Algorithm & $\beta$ & $\alpha$ & RMSE & $R^2$ \\
\midrule
BAE     & $-0.68\pm 0.03$ & $0.94\pm 0.29$ & 0.16 & 0.981 \\
BIQAE   & $-0.25\pm 0.04$ & $-1.44\pm 0.32$ & 0.17 & 0.853 \\
CABIQAE & $-0.69\pm 0.03$ & $0.77\pm 0.29$ & 0.18 & 0.976 \\
DCS     & $-0.30\pm 0.05$ & $-1.28\pm 0.39$ & 0.19 & 0.869 \\
\bottomrule[1.2pt]
\end{tabular}
\end{table}

\subsection{Classical CVA benchmark and discretisation analysis}
\label{sec:classical_benchmark_results}

 After considering several alternative derivative-portfolio configurations, we select as the baseline case a portfolio composed of two derivatives written on two correlated underlying assets. Specifications of this two-asset netting set can be found in Table \ref{tab:portfolio_two_assets_instance}. This choice reflects a compromise between trainability, circuit depth, and the need to retain a sufficiently realistic use case for credit valuation adjustment.

\begin{table}[htbp]
\centering
\caption{Portfolio composition for the two-asset CVA instance.}
\label{tab:portfolio_two_assets_instance}

\setlength{\tabcolsep}{5pt}
\renewcommand{\arraystretch}{1.15}

\begin{tabularx}{\textwidth}{@{}
    c
    l
    >{\raggedright\arraybackslash}X
    c
    c
    c
    c
@{}}
\toprule

\multicolumn{7}{@{}l@{}}{\textbf{Portfolio composition}} \\
\midrule

Instrument & Underlying & Type & Position & Multiplier & Strike & Maturity \\
\midrule

1
& EURO STOXX 50
& European call
& Long
& 4
& 4500
& 0.25 \\

2
& SMI
& European put
& Short
& 2
& 12500
& 0.50 \\

\bottomrule
\end{tabularx}
\end{table}
\vspace{-5pt}
\paragraph{Classical Monte Carlo CVA.}
For the portfolio in Table~\ref{tab:portfolio_two_assets_instance}, the
continuous-underlying reference value is computed by Monte Carlo simulation
under the calibrated risk-neutral multi-asset model described in
Section~\ref{sec:classical_cva_valuation}. Using the numerical controls
reported in Appendix~\ref{app:classical_benchmark_details}, we obtain
\begin{equation}\label{val:mc_cont_cva}
    \widehat{\mathrm{CVA}}^{\mathrm{cont}}_{\mathrm{MC}}
    =
    1.091 \pm 0.001.
\end{equation}
This estimate is used as the baseline reference for the subsequent analysis.
It is not affected by finite-grid discretisation, domain truncation, quantum
encoding error or amplitude estimation error. 

\paragraph{Classical finite-grid CVA.}
The same Monte Carlo simulation is then projected onto the finite
time--market grid used by the quantum CVA pipeline. In the full-pipeline
instance analysed below, each underlying is encoded with \(n_1=n_2=2\)
price-register qubits, corresponding to \(N_1=N_2=4\) price bins and a
\(4\times4\) joint market grid at each exposure date. Together with the
\(m=2\) time-register qubits required by the four exposure dates, this gives a
six-qubit time--market register, before adding the ancillas used in the CVA
encoding.

\begin{table}[H]
\centering
\normalsize
\caption{Finite-grid numerical inputs used in the two-asset quantum CVA pipeline in the six-qubit instance considered.}
\label{tab:finite_grid_cva_instance}

\setlength{\tabcolsep}{5pt}
\renewcommand{\arraystretch}{1.15}

\begin{tabularx}{\textwidth}{@{}
    c
    l
    c
    c
    >{\raggedright\arraybackslash}X
@{}}
\toprule

\multicolumn{5}{@{}l@{}}{\textbf{Underlying discretisation}} \\
\midrule

\(k\) & Underlying
& \(\mu_{S,k}(0,T)\)
& \(\eta_{S,k}(0,T)\)
& Bin edges \(\mathbf b^{(k)}\) \\
\midrule

1
& EURO STOXX 50
& 5716.98
& 779.15
& \makecell[l]{\((3379.53,\ 4548.26,\ 5716.98,\) \\
                 \(\phantom{(}6885.70,\ 8054.42)\)} \\

2
& SMI
& 14150.69
& 1562.06
& \makecell[l]{\((9464.52,\ 11807.60,\ 14150.69,\) \\
                 \(\phantom{(}16493.77,\ 18836.86)\)} \\

\midrule

\multicolumn{5}{@{}l@{}}{\textbf{Time grid, discounting and scaling}} \\
\midrule

\multicolumn{2}{@{}l}{Input}
& Symbol
& \multicolumn{2}{l@{}}{Value} \\
\midrule

\multicolumn{2}{@{}l}{Exposure dates}
& \(\mathbf t\)
& \multicolumn{2}{l@{}}{\((0.125,\ 0.250,\ 0.375,\ 0.500)\)} \\

\multicolumn{2}{@{}l}{Discount factors}
& \(\mathbf p\)
& \multicolumn{2}{l@{}}{\((0.9975,\ 0.9950,\ 0.9923,\ 0.9895)\)} \\

\multicolumn{2}{@{}l}{Default increments}
& \(\Delta \mathbf q\)
& \multicolumn{2}{l@{}}{\((1.9709,\ 1.9705,\ 1.9701,\ 1.9697)\times 10^{-4}\)} \\

\multicolumn{2}{@{}l}{Scaling constants}
& \((C_p,C_q,C_v)\)
& \multicolumn{2}{l@{}}{\((0.9975,\ 1.9709\times 10^{-4},\ 9581.04)\)} \\

\bottomrule
\end{tabularx}
\end{table}

The resulting finite-grid inputs are summarised in
Table~\ref{tab:finite_grid_cva_instance}. These inputs define the conditional
market probabilities, joint probability tensor, positive-exposure table,
discount vector, default-increment vector and scaling constants encoded in the
quantum circuit.\footnote{For brevity, the numerical values of the joint probability tensor
\(P_{i,\mathbf{j}}\) and positive-exposure tensor \(v_{i,\mathbf{j}}\) are not
reported here, as each contains \(64\) entries in the present instance. The
complete tensors, together with the code used to generate them, are available
in the project repository \parencite{AuthorRepo2026}.}

From these finite-grid objects, the tabulated CVA associated with the
\(n_1=n_2=2\) discretisation is
\begin{equation}\label{value:cva_tab_n2}
    \widehat{\mathrm{CVA}}_{\Delta}^{\mathrm{tab}}(n=4)
    =
    0.522 .
\end{equation}
This is the classical reference that the quantum CVA circuit must reproduce,
since it is computed from exactly the same finite-grid inputs encoded in the
pipeline. Relative to the continuous-underlying Monte Carlo benchmark, this
implies a relative finite-grid error, including both price-domain truncation
and discretisation effects, of
\begin{equation}
    \varepsilon_{\mathrm{grid}}(n=4)
    =
    \frac{
    \left|
    \widehat{\mathrm{CVA}}_{\Delta}^{\mathrm{tab}}(n=4)
    -
    \widehat{\mathrm{CVA}}_{\mathrm{MC}}^{\mathrm{cont}}
    \right|}
    {
    \widehat{\mathrm{CVA}}_{\mathrm{MC}}^{\mathrm{cont}}
    }
    =
    \frac{|0.522-1.091|}{1.091}
    =
    52.15\%.
\end{equation}
In addition, we assess the finite price-grid discretisation error convergence following
\textcite{Alcazar2021CVA}. The truncated domains \(\mathcal D_k\) are kept fixed,
while the number of market-register qubits is increased. The resulting
convergence profile is shown in Figure~\ref{fig:finite_grid_convergence}, with numerical results in Table \ref{tab:grid_cva_resource_scaling} of Appendix \ref{app:convergence_analysis}.

\begin{figure}[htbp]
    \centering
    \includegraphics[width=0.60\textwidth]{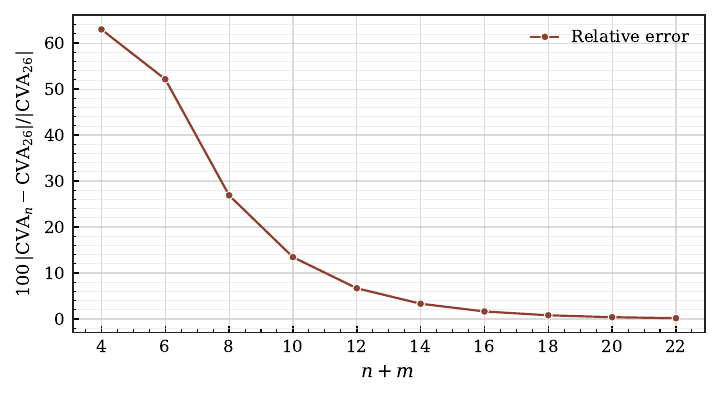}
    \caption{Relative error of the tabulated CVA
    as the number of market-register qubits increases while keeping $m=2$ fixed, measured against a
    high-resolution finite-grid reference. }
    \label{fig:finite_grid_convergence}
\end{figure}

The profile decreases monotonically as the price-register size increases,
showing that the error measured in Figure~\ref{fig:finite_grid_convergence} is
mainly driven by the coarseness of the market discretisation. For the
highest-resolution grid considered, \(n=26\), the tabulated value is
\begin{equation}\label{fine_grid_cva}
     \widehat{\mathrm{CVA}}_{\Delta}^{\mathrm{tab}}(n=26) :=\widehat{\mathrm{CVA}}_{\Delta}^{\mathrm{tab}}(26)
    =
    1.090,
\end{equation}
to be compared with the continuous-underlying Monte Carlo benchmark \eqref{val:mc_cont_cva}. Since the grid-convergence error at this resolution is already negligible, the
remaining relative gap,
\begin{equation}
    \varepsilon_{\mathrm{trunc}}
    :=
    \frac{
    \left|
    \widehat{\mathrm{CVA}}_{\Delta}^{\mathrm{tab}}(26)
    -
    \widehat{\mathrm{CVA}}_{\mathrm{MC}}^{\mathrm{cont}}
    \right|}
    {
    \widehat{\mathrm{CVA}}_{\mathrm{MC}}^{\mathrm{cont}}
    }
    =
    9.17\times 10^{-4},
\end{equation}
can be attributed to the finite truncation of the price domains
\(\mathcal D_k\), rather than to price-grid discretisation. 
\subsection{Quantum encoding of the finite-grid CVA}
 As discussed in Section~\ref{subsec:quantum_pipeline}, the goal of the encoding stage is to train a parametrised circuit $A_{\Theta}$ that prepares the finite-grid CVA instance for amplitude estimation. Since the final amplitude estimation step is intended to run under noisy hardware conditions, we first investigate whether the encoding circuits should also be trained under noisy conditions.

\begin{table}[htbp]
\centering
\normalsize
\caption{Circuit resources and training performance for the QCBM and CRCA blocks used in the finite-grid CVA encoding. The column 1Q reports the number of single-qubit gates, while 2Q counts are obtained after decomposing controlled rotations and $R_{ZZ}$ gates into elementary CNOT gates. Performance metrics are evaluated at the trained parameter values $\Theta$.} 
\label{tab:trained_circuit_specs}

\setlength{\tabcolsep}{4pt}
\renewcommand{\arraystretch}{1.2}

\setlength{\aboverulesep}{0pt}
\setlength{\belowrulesep}{0pt}

\begin{tabularx}{\textwidth}{@{} 
    >{\raggedright\arraybackslash}X 
    c 
    c 
    c 
    c 
    c 
    c 
    c 
@{}}
\toprule

\multicolumn{1}{@{} >{\raggedright\arraybackslash}X}{} 
& 
& 
& \multicolumn{3}{c}{\textbf{Logical Resources}} 
& \multicolumn{2}{c@{}}{\textbf{Performance Metrics}} \\

\cmidrule{4-6} \cmidrule{7-8}

\multicolumn{1}{@{} >{\raggedright\arraybackslash}X}{\textbf{}}
& \textbf{Layers}
& \textbf{Params}
& \textbf{1Q}
& \textbf{2Q}
& \textbf{Depth}
& \textbf{Noiseless}
& \multicolumn{1}{c@{}}{\textbf{Noise + Shots}} \\
\midrule

QCBM $G_{\theta}$
& $6$
& $51$
& $51$
& $30$
& $33$
& $\mathrm{KL}_{\epsilon}=1.059\times 10^{-2}$
& $\mathrm{KL}_{\epsilon}=5.006\times 10^{-2}$ \\

CRCA $R_q(\phi_q)$
& $1$
& $6$
& $6$
& $4$
& $10$
& $\mathcal{L}_q=2.918\times 10^{-8}$
& $\mathcal{L}_q=7.368\times 10^{-5}$ \\

CRCA $R_p(\phi_p)$
& $1$
& $6$
& $6$
& $4$
& $10$
& $\mathcal{L}_p=3.719\times 10^{-6}$
& $\mathcal{L}_p=8.207\times 10^{-5}$ \\

CRCA $R_v(\phi_v)$
& $2$
& $40$
& $98$
& $86$
& $184$
& $\mathcal{L}_v=3.197\times 10^{-3}$
& $\mathcal{L}_v=1.019\times 10^{-2}$ \\

CVA $A_{\Theta}$
& --
& $103$
& $161$
& $124$
& $213$
& --
& -- \\
\bottomrule
\end{tabularx} 
\end{table}

Training only in the ideal statevector regime may give a circuit that fits the target distribution well in simulation, but performs differently once finite shots and backend noise are introduced. To quantify this effect before fixing the final $n+m=6$ encoding, we compare ideal training against training under a simulated noise model, afterwards evaluating each circuit under both the noiseless and noisy finite-shot regimes. The
results for $G_\theta$, summarised in Figure \ref{fig:training_comparison}, indicate that noisy training does not provide a robustness advantage for
this instance: the optimisation landscape becomes harder, and the optimiser is
often unable to make consistent progress beyond the noise-induced floor. We
therefore select the statevector-trained QCBM circuit with \(L=6\), which
achieves a noisy finite-shot KL divergence of \(5.006\times10^{-2}\),
substantially below the noisy evaluation of the noisy-trained \(L=6\) circuit,
\(4.039\times10^{-1}\). An analogous layer-depth analysis is carried out for the remaining encoding
blocks. Based on this study, we choose \(L=1\) for \(R_p(\phi_p)\) and
\(R_q(\phi_q)\), and \(L=2\) for \(R_v(\phi_v)\), as a practical compromise
between expressivity, trainability and circuit depth.

Table~\ref{tab:trained_circuit_specs} summarises the selected blocks used in
the final finite-grid encoding, while Figure~\ref{fig:trained_circuits} shows
the corresponding training trajectories and fitted functions. The QCBM block
captures the joint time--market distribution with a statevector KL divergence
of \(1.059\times10^{-2}\), which increases to \(5.006\times10^{-2}\) under
finite shots and noise. This degradation is visible in the measured
distribution, but the dominant modes of the target distribution are still
preserved. The time-dependent discount and default-probability rotations are
essentially learned exactly in the ideal regime and remain stable under noisy
evaluation, with noisy losses below \(10^{-4}\). The exposure rotation is the
most demanding block, both in circuit resources and in approximation error,
because it must represent a less regular payoff surface over the joint grid.
Even so, the noisy loss remains of order \(10^{-2}\), which is sufficient for
the subsequent finite-grid CVA experiment.

\begin{figure}[H]
    \centering

    \begin{subfigure}{0.45\textwidth}
        \centering
        \includegraphics[width=\linewidth]{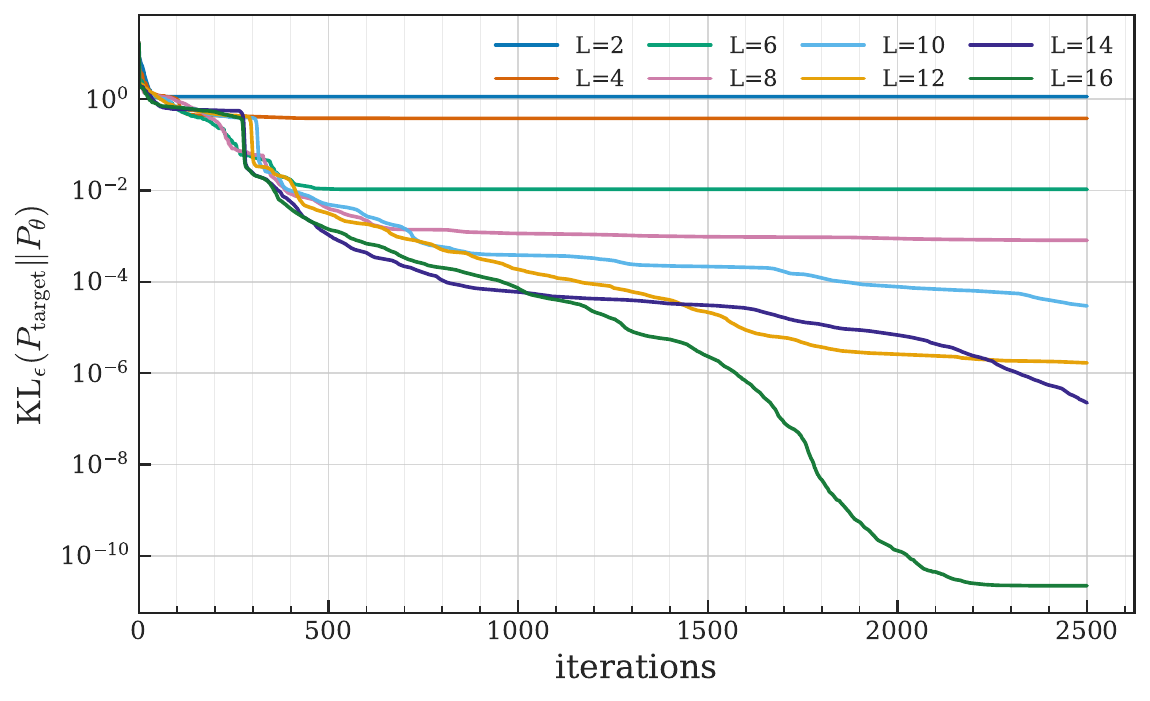}
        \label{fig:subplot_a}
    \end{subfigure}
    \hspace{-0.01\textwidth}
    \begin{subfigure}{0.45\textwidth}
        \centering
        \includegraphics[width=\linewidth]{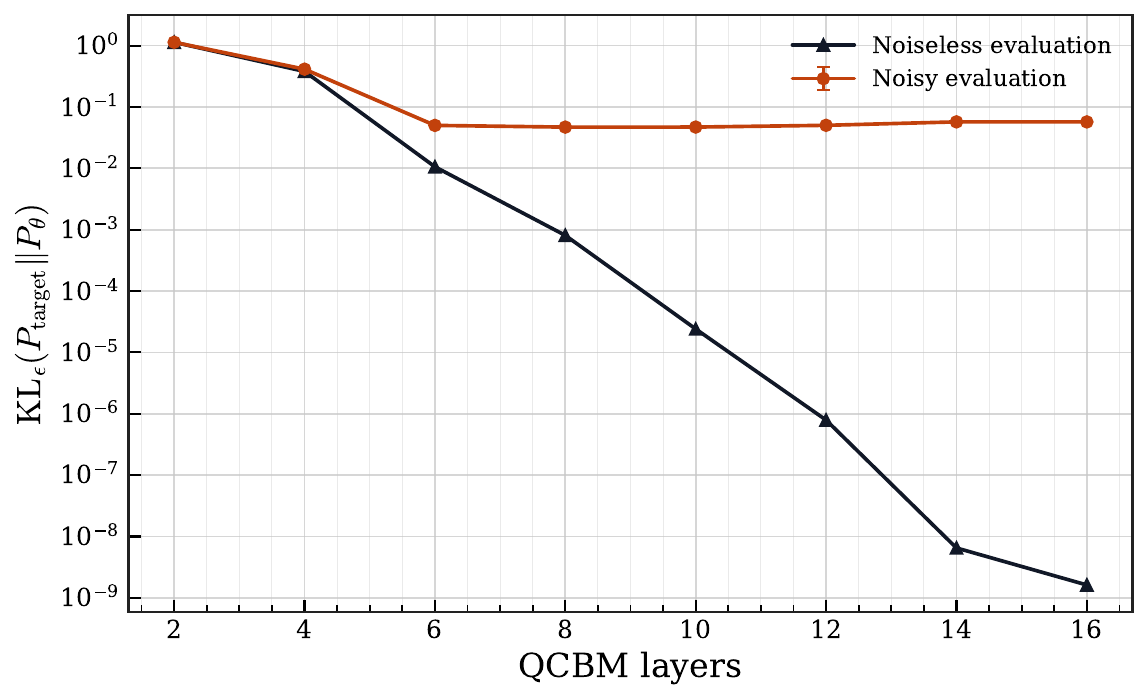}
        \label{fig:subplot_b}
    \end{subfigure}

    \vspace{-0.10cm}

    \begin{subfigure}{0.45\textwidth}
        \centering
        \includegraphics[width=\linewidth]{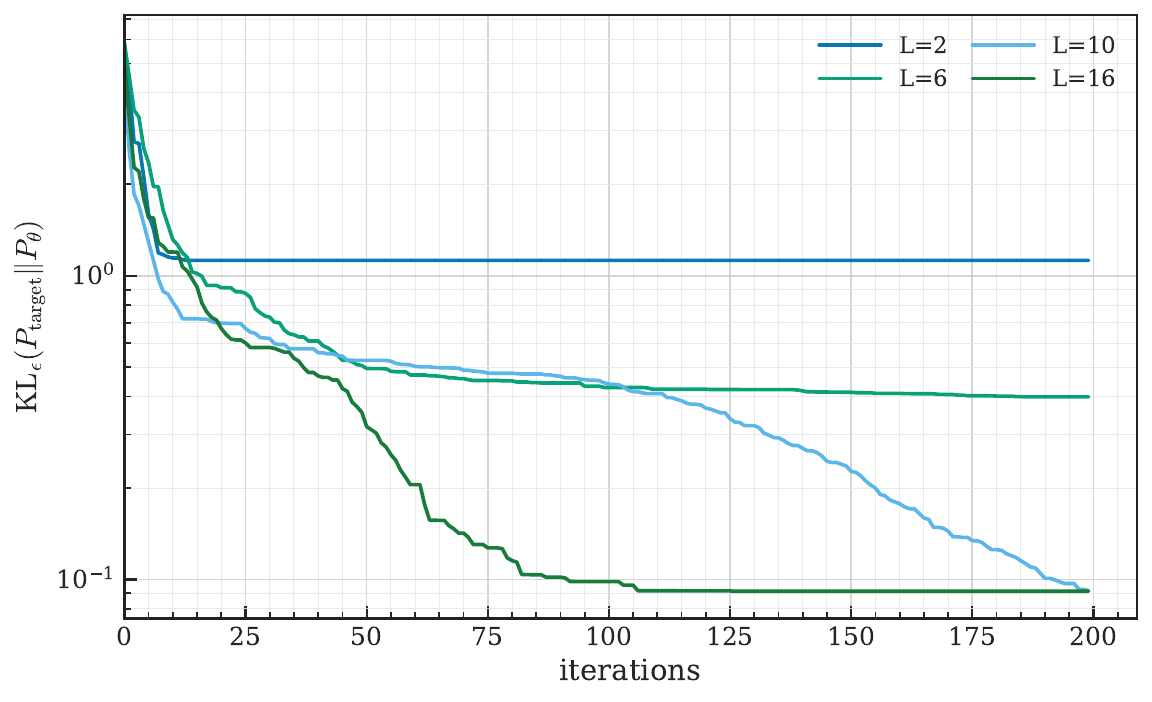}
        \label{fig:subplot_c}
    \end{subfigure}
    \hspace{-0.01\textwidth}
    \begin{subfigure}{0.45\textwidth}
        \centering
        \includegraphics[width=\linewidth]{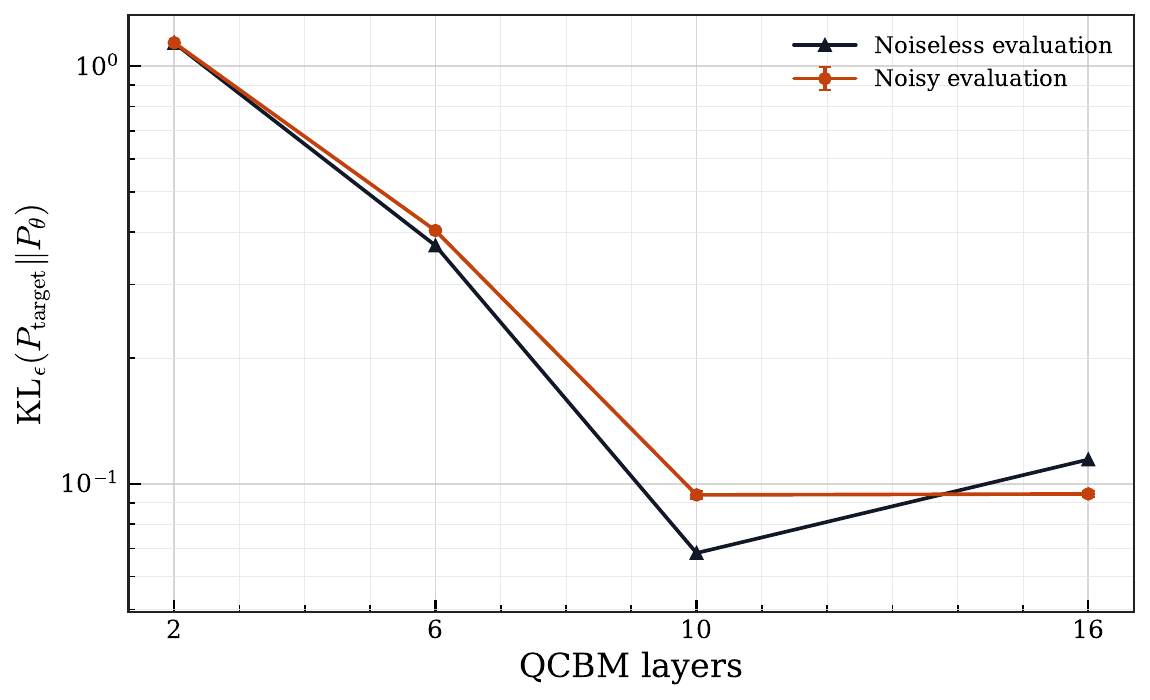}
        \label{fig:subplot_d}
    \end{subfigure}

    \caption{
KL trajectories and post-training evaluations for ideal and noise-aware QCBM training.
Top: ideal training and its clean/noisy evaluation; bottom: noisy finite-shot training and its clean/noisy evaluation.
}
    \label{fig:training_comparison}
\end{figure}

\paragraph{Noiseless quantum CVA value.}
Using the trained circuit \(A_{\Theta}\), the noiseless statevector evaluation
gives
\begin{equation}
    \widehat{\mathrm{CVA}}_{\mathrm{SV}}(\Theta) = 0.670 .
    \label{val:cva_statevector}
\end{equation}
We compare this value with the tabulated finite-grid benchmark
\(\widehat{\mathrm{CVA}}_{\Delta}^{\mathrm{tab}}(n=4)\) through the relative encoding
error
\begin{equation}
    \varepsilon_{\mathrm{enc}}(n=4)
    =
    \frac{
    \left|
        \widehat{\mathrm{CVA}}_{\mathrm{SV}}(\Theta)
        -
        \widehat{\mathrm{CVA}}_{\Delta}^{\mathrm{tab}}(n=4)
    \right|
    }{
        \widehat{\mathrm{CVA}}_{\Delta}^{\mathrm{tab}}(n=4)
    }
    =
    28.35\%.
    \label{eq:relative_encoding_error_result}
\end{equation}

This error should be interpreted as the residual training error of the quantum
encoding: it measures how far the trained QCBM and CRCA blocks are from the
tabulated finite-grid object when the encoded amplitude is read out exactly at
the statevector level.

The robustness analysis reported in Section~\ref{appsec:robust_analysis}
shows that the resulting quantum CVA estimator remains stable under the considered perturbations of the underlying risk factors.

\begin{figure}[H]
    \centering
    
    \includegraphics[width=0.85\textwidth]{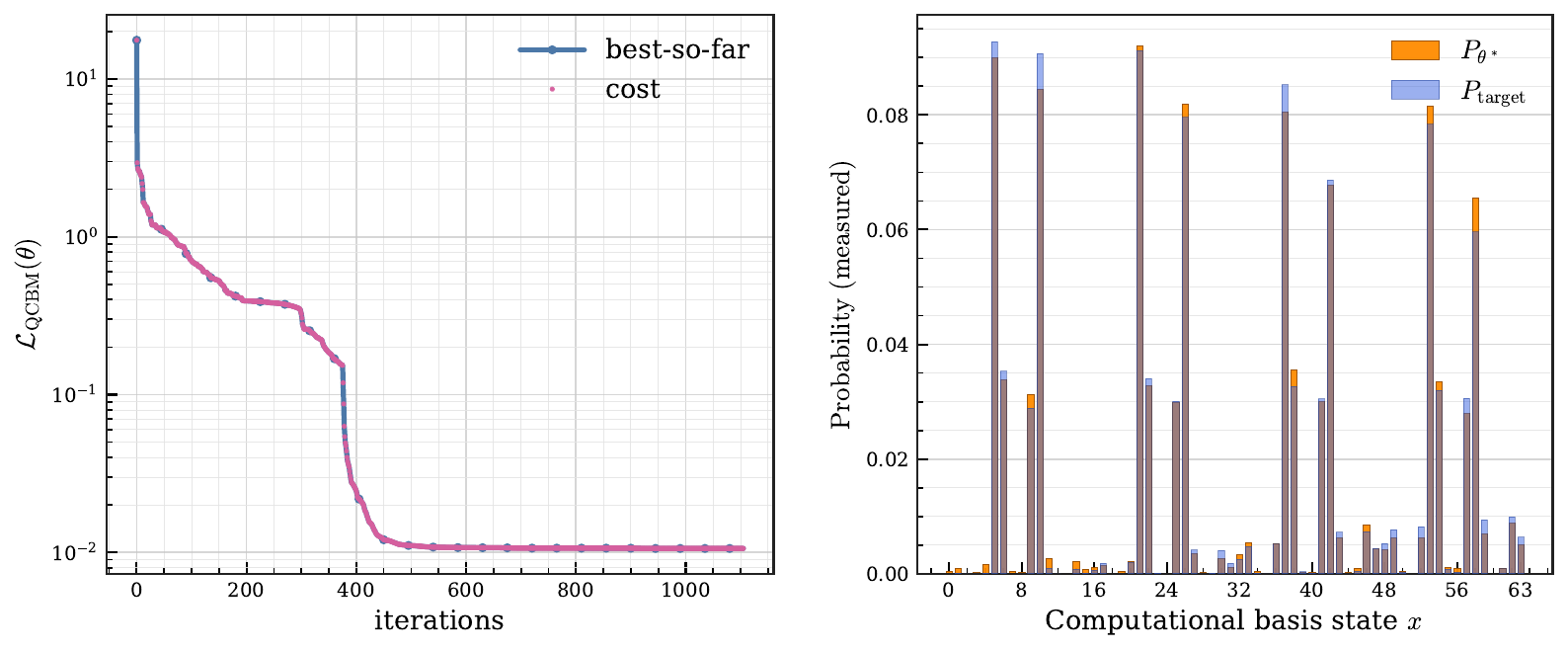}
    
    \vspace{0.05cm} 
    
    \includegraphics[width=0.85\textwidth]{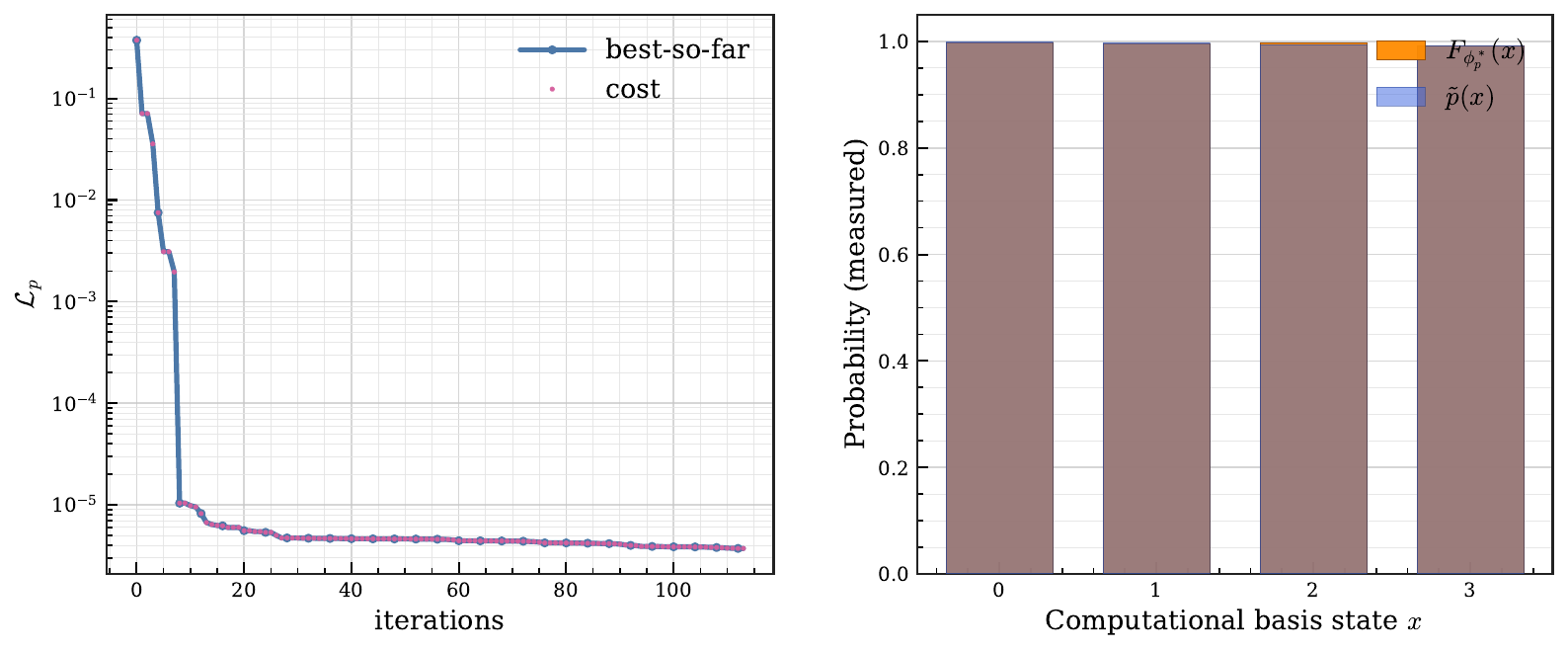}
    
    \vspace{0.05cm} 
    
    \includegraphics[width=0.85\textwidth]{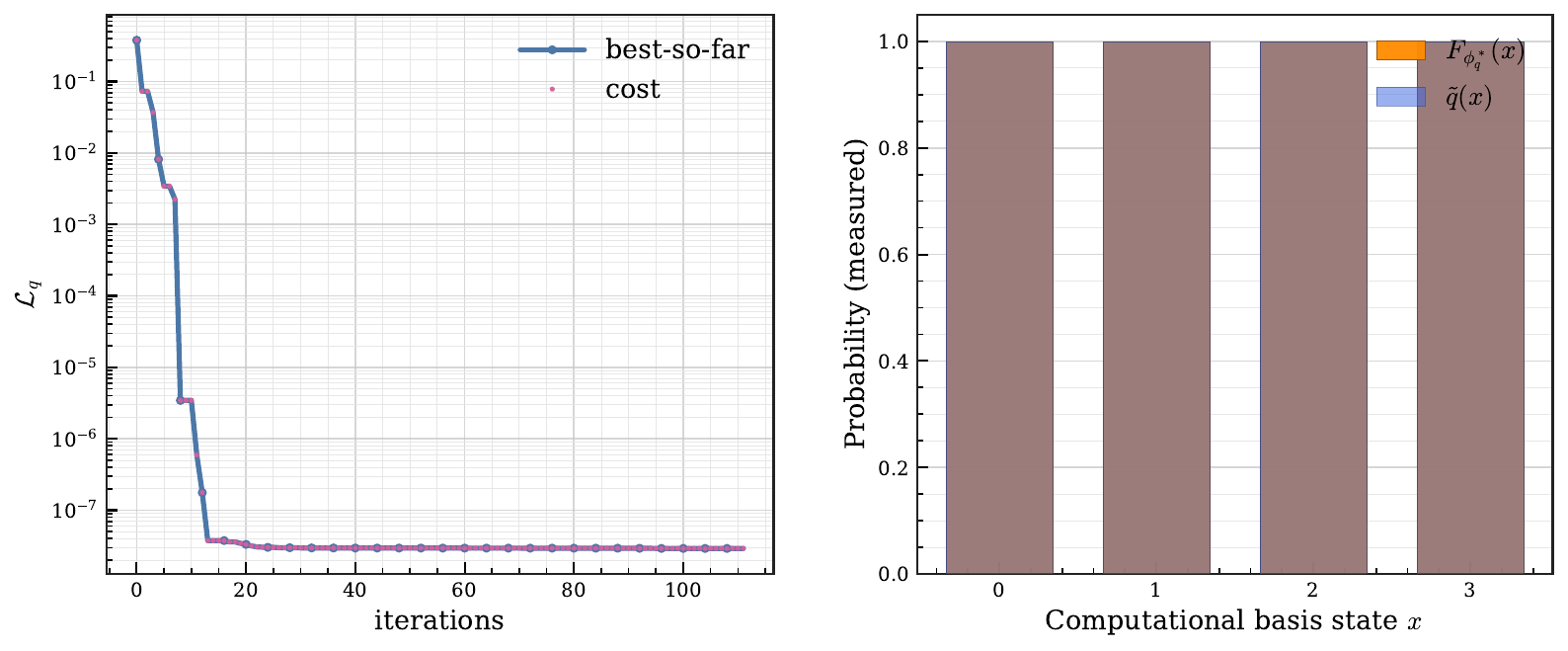}
    
    \vspace{0.05cm} 
    
    \includegraphics[width=0.85\textwidth]{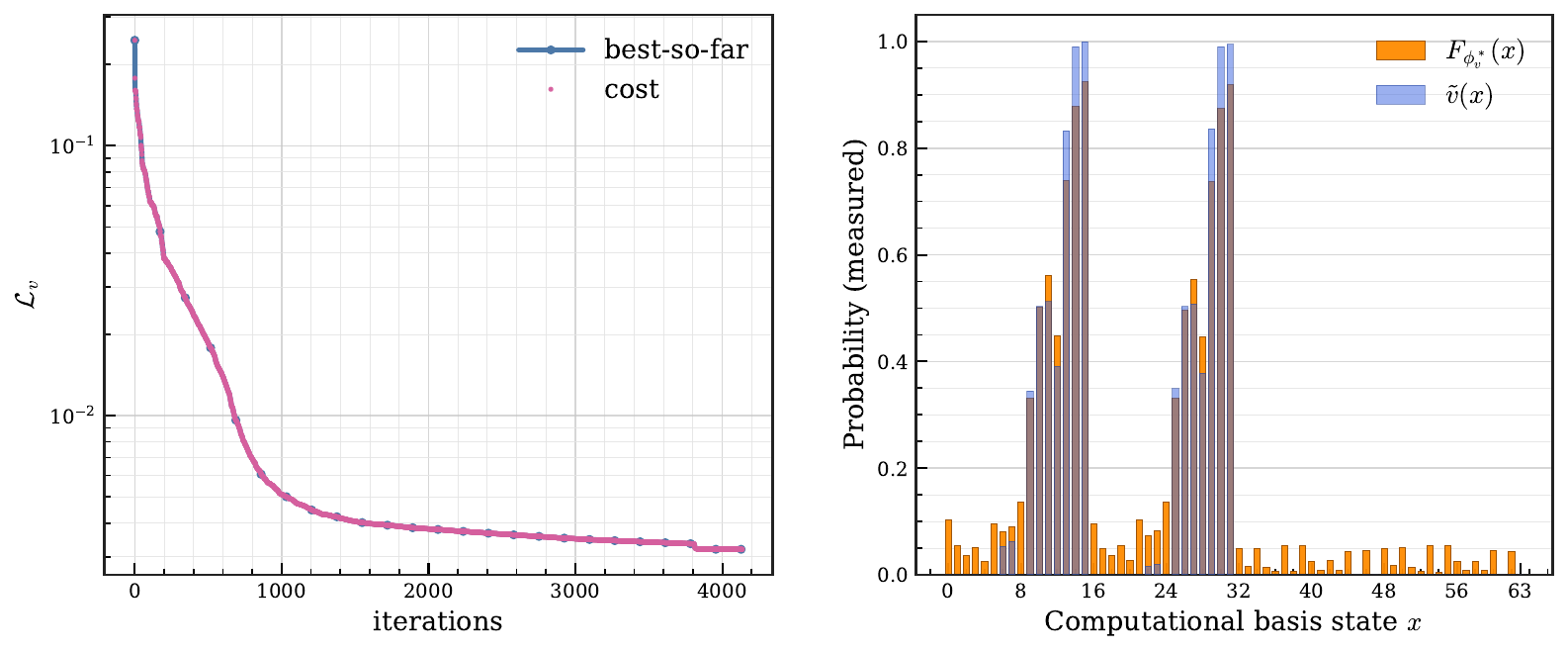}
    
    \caption{
Training diagnostics and final distribution matching for the quantum-encoding
blocks. The left column shows the evolution of the optimisation cost, whereas the right column
compares the target distributions with the probabilities measured from the
corresponding trained circuits.
}
\label{fig:encoding_training_diagnostics}
    \label{fig:trained_circuits}
\end{figure}

\subsection{Quantum CVA estimation using amplitude estimation}
\label{subsec:quantum_cva_ae_results}
Results reported in this section are obtained from $R=300$ independent estimation trajectories, while algorithmic hyperparameters are again reported in
Table~\ref{tab:ae_hyperparameters} of
Appendix~\ref{appsec:classical_market_inputs}. Trajectory-to-trajectory variability arises from independent finite-shot sampling of the corresponding measurement probabilities and, for adaptive AE algorithms, may result in different amplification schedules and actual query costs. BAE is omitted from this final CVA comparison because the validation study in Table~\ref{tab:reduced_instance_summary_compact} and Figure~\ref{fig:ae_validation_final_error_density} shows that its Bayesian post-processing can become a practical bottleneck: in hardware-replay it reaches a final-error scale comparable to CABIQAE, but with a median classical runtime of \(45.72\) seconds instead of \(1.25\) seconds. For the much deeper CVA unitary, repeating BAE across hundreds of trajectories and several precision targets would therefore be computationally impractical.

\begin{table}[H]
\centering
\normalsize
\setlength{\tabcolsep}{7pt}
\renewcommand{\arraystretch}{1.3}
\caption{
Six-qubit CVA instance summary for the noiseless and hardware-replay regimes.
All errors reported are relative CVA errors. Both regimes are evaluated over
\(R=300\) independent estimation trajectories.
}
\label{tab:cva_instance_summary_compact}
\begin{tabular}{@{} l c c c c c @{}}
\toprule[1.2pt]
\multicolumn{6}{c}{\textbf{Noiseless}} \\
\midrule[0.8pt]
Algorithm & \(\epsilon^{50}\) & \(N_q^{50}\) & \(t_{\mathrm{cl}}^{50}\) & Cov. & \(K_{50}\) \\
\midrule
BIQAE
& \(2.42{\times}10^{-4}\)
& \(5.50{\times}10^{4}\)
& \(8.98{\times}10^{-3}\)
& 0.983
& 1534 \\

CABIQAE
& \(2.27{\times}10^{-4}\)
& \(5.40{\times}10^{4}\)
& \(2.20{\times}10^{-3}\)
& 0.897
& 1711 \\

DCS
& \(6.72{\times}10^{-3}\)
& \(6.55{\times}10^{4}\)
& \(6.95{\times}10^{-5}\)
& --
& -- \\
\bottomrule[1.2pt]
\end{tabular}

\vspace{1.5em}

\begin{tabular}{@{} l c c c c c @{}}
\toprule[1.2pt]
\multicolumn{6}{c}{\textbf{Hardware-replay}} \\
\midrule[0.8pt]
Algorithm & \(\epsilon^{50}\) & \(N_q^{50}\) & \(t_{\mathrm{cl}}^{50}\) & Cov. & \(K_{50}\) \\
\midrule
BIQAE
& \(8.92{\times}10^{-1}\)
& \(1.46{\times}10^{4}\)
& \(4.35{\times}10^{-1}\)
& 0.003
& 89 \\
CABIQAE
& \(5.15{\times}10^{-2}\)
& \(9.98{\times}10^{4}\)
& \(1.96{\times}10^{0}\)
& 0.617
& 3 \\
DCS
& \(4.21{\times}10^{-1}\)
& \(1.00{\times}10^{5}\)
& \(1.12{\times}10^{-4}\)
& --
& -- \\
\bottomrule[1.2pt]
\end{tabular}
\end{table}

For the hardware-replay regime in the CVA amplitude estimation experiment, the
hardware probabilities used to construct the replay model were obtained from
hardware executions carried out on the \texttt{ibm\_basquecountry} backend using
Q-CTRL's Performance Management Qiskit Function. This error-mitigation layer was used to improve the reliability of the
measured success probabilities before replaying the adaptive CABIQAE
trajectories.

\begin{figure}[H]
    \centering
    \begin{subfigure}{0.8\textwidth}
        \centering
        \includegraphics[width=\linewidth]{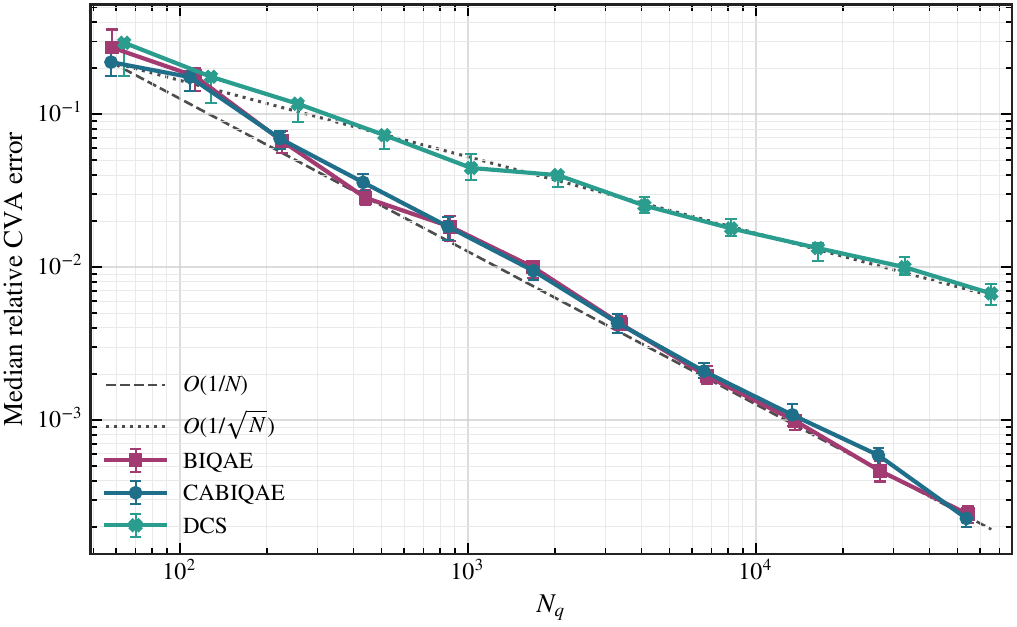}
    \end{subfigure}
    
    \vspace{0.5cm} 
    
    \begin{subfigure}{0.8\textwidth}
        \centering
        \includegraphics[width=\linewidth]{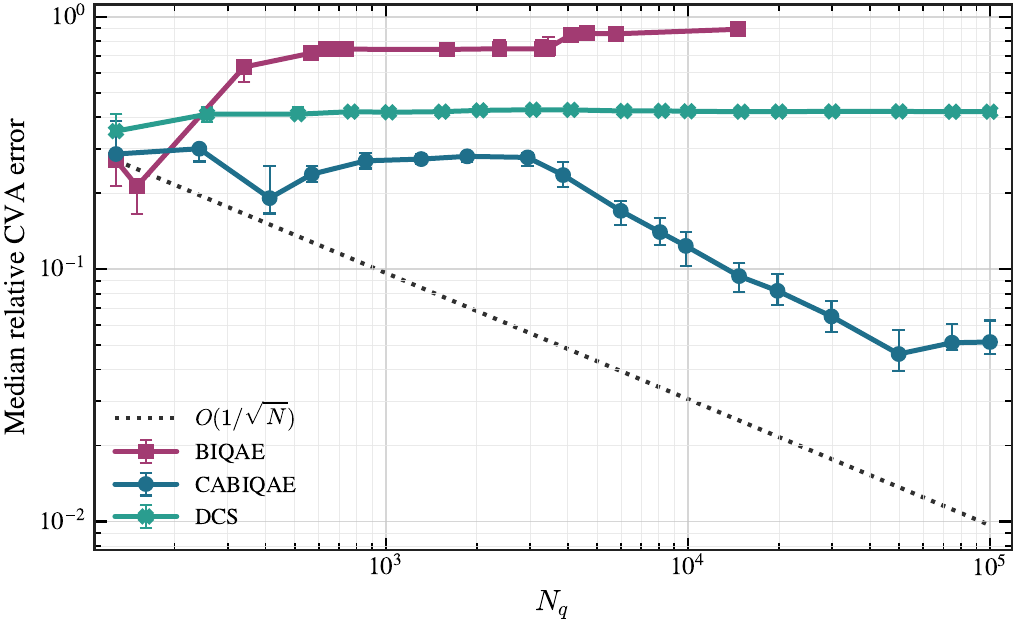}
    \end{subfigure}
    
    \caption{
    Quantum CVA (amplitude) estimation results in the noiseless (top) and hardware-replay (bottom) regimes.
    }
    \label{fig:cva_ae_reduced_main_results}
\end{figure}

Detailed logical and transpiled resource metrics for the amplified CVA circuits
are reported in Appendix~\ref{appsec:amplified_resource_tables}, Table \ref{tab:cva_amplified_logical_transpiled_resources}. Already at
\(k=1\), the transpiled circuit contains \(1122\) two-qubit gates and reaches
a transpiled depth of \(2691\), with both quantities increasing progressively with
the Grover power. This makes the full CVA circuit an exceptionally demanding
hardware workload. With Q-CTRL's Performance Management, the real-device measurements retain a
statistically resolved amplified signal through \(k=2\). The contrast model is
calibrated from four fitted points over the range \(k=0,\ldots,4\), with one
excluded point, yielding \(c_0=0.67\pm0.31\) and
\(\tau_{\mathrm c}=2.2\pm0.6\). This very short contrast length is
substantially smaller than the validation-circuit value
\(\tau_{\mathrm c}=33.96\pm2.80\), and is consistent with the much larger
transpiled depth of the full CVA amplification circuits. Full calibration
details can be found in Appendix~\ref{appsec:calibration},
Table~\ref{tab:cva_hardware_combined_calibration_summary}.

Table~\ref{tab:cva_instance_summary_compact} and
Figure~\ref{fig:cva_ae_reduced_main_results}, together with other complementary results in Appendix \ref{appsec:ae_extended_results}, report the final end-to-end
amplitude estimation experiment for the trained quantum CVA circuit. In the noiseless regime (Figure~\ref{fig:cva_ae_reduced_main_results} (top)), both adaptive AE routines recover the expected query-complexity advantage. BIQAE attains a median relative CVA error of
\(2.42\times10^{-4}\), while CABIQAE reaches \(2.27\times10^{-4}\), both far
below direct circuit sampling over comparable query budgets.

This quadratic improvement is not reproduced in hardware-replay. Hardware noise
strongly degrades both the error levels and the observed query scaling:
DCS remains essentially flat around \(4.21\times10^{-1}\) as the query budget
increases, while BIQAE deteriorates rather than converges, ending with median
error \(8.92\times10^{-1}\) and almost zero coverage. For this reason the comparison is
reported directly in terms of final error, query budget, maximum amplification,
classical runtime and coverage, rather than through a single aggregate index.

CABIQAE follows a markedly different trend in Figure~\ref{fig:cva_ae_reduced_main_results}
(bottom). Although the noisy replay does not recover a realistic quadratic
improvement, the contrast-aware scheduler prevents the progressive deterioration
seen in BIQAE and produces a clear downward error trend over the accessible
query range. Its median error reaches \(5.15\times10^{-2}\) at
\(N_q^{50}=9.98\times10^4\), with \(K_{50}=3\) and coverage \(0.617\),
substantially improving over both BIQAE and the nearly flat DCS baseline. The
curve nevertheless remains above the dashed \(O(1/\sqrt{N})\) reference, which
corresponds to the scaling expected from a classical Monte Carlo estimator in
the absence of hardware bias.

\subsection{Formal error decomposition}
\label{subsec:error_budget}

The final CVA estimate differs from the continuous Monte Carlo benchmark
through three successive approximations: replacing the continuous market model
by a finite grid, encoding the resulting probability table in a trained quantum
circuit, and estimating the marked amplitude from amplified circuit
measurements. Let \(\widehat{\mathrm{CVA}}_{\mathrm{AE}}(\Theta)\) denote the
amplitude estimation (AE) output after conversion to monetary units using the
rescaling factor \(M(1-R_{\mathrm{CVA}})C_vC_pC_q\). By inserting the exact statevector value
\(\widehat{\mathrm{CVA}}_{\mathrm{SV}}(\Theta)\) and the finite-grid benchmark
\(\widehat{\mathrm{CVA}}_{\Delta}^{\mathrm{tab}}(n)\) between this estimate
and \(\widehat{\mathrm{CVA}}_{\mathrm{MC}}^{\mathrm{cont}}\), the triangle
inequality yields
\begin{equation}
\varepsilon_{\mathrm{tot}}
:=
\frac{
  \left|
    \widehat{\mathrm{CVA}}_{\mathrm{AE}}(\Theta)
    -
    \widehat{\mathrm{CVA}}_{\mathrm{MC}}^{\mathrm{cont}}
  \right|
}{
  \widehat{\mathrm{CVA}}_{\mathrm{MC}}^{\mathrm{cont}}
}
\;\leq\;
\underbrace{\varepsilon_{\mathrm{grid}}}_{\text{(i)}}
\;+\;
\underbrace{\alpha_n\,\varepsilon_{\mathrm{enc}}}_{\text{(ii)}}
\;+\;
\underbrace{\beta_\Theta\,\varepsilon_{\mathrm{AE}}}_{\text{(iii)}} .
\label{eq:relative_error_decomposition}
\end{equation}
The dimensionless scale factors
\begin{equation}
\alpha_n
:=
\frac{\widehat{\mathrm{CVA}}_{\Delta}^{\mathrm{tab}}(n)}
     {\widehat{\mathrm{CVA}}_{\mathrm{MC}}^{\mathrm{cont}}},
\qquad
\beta_\Theta
:=
\frac{\widehat{\mathrm{CVA}}_{\mathrm{SV}}(\Theta)}
     {\widehat{\mathrm{CVA}}_{\mathrm{MC}}^{\mathrm{cont}}}
\label{eq:rescaling_factors}
\end{equation}
map each component-level relative error onto the common scale of the
continuous-model benchmark. Term~(i) captures the joint contribution of
price-domain truncation and discretisation; term~(ii) measures the
residual training error of the QCBM and CRCA blocks relative to the tabulated
finite-grid target; and term~(iii) corresponds to the statistical residual of
the amplitude estimator with respect to the probability encoded by the trained
circuit $A_\Theta$.  For the $n=4$ trained instance, the scale factors are
$\alpha_4 = 0.522/1.091 = 0.4785$ and
$\beta_\Theta = 0.670/1.091 = 0.6141$.
\begin{table}[htbp]
\centering
\caption{
Numerical error budget for the six-qubit trained CVA instance. The aggregate quantities are conservative propagated bounds and should not be interpreted as realised signed errors.  All errors are reported as percentages. CVA-scale contributions use
$\alpha_4 = 0.4785$ and $\beta_\Theta = 0.6141$.
}
\label{tab:relative_error_budget_compact}
\setlength{\tabcolsep}{5pt}
\renewcommand{\arraystretch}{1.12}
\begin{tabular}{@{} l c c @{}}
\toprule
\textbf{Contribution}
& \textbf{Component error}
& \textbf{CVA-scale contribution} \\
\midrule
Finite grid / truncation  (i)
  & $\varepsilon_{\mathrm{grid}} = 52.15\%$
  & $52.15\%$ \\
Encoding / training  (ii)
  & $\varepsilon_{\mathrm{enc}} = 28.35\%$
  & $13.57\%$ \\
\midrule
\textbf{Structural bound (i)\,+\,(ii)}
  & & $65.72\%$ \\
\midrule
Noiseless AE -- CABIQAE  (iii)
  & $\varepsilon_{\mathrm{AE}}^{\mathrm{id}} = 0.0227\%$
  & $0.0139\%$ \\
Hardware-replay -- CABIQAE  (iii)
  & $\varepsilon_{\mathrm{AE}}^{\mathrm{hw}} = 5.15\%$
  & $3.16\%$ \\
\midrule
\textbf{Total bound -- noiseless}
  & & $\approx 65.73\%$ \\
\textbf{Total bound -- hardware}
  & & $\approx 68.88\%$ \\
\bottomrule
\end{tabular}
\end{table}

Table~\ref{tab:relative_error_budget_compact} reports all contributions on the
common scale of the continuous Monte Carlo benchmark. The structural contribution, which combines discretisation and trained-circuit encoding errors, is therefore a conservative triangle-inequality upper bound,
\begin{equation}
\varepsilon_{\mathrm{struct}}^{\mathrm{bound}}
=
\varepsilon_{\mathrm{grid}}
+
\alpha_4\,\varepsilon_{\mathrm{enc}}
=
52.15\% + 13.57\%
=
65.72\%.
\label{eq:structural_error_budget}
\end{equation}
The realised statevector discrepancy is smaller,
\(|0.670-1.091|/1.091=38.59\%\), because the grid and encoding errors partially
offset each other. For context, in the four-qubit CVA instance of \textcite{Alcazar2021CVA},
the values reported in their work,
\(
\widehat{\mathrm{CVA}}_{\mathrm{MC}}^{\mathrm{cont},\mathrm{Alcazar}}
= 5.599 \times 10^{-5}
\)
and
\(
\mathrm{CVA}_{\mathrm{SV}}^{\mathrm{Alcazar}}(\Theta)
= 1.987 \times 10^{-5}
\),
imply a realised quantum-circuit discrepancy of \(64.5\%\),
illustrating that large low-resolution structural discrepancies are also
present in the reference circuit-level CVA study.

In the noiseless regime, CABIQAE reaches a median relative AE error of
\(0.0227\%\) with a median query count of \(5.40\times10^4\). Its contribution
on the common CVA scale is only \(0.0139\%\), leaving the propagated bound
numerically unchanged:
\begin{equation}
\varepsilon_{\mathrm{tot}}^{\mathrm{id}}
\leq
65.72\% + 0.0139\%
\approx
65.73\%.
\label{eq:ideal_total_error_budget}
\end{equation}

Under hardware-replay, CABIQAE reaches a median AE error of \(5.15\%\) with a
median query count of \(9.98\times10^4\). After rescaling, this corresponds to
a CVA-scale contribution of only \(3.16\%\). Thus, although hardware noise
clearly prevents the AE layer from reproducing its ideal behaviour, the noisy
AE residual is no longer a leading term in the end-to-end budget: it remains
well below both the finite-grid error (\(52.15\%\)) and the encoding and training
contribution (\(13.57\%\)). The amplitude estimation layer therefore remains the
most accurate stage of the evaluated noisy pipeline. The resulting bound is
\begin{equation}
\varepsilon_{\mathrm{tot}}^{\mathrm{hw}}
\leq
65.72\% + 3.16\%
\approx
68.88\%.
\label{eq:hardware_total_error_budget}
\end{equation}

The result depends materially on the noise-aware estimator. Under the same
hardware-replay conditions, BIQAE would contribute \(54.8\%\) on the common
CVA scale and increase the aggregate bound to roughly \(120.5\%\). By contrast,
CABIQAE limits the additional error to \(3.16\%\) by exploiting the calibrated
shallow-depth likelihood instead of treating degraded amplified circuits as
fully informative.

The closeness of the ideal and hardware aggregate bounds is therefore driven by
the large structural terms, not by identical AE behaviour. In the ideal regime,
the CABIQAE residual continues to decrease with the query budget, consistently
with the expected quadratic AE scaling. Under hardware-replay, CABIQAE still
improves over the accessible query range, but with a much weaker, noise-limited
trend: additional queries are useful only while the calibrated shallow-depth
response remains informative. This difference is largely masked at the present
discretisation level, but it would become more visible if the structural errors
were reduced.
\vspace{-5pt}
\paragraph{Aggregate pipeline error.}
The principal numerical result of the error decomposition is obtained by
combining the three pipeline contributions. Using CABIQAE for the amplitude
estimation layer and
\(\widehat{\mathrm{CVA}}_{\mathrm{MC}}^{\mathrm{cont}}=1.091\), the propagated
bounds are
\begin{equation}
\begin{aligned}
\varepsilon_{\mathrm{tot}}^{\mathrm{id}}
&\leq 65.73\%,
&
\left|
\widehat{\mathrm{CVA}}_{\mathrm{AE}}^{\mathrm{id}}
-
\widehat{\mathrm{CVA}}_{\mathrm{MC}}^{\mathrm{cont}}
\right|
&\lesssim 0.717
\\[2pt]
\varepsilon_{\mathrm{tot}}^{\mathrm{hw}}
&\leq 68.88\%,
&
\left|
\widehat{\mathrm{CVA}}_{\mathrm{AE}}^{\mathrm{hw}}
-
\widehat{\mathrm{CVA}}_{\mathrm{MC}}^{\mathrm{cont}}
\right|
&\lesssim 0.752
\end{aligned}
\label{eq:aggregate_pipeline_error}
\end{equation}
These bounds make clear that the present end-to-end error is dominated by the
grid and encoding stages, while the amplitude estimation layer residual remains secondary even
under hardware-replay.

\section{Conclusions and future work}
\label{sec:conclusions_future_work}

This work has studied quantum amplitude estimation for Credit Valuation
Adjustment as an end-to-end workflow, not as an isolated query-complexity
primitive. The relevant object is the full map from market calibration,
finite-grid approximation and variational quantum encoding to amplified
finite-shot measurements and statistical inference. Practical quantum advantage
for CVA therefore requires both a CVA-faithful encoded oracle and amplified
circuits with enough contrast for Grover queries to remain statistically
informative.

The work makes three main contributions. First, it develops a general software
pipeline for quantum CVA estimation. Although the numerical study focuses on a
calibrated two-asset netting set, the implementation is modular: market
calibration, default and discount curves, finite-grid construction, QCBM-based
state preparation, controlled payoff rotations, amplitude-estimation routines
and diagnostics are reusable components. The software therefore goes beyond the
benchmark instance and provides a framework for studying correlated,
market-calibrated derivative portfolios.

Second, the work introduces contrast-aware Bayesian iterative quantum
amplitude estimation CABIQAE, a noise-aware extension of the Beta-BIQAE
framework of \textcite{Li2026BIQAE}. The contribution is not a universal
complexity claim that CABIQAE dominates all Bayesian QAE variants in every
noise model or amplitude regime; rather, the validation experiments show that,
in the hardware-replay regime studied here, CABIQAE achieves a more favourable
accuracy--cost--depth trade-off than the competing estimators. It improves on
the strongest noise-aware baseline, BAE \parencite{Ramoa2025bayesianquantum},
by retaining comparable noisy-replay accuracy with substantially lower
classical post-processing cost and shallower amplified circuits. For the full
CVA oracle, where unrestricted BAE becomes impractical as a production routine,
CABIQAE is the most effective estimator in the benchmarked pipeline: it avoids
the noise-induced degradation of BIQAE, outperforms direct circuit sampling at
comparable query budgets and uses the limited contrast-preserving amplification
window more efficiently.

Third, the work carries the quantum CVA estimate beyond exact statevector
evaluation. In the pioneering circuit-level construction of
\textcite{Alcazar2021CVA}, the CVA observable was formulated as a projector
expectation and the final quantum value was obtained through noiseless
simulation and resource modelling. Here, the trained CVA oracle is evaluated
through finite-shot amplitude-estimation trajectories: in the ideal regime from
the exact amplified probability law, and in hardware-replay from empirical
amplified-circuit responses measured on IBM quantum hardware. Hardware-replay
does not mean that every adaptive trajectory was executed live on the device,
but that real hardware measurement data are reused for robust \(R=300\)
trajectory-level comparisons without prohibitive quantum-compute cost. Thus,
the work does not only report a statevector CVA amplitude, but studies CVA as
a measured quantity with error-versus-query curves, coverage diagnostics and
depth-dependent hardware response.

The error budget identifies the present bottlenecks. The continuous Monte Carlo
benchmark is \(\widehat{\mathrm{CVA}}_{\mathrm{MC}}^{\mathrm{cont}}=1.091\),
while the six-qubit finite-grid benchmark is
\(\widehat{\mathrm{CVA}}_{\Delta}^{\mathrm{tab}}=0.522\), a \(52.15\%\)
discretisation and truncation gap. The trained QCBM/CRCA oracle gives
\(\widehat{\mathrm{CVA}}_{\mathrm{SV}}(\Theta)=0.670\), corresponding to a
\(38.59\%\) realised discrepancy and a conservative structural bound of
\(65.72\%\). By comparison, the CABIQAE amplitude-estimation residual is only
\(0.0139\%\) noiseless and \(3.16\%\) under hardware-replay. Thus, the limiting factor is not
statistical estimation of the encoded amplitude, but constructing an encoded
observable whose expectation remains faithful to the continuous CVA functional. This distinction is important: the circuit approximates not merely a probability
distribution, but a risk functional weighted by positive exposure, discounting
and default increments, so low-probability regions may still matter when they
coincide with large exposure or default-weighted contribution. The finite-grid
convergence results identify the coarse \(n=4\) market register as the leading
source of bias, while training diagnostics make the exposure rotation the
hardest encoding component. Future oracle design should therefore be driven by
observable fidelity, not only by global distributional metrics such as KL
divergence or pointwise rotation losses.

The hardware results explain why the ideal amplified advantage does not survive
unchanged. In the reduced validation circuit, the calibrated contrast length is
\(\tau_{\mathrm c}=33.96\pm2.80\), whereas for the full CVA oracle it collapses
to \(\tau_{\mathrm c}=2.2\pm0.6\). This matches the amplified financial oracle's
compiled resources: already at \(k=1\), \(Q A_{\Theta}\) contains \(1122\)
transpiled two-qubit gates and reaches depth \(2691\), with the global
reflection \(S_0\) dominating the compiled depth. The obstruction is the
interaction between Grover amplification, non-local reflections and the
compiled CVA-oracle structure, not generic hardware noise alone.

These results give a cautious but constructive answer to the research question:
a coherent CVA circuit amplitude can be constructed for a calibrated multi-asset netting
set; ideal QAE recovers the expected amplified query behaviour on the trained
oracle; and contrast-aware Bayesian inference materially improves the
hardware-replay regime of the present implementation. This is not an end-to-end
quantum advantage for CVA; instead, the work makes the obstruction
measurable: financial discretisation and encoding dominate the current CVA bias,
while hardware contrast loss restricts the useful amplification window.

Future work should co-design the financial approximation, oracle and estimator.
First, discretisation and training objectives should become CVA-aware: adaptive
price grids, payoff-weighted losses, correlation-aware state preparation, and exposure-focused rotation training
should be evaluated by final-CVA error, not only by generic distributional or
circuit-level metrics. Higher market-register resolution can reduce bias, but
naively increases oracle depth and may destroy the contrast required by
amplitude estimation.

Second, the amplified oracle must be redesigned around hardware-native
structure. The most urgent target is the global reflection \(S_0\), with
possible directions including relative-phase constructions, layout-aware
synwork or approximate reflections. These changes should be assessed through
CABIQAE's likelihood-level contrast model, since a logical gate-count reduction
is valuable only if it preserves the statistical relation between Grover depth,
amplitude and observed success probability.

Third, CABIQAE should become a self-calibrating Bayesian estimator. Here, the
contrast scale is externally calibrated from the known ideal amplitude, suitable
for controlled benchmarking but not deployment. A production version should
jointly infer amplitude and contrast parameters, select Grover powers by
posterior utility, and include posterior predictive checks, model-evidence
comparisons or interleaved reference circuits to test contrast identifiability.
A theoretical complexity analysis under explicit contrast-decay assumptions
should also characterise the transition between the ideal \(O(1/N_q)\)
amplified regime and the hardware-observed contrast-limited saturation regime.

Finally, the methodology should move from hardware-replay to larger direct
hardware experiments as devices improve. Hardware-replay is right for robust
present day statistics, but future experiments should execute adaptive
trajectories directly when queueing, calibration drift and quantum-compute
budgets permit. Financially, the framework should be tested on larger netting
sets, richer exposure profiles, wrong-way risk or stochastic rates to assess stability for production-like
CVA oracles.

Overall, this work shows that practical quantum CVA cannot be validated by
statevector oracle evaluation or QAE asymptotics alone. It requires a
quantitative error chain linking financial modelling error, variational encoding
error, amplified-circuit resources and hardware-calibrated inference. Within
that chain, CABIQAE improves the amplitude estimation stage of the pipeline by
incorporating realistic contrast loss into both Bayesian inference and
Grover-depth selection. The next step is to make the encoded CVA observable and amplified oracle
equally reliable; only with lower-error devices and/or longer
contrast-preserving circuit depths can the theoretical advantage of quantum
amplitude estimation become a credible numerical advantage for
counterparty-credit valuation.

\phantomsection
\section*{Author contributions}
\addcontentsline{toc}{section}{Author contributions}
All authors contributed to this work. Artificial intelligence tools were used to refine some early versions of this manuscript, including aspects of language, formatting and manuscript structure. These tools did not determine the scientific content, analysis, conclusions or interpretation of the results. Authorship of the work belongs to the authors, and only to the authors, who reviewed the manuscript and take full responsibility for its final content.

\phantomsection
\section*{Acknowledgements}
\addcontentsline{toc}{section}{Acknowledgements}
The authors would like to express their sincere gratitude to Fernando de Lope Contreras for his invaluable guidance, support and involvement throughout the development of this work. His technical insight, availability and commitment were instrumental in addressing many of the challenges encountered during the project.

The authors are also especially grateful to Arantza Gorostiaga for her continued support, trust in the project and assistance in enabling its development from an initial proposal into a fully realised research study.

We further acknowledge the Basque Quantum team for their collaboration and for facilitating access to quantum hardware. Their support was essential to the experimental component of this work.

Finally, the authors thank the technical teams at IBM and Q-CTRL for their assistance in resolving the issues encountered during the execution of the quantum hardware experiments.

\newpage

\begingroup
\normalsize
\printbibliography[heading=bibintoc,title={References}]
\endgroup
\clearpage

\newgeometry{left=1.5cm,right=1.5cm,top=1.5cm,bottom=1.5cm}
\fancyhfoffset{0pt}
\setlength{\headwidth}{\textwidth}
\appendix
\section{Algorithmic and implementation details}
\label{app:classical_benchmark_details}

\subsection{Market data and numerical controls}
\label{appsec:classical_market_inputs}
\begin{table}[htbp]
\centering
\normalsize
\setlength{\tabcolsep}{3pt}
\renewcommand{\arraystretch}{0.95}
\caption{Market data and risk-driver inputs used in the classical CVA benchmark. The full set of OIS pillars is shown for market-data completeness. In the CVA calculation, however, discounting is performed with a flat continuously compounded proxy rate calibrated from the one-year discount factor.}
\label{tab:classical_market_inputs_condensed}
\begin{tabularx}{\textwidth}{@{}p{2.7cm}p{5.0cm}X@{}}
\toprule
\textbf{Block} & \textbf{Input} & \textbf{Numerical details / role} \\
\midrule

Equity data &
Historical closes, spots and dividends &
Closes from 13-Mar-2021 to 13-Mar-2026 estimate log-return correlations. 
Spots/dividends: EURO STOXX 50 \(5716.61\), \(2.88\%\); SMI \(14214.72\), \(3.16\%\).\\

Volatility data &
ATM (at-the-money) implied-volatility slices &
EURO STOXX 50 surface: 35 expiries, 2026--2035; SMI surface: 18 expiries, 2026--2030. ATM slices define piecewise-constant volatility buckets on the exposure dates. \\

Dependence &
Retained correlation matrix &
Two-asset portfolio matrix with unit diagonal and
\(\rho_{\mathrm{SX5E},\mathrm{SMI}}=0.6872\). \\

Interest rates &
EUR OIS curve \((\mathrm{EURESTOIS})\) &
Pillars 7D--60Y; mid rates from \(1.94\%\) to \(2.73\%\); discount factors from \(0.9996\) to \(0.3754\). Used to construct the benchmark discount curve. \\

Credit data &
Iberdrola CDS par-spread curve &
Tenors 6M--30Y; spreads from \(9.46\) bps to \(120.37\) bps. Used to bootstrap the deterministic CDS-implied survival curve. \\

Recovery &
CDS and CVA recovery assumptions &
CDS bootstrap uses \(R_{\mathrm{CDS}}=40\%\); CVA loss calculation uses \(R_{\mathrm{CVA}}=41.5\%\). \\

\bottomrule
\end{tabularx}
\end{table}

\label{appsec:classical_benchmark_controls}

\begin{table}[H]
\normalsize
\centering
\normalsize
\caption{Numerical controls for the continuous-underlying benchmark.}
\label{tab:classical_benchmark_controls_condensed}
\begin{tabularx}{\textwidth}{
@{}
>{\hsize=0.7\hsize\raggedright\arraybackslash}X
>{\hsize=0.9\hsize\raggedright\arraybackslash}X
>{\hsize=1.4\hsize\raggedright\arraybackslash}X
@{}}
\toprule
\textbf{Component} & \textbf{Implementation choice} & \textbf{Details} \\
\midrule
Valuation setup &
Valuation date and time grid &
Valuation date: 13-Mar-2026. CVA maturity: \(T=6\) months. Exposure dates are uniformly placed on \((0,T]\), excluding \(t=0\). Number of exposure intervals: \(M=4\). \\
\addlinespace

Market data &
Input blocks &
Discount curve, Iberdrola CDS data, historical equity series, deterministic dividend yields, and ATM volatility-surface slices. \\
\addlinespace

Equity dynamics &
Volatility and dependence &
Piecewise-constant ATM volatility grid on the exposure dates and empirical log-return correlation matrix from historical data. \\
\addlinespace

Discounting &
Flat continuously compounded curve &
Benchmark discounting uses \(D(0,t)=e^{-rt}\), with constant \(r\). The risk-free rate used in instrument MTM is \(r=-\log D(0,1)\). \\
\addlinespace

Default model &
CDS-implied deterministic piecewise-constant survival curve &
Survival probabilities are bootstrapped from CDS par spreads with quarterly payments. The CVA recovery rate is \(R_{\mathrm{CVA}}=0.415\). \\
\addlinespace

Monte Carlo controls &
Sampling and reproducibility &
The benchmark uses \(N_{\mathrm{MC}}=2\times10^5\) base normal draws and random seed $105$. \\
\bottomrule
\end{tabularx}
\end{table}

\begin{table}[H]
\centering
\small
\setlength{\tabcolsep}{2.5pt}
\renewcommand{\arraystretch}{1.02}

\caption{Numerical controls and hyperparameter configuration used in the amplitude estimation experiments. The left panel corresponds to the validation experiment in Section \ref{subsec:noise_aware_ae_results}, whereas the right panel reports the quantum CVA experiment in Section \ref{subsec:quantum_cva_ae_results}. In the multi-asset CVA experiment, the BIQAE target half-width is set smaller than that of CABIQAE to obtain comparable query budgets, since the noise-naive BIQAE scheduler closes its internal interval prematurely by treating noisy high-\(K\) circuits as informative. For the meaning of the CABIQAE/BIQAE parameters, see
Appendix~\ref{app:cabiqae_modules}; for the BAE-specific hyperparameters, see
  \textcite{Ramoa2025bayesianquantum}.}
\label{tab:ae_hyperparameters}

\begin{tabular}{@{}lccc@{\qquad}cc@{}}
\toprule
&
\multicolumn{3}{c}{\textbf{Validation}} &
\multicolumn{2}{c}{\textbf{Quantum CVA}} \\
\cmidrule(lr){2-4}
\cmidrule(l){5-6}
\textbf{Control / hyperparameter}
& \makecell[c]{CABIQAE}
& BIQAE
& BAE
& \makecell[c]{CABIQAE}
& BIQAE \\
\midrule

Target half-width, \(\varepsilon\)
    & \(10^{-3}\) & \(10^{-3}\) & \(10^{-3}\)
    & \(10^{-2}\) & \(10^{-3}\) \\

Failure probability, \(\alpha\)
    & \(0.10\) & \(0.10\) & \(0.10\)
    & \(0.10\) & \(0.10\) \\

Shots per circuit batch, \(n_{\mathrm{batch}}\)
    & \(256\) & \(256\) & \(256\)
    & \(128\) & \(128\) \\

Validation amplitudes, \(a_{\mathrm{true}}\)
    & \makecell[c]{Noiseless: \(0.36027,0.49217\),\\
      \(0.53032,0.56824\),\\
      \(0.60567,0.64237\),\\
      \(0.66927,0.70413\);\\
      Hardware: \(0.36027\)}
    & Same & Same
    & \textemdash & \textemdash \\

Credible-interval method
    & Beta & Beta & \textemdash
    & Beta & Beta \\

Minimum stage-growth ratio, \(\rho_{\min}\)
    & \(2\) & \(2\) & \textemdash
    & \(2\) & \(2\) \\

Contrast-decay scale, \(\tau_{\mathrm c}\)
    & \(33.96\) & \textemdash & \(33.96\)
    & \(2.20\) & \textemdash \\

Noise floor, \(b\)
    & \(0.5\) & \textemdash & \(0.5\)
    & \(0.15\) & \textemdash \\

Contrast prefactor (noiseless/hardware), \(c_0\)
    & \(1.0/1.16\) & \textemdash & \textemdash
    & \(1.0/0.67\) & \textemdash \\

Prior-transport samples, \(N_{\mathrm{prior}}\)
    & \(2000\) & \textemdash & \textemdash
    & \(2000\) & \textemdash \\

Online estimation of \(\tau_{\mathrm c}\)
    & No & \textemdash & No
    & No & \textemdash \\

BAE warm-up shots, \(wN_s\)
    & \textemdash & \textemdash & \(256\) \\

BAE particles, \(N_{\mathrm{part}}\)
    & \textemdash & \textemdash & \(300\) \\

BAE resampling threshold, \(\mathrm{thr}\)
    & \textemdash & \textemdash & \(0.4\) \\

BAE initial control, \(k_0\)
    & \textemdash & \textemdash & \(1\)  \\

BAE refinement parameters, \((e_{\mathrm{refs}},e_{\mathrm{thr}})\)
    & \textemdash & \textemdash & \((1,1)\) \\

BAE stochastic control
    & \textemdash & \textemdash & False \\

\bottomrule
\end{tabular}
\end{table}

\subsection{Quantum hardware and error-mitigation setup}
\label{app:hardware}

The final hardware experiments reported in this work were executed on
\texttt{ibm\_basquecountry}, the 156-qubit BasQ--IBM Quantum System
Two QPU installed in Donostia--San Sebastián and equipped with an IBM Heron r2
superconducting processor based on the heavy-hexagonal architecture \parencite{IBMQuantumProcessorTypes2026}.
Access to this backend was enabled by Basque Quantum, whose support was
essential for the hardware experiments reported in this work. Before selecting
the final backend, exploratory tests were also carried out on the IBM heavy-hex
backends \texttt{ibm\_aachen}, \texttt{ibm\_pittsburgh} and
\texttt{ibm\_boston}. The observed stability and results made
\texttt{ibm\_basquecountry} the best option for the reported validation and CVA
experiments.

For the CVA amplitude estimation experiment, \texttt{ibm\_basquecountry} was
accessed through Q-CTRL's Performance Management Qiskit
Function \parencite{QCTRLPerformanceManagement2026}. This function was used as a
managed hardware-execution layer designed to improve the reliability of the hardware samples returned by the
processor. Operationally, the function accepts abstract, untranspiled circuits
and applies an automated error-suppression workflow internally, including
hardware-aware layout selection, routing, dynamical decoupling and
noise-suppression techniques. Similar deterministic error-suppression workflows
underlying Q-CTRL's Fire Opal framework have been experimentally benchmarked on
IBM hardware \parencite{Mundada2023ErrorSuppression}. 

This distinction is important for amplitude estimation because the Grover-power
circuits used in the experiments are very deep (Table \ref{tab:cva_amplified_logical_transpiled_resources}), especially as the amplification
power \(k\) increases. The Q-CTRL workflow was
therefore used on \texttt{ibm\_basquecountry} to make the hardware execution of
the amplified CVA circuits as stable as possible, while keeping the algorithmic
comparison between BIQAE and CABIQAE unchanged. 

\paragraph{Heavy-hex compatibility.}
The variational circuits introduced in Section~\ref{subsec:quantum_pipeline}
were designed to be \emph{heavy-hex friendly} in a strict graph-theoretic
sense. Let \(G_{\mathrm{HH}}=(V_{\mathrm{HH}},E_{\mathrm{HH}})\) be the
coupling graph of the target backend and
\(G_{\mathrm{var}}=(V_{\mathrm{var}},E_{\mathrm{var}})\) the logical
interaction graph of a variational block. The block is heavy-hex friendly if
there exists an injective map
\(\phi: V_{\mathrm{var}}\to V_{\mathrm{HH}}\) such that
\begin{equation}
    \{u,v\}\in E_{\mathrm{var}}
    \;\Longrightarrow\;
    \{\phi(u),\phi(v)\}\in E_{\mathrm{HH}}.
    \label{eq:heavy_hex_embedding}
\end{equation}
When \eqref{eq:heavy_hex_embedding} holds for every entangling layer,
the transpiler can realise that layer with zero SWAP gate overhead. This criterion
is operationally decisive for the amplitude estimation experiments: every
avoidable SWAP inserted at the base-circuit level is replicated inside each
\(Q^k A\) application, generating a multiplicative depth penalty that
directly degrades the calibratable contrast.

\begin{table}[htbp]
\centering
\normalsize 
\caption{Heavy-hex compatibility of the three variational blocks.}
\label{tab:heavyhex_compatibility}
\setlength{\tabcolsep}{5pt}
\renewcommand{\arraystretch}{1.15}
\begin{tabularx}{\textwidth}{@{} l c c X @{}}
\toprule
\textbf{Block} & \textbf{Logical qubits}
& \textbf{Max.\ degree} & \textbf{Embedding feature} \\
\midrule
CRCA \(R_p, R_q\) & 3 & 2
& 3-node path graph; embeds into any heavy-hex path of
length two with zero routing. \\
\addlinespace
CRCA \(R_v\) & 7 & 3
& Logical star realised through outward/inward snake layers
over a connected 7-qubit patch; ancilla physical degree
never exceeds three. \\
\addlinespace
QCBM \(G_{\mathcal{P}}\) & 6 & 3
& Acyclic \(ZZ\)-entangler graph with
\(|E_{\mathrm{QCBM}}|=5\); embeds into a connected
6-qubit heavy-hex subgraph with zero routing. \\
\bottomrule
\end{tabularx}
\end{table}

Table~\ref{tab:heavyhex_compatibility} summarises the compatibility
properties of the three variational blocks used in the work.
The three-qubit CRCA blocks for \(R_p\) and \(R_q\)
(Figure~\ref{fig:crca_default_probabilities_native_tree_ansatz}) act on the
two-qubit time register \((t_0, t_1)\) plus one ancilla \(a\). Their logical
interaction graph is a 3-node path \(t_0 - a - t_1\), a strict subgraph of
any heavy-hex 3-path; no dense couplings are required. Sparse connectivity
is appropriate because the target functions \(\tilde{p}_i\) and
\(\tilde{q}_i\) are 4-point deterministic maps: richer entanglement would
increase two-qubit exposure without increasing representational capacity.

The seven-qubit CRCA for \(R_v\)
(Figure~\ref{fig:crca_heavy_hex_star_two_subcircuits}) must
approximate a non-separable map over all six computational qubits
\((t_0, t_1, s_0, s_1, s_2, s_3)\) and one ancilla. A literal
ancilla-centred star with simultaneous degree-6 access is not physically
realisable on Heron hardware (\(\max\deg G_{\mathrm{HH}} = 3\)). The
adopted architecture resolves this through alternating outward and inward
snake layers: information is transported sequentially across the 7-qubit
patch, expanding the ancilla's effective receptive field through
nearest-neighbour interactions without violating the degree bound. This
converts a routing problem into a controlled depth cost, which is preferable
because snake-layer depth scales linearly and predictably with register size.

The six-qubit QCBM for \(G_{\mathcal{P}}\)
(Figure~\ref{fig:qcbm_heavyhex6_two_layer_pattern}) uses local
\(R_x, R_z\) rotations followed by a sparse \(ZZ\)-entangler layer with
interaction graph
\begin{equation}
    E_{\mathrm{QCBM}}
    =
    \bigl\{
    (t_0,t_1),\,(t_0,s_0),\,(t_0,s_1),\,(t_1,s_2),\,(s_0,s_3)
    \bigr\}.
    \label{eq:qcbm_entangler_graph}
\end{equation}

The graph is connected, acyclic, and has maximum degree three, admitting
a direct embedding into a 6-qubit heavy-hex subgraph with zero routing.
Equation~\eqref{eq:qcbm_entangler_graph} couples the time register
directly to both asset sub-registers, enabling cross-asset temporal
correlations without introducing interactions foreign to the hardware
topology.  

\begin{figure}[H]
    \centering
    \includegraphics[width=0.7\textwidth]{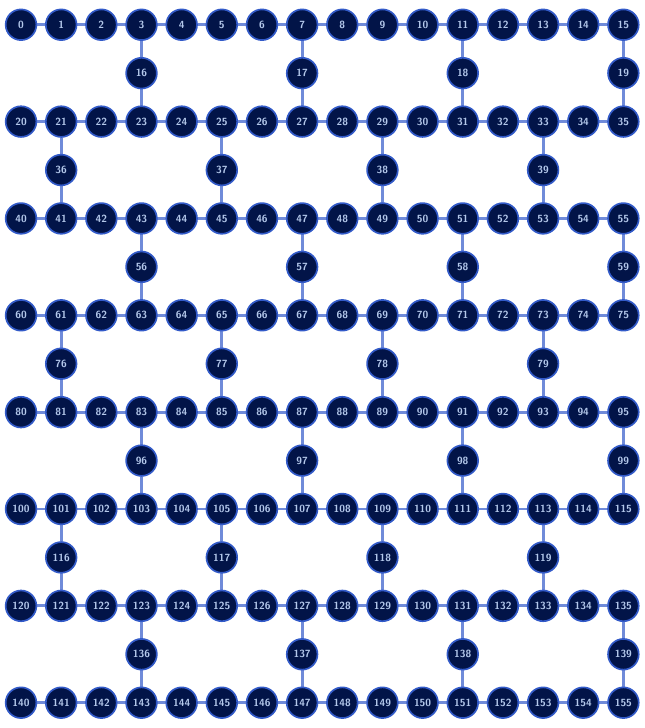}
    \caption{Heavy-hex connectivity map of the Heron-class backend. Every
    physical qubit has degree at most three. The circuit families used in
    this work were designed so that their logical entangler graphs embed
    into connected 3-, 6-, and 7-qubit subgraphs of this coupling map,
    eliminating routing overhead at the block level.}
    \label{fig:heavyhex_connectivity_map}
\end{figure}

\paragraph{Compilation protocol.}
For all direct experiments on \texttt{ibm\_basquecountry}, circuits were
compiled to the Heron ISA using fixed-layout transpilation at optimisation
level~3 with multi-seed selection, targeting a connected subgraph that
admits the embedding \eqref{eq:heavy_hex_embedding} with no SWAP
insertions. For Q-CTRL experiments on \texttt{ibm\_basquecountry}, abstract
circuits were submitted directly to the Performance Management Qiskit
Function, so the final ISA mapping, pulse-level scheduling choices and
error-suppression details were handled internally by the managed workflow.
Under this protocol, the ISA depth and two-qubit counts reported in
Tables~\ref{tab:atest_amplified_logical_transpiled_resources}
and~\ref{tab:cva_amplified_logical_transpiled_resources} are exact properties
of the directly compiled \texttt{ibm\_basquecountry} circuits; the corresponding
quantities for Q-CTRL-managed executions are internal to the pipeline and not
directly observable.

\subsection{CABIQAE algorithm and implementation details}
\label{app:cabiqae_modules}
CABIQAE follows
the Beta-BIQAE credible-interval--transport--schedule loop of \textcite{Li2026BIQAE}, but all
probabilistic transformations are performed through the calibrated
hardware-observation model
\[
    p_k(\vartheta)
    :=
    p_{\mathrm{obs}}(\vartheta,k)
    =
    b+c_k\left(\sin^2(K\vartheta)-b\right),
    \qquad
    K=2k+1,
    \qquad
    c_k=c_0e^{\!\left(-K/\tau_{\mathrm c}\right)},
    \label{eq:appendix_cabiqae_pobs}
\]
with \(\vartheta\in[0,\pi/2]\). The target amplitude is
\(a=\sin^2\vartheta\). The integer \(k\) denotes the Grover power, while
\(K=2k+1\) is the corresponding number of the state preparation unitary queries per circuit execution.

\begin{algorithm}[htbp]
\caption{Contrast-aware Bayesian iterative quantum amplitude estimation}
\label{alg:cabiqae_main}
\normalsize
\algrenewcommand\algorithmicindent{0.9em}

\begin{algorithmic}[1]

\Statex \textbf{Parameters:}
\(\varepsilon,\alpha,n_{\mathrm{batch}},(c_0,\tau_{\mathrm c},b),\rho_{\min}>1,N_{\mathrm{prior}}\).
\Statex \textbf{Input:}
state-preparation unitary $A$, Grover operator $Q$.
\Statex \textbf{Output:}
amplitude estimate \(\widehat a\), credible interval \([a_l,a_u]\).

\State Initialise
\(t=0\), \(k_0=0\), \(I_0=[0,\pi/2]\), \([a_l,a_u]=[0,1]\),
\(\alpha_{0,0}=\beta_{0,0}=1/2\), and \((N_0,Y_0)=(0,0)\).

\While{\(a_u-a_l>2\varepsilon\)}

    \State Run \(\mathcal Q^{k_t}\mathcal A\) for \(n_{\mathrm{batch}}\) shots;
    set \(N_t\leftarrow N_t+n_{\mathrm{batch}}\) and
    \(Y_t\leftarrow Y_t+\#\{\mathrm{good\ outcomes}\}\).

    \State
    \((\alpha_{\mathrm{post},t},\beta_{\mathrm{post},t})
    \leftarrow
    \textsc{BayesianUpdate}(\alpha_{0,t},\beta_{0,t},N_t,Y_t)\).

    \State
    \((I_t^{\mathrm{new}},[a_l,a_u],\pi_t)
    \leftarrow
    \textsc{ComputeCRI}(I_t,k_t,\alpha_{\mathrm{post},t},
    \beta_{\mathrm{post},t},\alpha)\).

    \If{\(a_u-a_l\leq2\varepsilon\)}
        \State \Return (\(\widehat a=(a_l+a_u)/2\), \([a_l,a_u]\)).
    \EndIf

    \State
    \(k_{\mathrm{next}}
    \leftarrow
    \textsc{FindNextK}(I_t^{\mathrm{new}},k_t,\tau_{\mathrm c},
    \rho_{\min},\pi_t)\).

    \If{\(k_{\mathrm{next}}>k_t\)}
        \State
        \((\alpha_{0,t+1},\beta_{0,t+1})
        \leftarrow
        \textsc{PreparePrior}(\pi_t,k_{\mathrm{next}},N_{\mathrm{prior}})\).

        \State Set
        \(k_{t+1}=k_{\mathrm{next}}\),
        \(I_{t+1}=I_t^{\mathrm{new}}\),
        \((N_{t+1},Y_{t+1})=(0,0)\),
        and \(t\leftarrow t+1\).
    \Else
        \State Set \(I_t\leftarrow I_t^{\mathrm{new}}\) and continue sampling
        at the same Grover power.
    \EndIf

\EndWhile

\State \Return \(\widehat a=(a_l+a_u)/2\), \([a_l,a_u]\).

\end{algorithmic}
\end{algorithm}

The implementation keeps, at each stage \(t\), an identifiable interval
\(I_t\subset[0,\pi/2]\), a Grover power \(k_t\), cumulative counts
\((N_t,Y_t)\), and a Beta prior $p_t\sim\mathrm{Beta}(\alpha_{0,t},\beta_{0,t})$, with $p_t(\vartheta):=p_{k_t}(\vartheta)$. The prior is placed on the observed success probability \(p_t\), not on the
ideal probability \(\sin^2(K_t\vartheta)\). This is the point at which
CABIQAE differs from ideal Beta-BIQAE.

\paragraph{BayesianUpdate.}
At stage \(t\), let \(X_{t,r}\in\{0,1\}\) be the indicator of a good-state
measurement outcome and
\[
    Y_t=\sum_{r=1}^{N_t}X_{t,r}.
\]
Conditionally on the latent angle and current Grover power,
\[
    X_{t,r}\mid\vartheta,k_t
    \sim
    \mathrm{Bernoulli}(p_t(\vartheta)),
    \qquad
    Y_t\mid\vartheta,k_t
    \sim
    \mathrm{Binomial}(N_t,p_t(\vartheta)).
\]
Using the Beta prior on the observed success probability, $p_t\sim\mathrm{Beta}(\alpha_{0,t},\beta_{0,t})$,
the conjugate posterior is 
\begin{equation}
p_t\mid Y_t
    \sim
    \mathrm{Beta}(\alpha_{\mathrm{post},t},\beta_{\mathrm{post},t})
    \label{eq:appendix_bayesian_update}
\end{equation}
with $\alpha_{\mathrm{post},t}
    =
    \alpha_{0,t}+Y_t$,
    $\beta_{\mathrm{post},t}
    =
    \beta_{0,t}+N_t-Y_t$.
\paragraph{ComputeCRI.}
The posterior in \eqref{eq:appendix_bayesian_update} lives in observed
probability space. To obtain an interval for the amplitude, CABIQAE pulls this
posterior back to the latent angle \(\vartheta\).

Let
\[
    I_t=[\vartheta_l^{(t)},\vartheta_u^{(t)}]
\]
be the current identifiable branch.\footnote{See \textcite{Grinko2021IQAE} for a more detailed explanation of the identifiability problem and tracking.} The branch is chosen so that
\(K_t I_t\) lies inside a single interval
\[
    \left[
    \ell\frac{\pi}{2},
    (\ell+1)\frac{\pi}{2}
    \right],
\]
for some integer \(\ell\). On such an interval, \(p_t(\vartheta)\) is treated
as a single-branch transformation of the latent angle. The unnormalised
pullback density is
\[
    \widetilde{\pi}_t(\vartheta)
    =
    f_{\mathrm{Beta}}\!\left(
        p_t(\vartheta);
        \alpha_{\mathrm{post},t},\beta_{\mathrm{post},t}
    \right)
    \left|
        p_t'(\vartheta)
    \right|
    \mathbf{1}_{I_t}(\vartheta),
    \label{eq:appendix_pullback_density}
\]
where
\begin{equation}\label{eq:appendix_pobs_derivative}
        p_t'(\vartheta)
    =
    c_{k_t}K_t\sin(2K_t\vartheta)
\end{equation}
After normalisation,
\[
    Z_t
    =
    \int_{I_t}\widetilde{\pi}_t(u)\,du,
    \qquad
    \pi_t(\vartheta)
    =
    Z_t^{-1}\widetilde{\pi}_t(\vartheta).
\]
The CDF of the latent posterior is
\[
    F_t(x)
    =
    \int_{\vartheta_l^{(t)}}^x\pi_t(u)\,du.
\]
The equal-tailed credible interval at level \(1-\alpha\) is
\begin{equation}\label{eq:appendix_latent_cri}
    I_t^{\mathrm{new}}
    =
    [\vartheta_l,\vartheta_u]
    =
    \left[
    F_t^{-1}\!\left(\frac{\alpha}{2}\right),
    F_t^{-1}\!\left(1-\frac{\alpha}{2}\right)
    \right].
\end{equation}

Since \(a=\sin^2\vartheta\) is monotone on \([0,\pi/2]\), the corresponding
amplitude interval is
\begin{equation}\label{eq:appendix_amplitude_interval}
    [a_l,a_u]
    =
    [\sin^2\vartheta_l,\sin^2\vartheta_u].
\end{equation}

\paragraph{FindNextK.}
The scheduler searches for a larger Grover power that remains identifiable
and is expected to be informative under
the current posterior.

Given the updated interval \(I_t^{\mathrm{new}}\), a candidate \(k\) is
admissible if
\begin{equation}\label{eq:appendix_admissible_set_conditions}
    K I_t^{\mathrm{new}}
    \subset
    \left[
    \ell\frac{\pi}{2},
    (\ell+1)\frac{\pi}{2}
    \right]
    \quad\text{for some }\ell\in\mathbb{Z},
    \qquad
    \frac{K}{K_t}\geq\rho_{\min}.
\end{equation}
The first condition preserves single-branch identifiability, while the second
prevents negligible Grover-power increases.

For a Bernoulli observation at candidate Grover power \(k\), the log-likelihood
is
\[
    \ell(\vartheta;x,k)
    =
    x\log p_k(\vartheta)
    +
    (1-x)\log(1-p_k(\vartheta)).
\]
The corresponding Fisher information in the latent angle is
\begin{equation}\label{eq:appendix_fisher_information}
    \mathcal{I}_k(\vartheta)
    =
    \frac{\left[p_k'(\vartheta)\right]^2}
    {p_k(\vartheta)(1-p_k(\vartheta))}
    =
    \frac{
    c_k^2K^2\sin^2(2K\vartheta)
    }{
    p_k(\vartheta)(1-p_k(\vartheta))
    }.
\end{equation}
The theoretical scheduling score is the posterior expected Fisher information
per query,
\begin{equation}\label{eq:appendix_scheduler_score}
    \mathcal{S}_t(k)
    =
    \frac{1}{K}
    \int_{I_t^{\mathrm{new}}}
    \mathcal{I}_k(\vartheta)\pi_t(\vartheta)\,d\vartheta .
\end{equation}
CABIQAE selects
\begin{equation}\label{eq:appendix_find_next_k}
    k_{\mathrm{next}}
    =
    \arg\max_{k\in\mathcal{K}_t^{\mathrm{adm}}}
    \mathcal{S}_t(k),
\end{equation}
where \(\mathcal{K}_t^{\mathrm{adm}}\) is the set of admissible candidates
defined by \eqref{eq:appendix_admissible_set_conditions}. If this set is
empty, the algorithm keeps the same Grover power: $k_{\mathrm{next}}=k_t$.

\paragraph{PreparePrior.}
When \(k_{\mathrm{next}}>k_t\), the posterior must be transported to the
observed-probability space associated with the new Grover power. CABIQAE does
this in latent-angle space:
\begin{equation}\label{eq:appendix_prior_transport_samples}
    \vartheta^{(n)}\sim\pi_t,
    \qquad
    p^{(n)}
    =
    p_{\mathrm{obs}}(\vartheta^{(n)},k_{\mathrm{next}}),
    \qquad
    n=1,\ldots,N_{\mathrm{prior}}.
\end{equation}
The transported sample
\(\{p^{(r)}\}_{r=1}^{N_{\mathrm{prior}}}\) is approximated by a Beta distribution. Define
\[
    \widehat{\mu}
    =
    \frac{1}{N_{\mathrm{prior}}}\sum_{r=1}^{N_{\mathrm{prior}}} p^{(r)},
    \qquad
    \widehat{v}
    =
    \frac{1}{N_{\mathrm{prior}}}\sum_{r=1}^{N_{\mathrm{prior}}}
    \left(p^{(r)}-\widehat{\mu}\right)^2 .
\]
The unconstrained moment-matched concentration is
\[
    \widehat{\phi}
    =
    \frac{\widehat{\mu}(1-\widehat{\mu})}{\widehat{v}}-1.
\]
To avoid spuriously concentrated priors when the hardware contrast is small,
the fitted concentration is capped by the effective transported sample size
\begin{equation}\label{eq:appendix_concentration_cap}
    \phi
    =
    \min\left\{
        \widehat{\phi},
        c(k_{\mathrm{next}})^2N_{\mathrm{prior}}
    \right\}.
\end{equation}
The new Beta parameters are then
\begin{equation}\label{eq:appendix_transport_beta_parameters}
    \alpha_{0,t+1}
    =
    \max\{\widehat{\mu}\phi,\alpha_{\min}\},
    \qquad
    \beta_{0,t+1}
    =
    \max\{(1-\widehat{\mu})\phi,\beta_{\min}\}.
\end{equation}
If the moment fit is numerically ill-conditioned, for example if
\(\widehat v\leq0\) or
\(\widehat v\geq\widehat\mu(1-\widehat\mu)\), the implementation falls back to
a weakly informative Beta law. 

\paragraph{Numerical implementation.}
The latent posterior \(\pi_t\) is represented on a uniform grid
\(\{\vartheta_j\}_{j=1}^{J}\) over the current identifiable interval
\(I_t^{\mathrm{new}}\), with spacing
\(\Delta\vartheta=\vartheta_{j+1}-\vartheta_j\). The unnormalised posterior
values are denoted by \(\widetilde{\pi}_t(\vartheta_j)\). The normalisation
constant is computed by trapezoidal quadrature,
\[
    Z_t
    \simeq
    \sum_{j=1}^{J-1}
    \frac{
    \widetilde{\pi}_t(\vartheta_j)
    +
    \widetilde{\pi}_t(\vartheta_{j+1})
    }{2}
    \Delta\vartheta ,
\]
and the normalised grid density is
\[
    \pi_t(\vartheta_j)
    =
    Z_t^{-1}\widetilde{\pi}_t(\vartheta_j).
\]
Equivalently, introducing the trapezoidal weights
\[
    w_1=w_J=\frac{\Delta\vartheta}{2},
    \qquad
    w_j=\Delta\vartheta,
    \quad j=2,\ldots,J-1,
\]
the grid posterior satisfies
\[
    \sum_{j=1}^{J} w_j\pi_t(\vartheta_j)\simeq 1.
\]
The corresponding CDF is computed as
\[
    F_t(\vartheta_j)
    \simeq
    \sum_{r=1}^{j-1}
    \frac{
    \pi_t(\vartheta_r)
    +
    \pi_t(\vartheta_{r+1})
    }{2}
    \Delta\vartheta .
\]
Quantiles and posterior samples are obtained by linear interpolation of this
discrete CDF. Probabilities passed to Beta densities are clipped to
\([\delta,1-\delta]\), invalid normalisations are replaced by a uniform
density on the current interval, and Beta shape parameters are kept strictly
positive.

The numerical scheduler evaluates the posterior expected Fisher information
per query in \eqref{eq:appendix_scheduler_score} on the same grid. The
stabilised score used in the experiments is
\begin{equation}\label{eq:appendix_grid_scheduler_score}
    \mathcal{S}^{\mathrm{grid}}_t(k)
    =
    \frac{\gamma_k}{K}
    \sum_{j=1}^{J}
    w_j\,
    \pi_t(\vartheta_j)
    \frac{
    \left[p_k'(\vartheta_j)\right]^2
    }{
    \max\{p_k(\vartheta_j)(1-p_k(\vartheta_j)),\delta\}
    },
    \qquad
    \gamma_k
    =
    \sum_{j=1}^{J}
    w_j\,
    \pi_t(\vartheta_j)
    \left|\sin(2K\vartheta_j)\right|.
\end{equation}
Here \(\delta>0\) prevents division by zero, while \(\gamma_k\) is a
posterior-weighted flatness penalty for candidate Grover powers whose
amplified branch is locally uninformative over the region where the current
posterior places mass. Without the factor \(\gamma_k\) and the clipping by
\(\delta\), \(\mathcal{S}^{\mathrm{grid}}_t(k)\) is the direct trapezoidal
approximation of the theoretical score in
\eqref{eq:appendix_scheduler_score}.

\section{Quantum amplitude estimation and error metrics}
\label{appsection:quantum_mechanics}
\label{appsection:qae}

Quantum amplitude estimation (QAE), introduced by
\textcite{Brassard2000QAE}, is the primitive underlying the quadratic Monte
Carlo speedup used in this work. Given a state-preparation unitary \(A\) and
a projector \(\Pi_{\mathrm{good}}\) onto the standard \emph{good} (or marked)
subspace used in amplitude amplification, one can write
\begin{equation}
    A\ket{0}
    =
    \sqrt{1-a}\,\ket{\Psi_0}
    +
    \sqrt{a}\,\ket{\Psi_1},
    \qquad
    a=\bra{0}A^\dagger\Pi_{\mathrm{good}}A\ket{0},
    \label{eq:qae_state_decomposition_appendix}
\end{equation}
where, with a slight abuse of notation, \(\ket{\Psi_1}\) denotes the normalized component in the good subspace and
\(\ket{\Psi_0}\) the corresponding orthogonal component. Direct measurement of
\(A\ket{0}\) gives Bernoulli samples with success probability \(a\), so the
sample mean has standard error of order \(N^{-1/2}\).

QAE improves this by applying coherent amplitude amplification. Define
\(S_0:=I-2\ket{0}\bra{0}\), \(S_f:=I-2\Pi_{\mathrm{good}}\), and the Grover unitary
\(Q:=-AS_0A^\dagger S_f\). Writing \(a=\sin^2\theta\), with
\(\theta\in[0,\pi/2]\), the Grover unitary acts as a rotation in
\(\operatorname{span}\{\ket{\Psi_0},\ket{\Psi_1}\}\). After \(k\) Grover
applications,
\begin{equation}
    Q^kA\ket{0}
    =
    \cos((2k+1)\theta)\ket{\Psi_0}
    +
    \sin((2k+1)\theta)\ket{\Psi_1},
    \qquad
    p_k=\sin^2((2k+1)\theta),
    \label{eq:amplified_probability_appendix}
\end{equation}
up to global phases. Thus the unknown amplitude is inferred from an amplified
oscillation rather than from repeated unamplified Bernoulli sampling.

The original QAE algorithm estimates \(\theta\) by applying quantum phase
estimation to \(Q\), whose restriction to the two-dimensional invariant
subspace has eigenvalues \(e^{\pm 2i\theta}\). With a Fourier register of
order \(M\), controlled powers of \(Q\), and an inverse quantum Fourier
transform, it outputs \(\widetilde{a}=\sin^2(\widetilde{\theta})\);
equivalently, if the measured phase-register value is
\(y\in\{0,\ldots,M-1\}\), then
\(\widetilde{a}=\sin^2(\pi y/M)\).

The key guarantee is
\begin{equation}
    |\widetilde{a}-a|
    \leq
    \frac{2\pi\sqrt{a(1-a)}}{M}
    +
    \frac{\pi^2}{M^2}
    \label{eq:qae_brassard_bound_appendix}
\end{equation}
with probability at least \(8/\pi^2\), using \(\mathcal{O}(M)\) coherent
oracle queries \parencite{Brassard2000QAE}. Since
\(\sqrt{a(1-a)}\leq1/2\), this gives
\[
    |\widetilde{a}-a|
    =
    \mathcal{O}(M^{-1})
    =
    \mathcal{O}(N_q^{-1}),
\]
where \(N_q\) is the number of oracle queries. Hence, in the ideal
coherent-query model, QAE improves the direct-sampling scaling
\(\mathcal{O}(N_q^{-1/2})\) to the Heisenberg-limited scaling
\(\mathcal{O}(N_q^{-1})\). Equivalently, additive precision
\(\varepsilon\) requires \(\mathcal{O}(\varepsilon^{-1})\) queries instead of
the classical \(\mathcal{O}(\varepsilon^{-2})\) scaling
\parencite{Montanaro2015}.

\paragraph{Median errors and query-scaling exponents.}

The empirical curves in this work report median relative errors across
independent adaptive trajectories, rather than RMSEs. This is a
robustness choice: medians reduce sensitivity to rare unstable schedules and
finite-shot outliers. The following result makes precise that, in the ideal
statistical regime, this change of summary statistic preserves the same
query-scaling exponent.

For a non-negative random variable \(U\), define its median by
\[
    \operatorname{Med}(U)
    :=
    \inf\{t:\mathbb{P}(U\leq t)\geq1/2\}.
\]

\begin{proposition}[Median-error scaling]
\label{prop:median-error-scaling}
Let \(e_N=\hat{x}_N-x\). Suppose that, for some \(\alpha>0\),
\(N^\alpha e_N\Rightarrow Y\), where \(\mathbb{E}[Y^2]\in(0,\infty)\), and \(|Y|\) has a unique positive median.
If \(\{N^{2\alpha}e_N^2\}_{N\geq1}\) is uniformly integrable, then
\[
    \operatorname{RMSE}(\hat{x}_N)
    \sim
    N^{-\alpha}\sqrt{\mathbb{E}[Y^2]},
    \qquad
    \operatorname{Med}(|e_N|)
    \sim
    N^{-\alpha}\operatorname{Med}(|Y|).
\]
Thus RMSE and median absolute error have the same power-law exponent.
\end{proposition}

\begin{proof}
Set \(Z_N=N^\alpha e_N\). Then
\[
    \operatorname{RMSE}(\hat{x}_N)
    =
    N^{-\alpha}\bigl(\mathbb{E}[Z_N^2]\bigr)^{1/2}.
\]
Since \(Z_N\Rightarrow Y\), the continuous mapping theorem gives
\(Z_N^2\Rightarrow Y^2\). Uniform integrability implies
\(\mathbb{E}[Z_N^2]\to\mathbb{E}[Y^2]\), which proves the RMSE statement.
For the median, \(|Z_N|\Rightarrow|Y|\), and uniqueness of the median of
\(|Y|\) gives
\(\operatorname{Med}(|Z_N|)\to\operatorname{Med}(|Y|)\). Since positive
deterministic rescaling preserves quantiles exactly,
\[
    \operatorname{Med}(|e_N|)
    =
    N^{-\alpha}\operatorname{Med}(|Z_N|)
    \sim
    N^{-\alpha}\operatorname{Med}(|Y|).
\]
\end{proof}

For direct Monte Carlo,
\(\hat{\mu}_N=N^{-1}\sum_{i=1}^N X_i\), with
\(\operatorname{Var}(X_1)=\sigma^2\in(0,\infty)\), satisfies
\(\operatorname{RMSE}(\hat{\mu}_N)=\sigma/\sqrt{N}\). Moreover,
\(\sqrt{N}(\hat{\mu}_N-\mu)\Rightarrow\sigma Z\), with
\(Z\sim\mathcal{N}(0,1)\). Proposition~\ref{prop:median-error-scaling} gives
\[
    \operatorname{Med}(|\hat{\mu}_N-\mu|)
    \sim
    \frac{\sigma\Phi^{-1}(3/4)}{\sqrt{N}}
    \approx
    0.67449\,\frac{\sigma}{\sqrt{N}},
\]
so the median error preserves the classical \(N^{-1/2}\) exponent.

For QAE, the median statement follows even more directly from
\eqref{eq:qae_brassard_bound_appendix}. Since \(8/\pi^2>1/2\), the
high-probability bound holds at a probability level above the median level,
and therefore
\[
    \operatorname{Med}(|\widetilde{a}-a|)
    \leq
    \frac{2\pi\sqrt{a(1-a)}}{M}
    +
    \frac{\pi^2}{M^2}
    =
    \mathcal{O}(M^{-1})
    =
    \mathcal{O}(N_q^{-1}).
\]
Thus the median absolute error preserves the ideal QAE exponent.

The empirical plots use relative errors and deterministic CVA rescalings,
which do not affect exponents. If \(x\neq0\), then
\(|(\hat{x}_N-x)/x|=|x|^{-1}|\hat{x}_N-x|\). Likewise, if
\({\rm CVA}_\Delta=M(1-R_{\mathrm{CVA}})C_vC_pC_q\,a_{\rm CVA}\), then relative CVA error
and relative amplitude error coincide whenever \(a_{\rm CVA}>0\). Therefore,
\(\mathcal{O}(N_q^{-1/2})\) and \(\mathcal{O}(N_q^{-1})\) remain the
appropriate ideal Monte Carlo and QAE reference slopes for the median-error
curves. Deviations in hardware-replay should instead be attributed to contrast
loss, bias or noise saturation, not to the use of medians.

\paragraph{QFT-free implementations and CVA encoding.}

The canonical QAE construction attains the quadratic query scaling at the cost
of deep circuits involving controlled powers of \(Q\), a Fourier register and
an inverse quantum Fourier transform. Although optimal in the ideal
coherent-query model, it is poorly suited to near-term hardware. QFT-free
approaches, beginning with simplified quantum approximate counting
\parencite{AaronsonRall2020}, replace the phase-estimation register by
families of Grover-amplified circuits and classical post-processing. Iterative
and Bayesian variants such as IQAE, BIQAE and BAE retain the same
Grover-measurement structure while avoiding the explicit QFT-based circuit,
making them better aligned with NISQ constraints such as depth, calibration
stability and noise accumulation.

In the CVA setting of this work, \(A\) is the complete encoding circuit that
prepares the time--market distribution and applies the payoff, discount and
default rotations. Once this circuit has encoded the rescaled CVA as an
amplitude, amplitude estimation infers that amplitude, and the financial
quantity is recovered by multiplying the estimate by the deterministic
rescaling constants introduced in the main text.

\section{Extended results}

\subsection{Quantum CVA robustness analysis} \label{appsec:robust_analysis}
To assess the stability of the ideal statevector CVA pipeline, we perform a robustness test over 22 controlled market scenarios. The scenarios perturb the main financial inputs of the problem, including payoff strikes, volatilities, interest rates, and CDS-implied default spreads, while also including joint risk-on and risk-off stress configurations. Table~\ref{tab:ideal_cva_scenario_correspondence} summarises the correspondence between each scenario label and its associated perturbation. For each scenario, the exposure component is retrained in the ideal statevector regime and the resulting quantum CVA is compared against the corresponding classical small-grid benchmark. The relative errors and CVA values are reported in Table~\ref{tab:ideal_cva_robustness_results_compact}. Overall, the ideal quantum CVA encoding remains numerically stable in the robustness sense: across all 22 scenarios the error remains bounded, with a maximum relative error of \(85.4\%\), and no scenario exhibits divergence, sign instability, or an error exceeding one order of magnitude.

\begin{table}[H]
\normalsize
\centering
\caption{Ideal robustness-test CVA results. The tabulated benchmark is the scenario-specific classical small-grid CVA with \(n=4\), and the statevector value is evaluated using the ideal trained quantum pipeline.}
\label{tab:ideal_cva_robustness_results_compact}
\setlength{\tabcolsep}{3pt}
\renewcommand{\arraystretch}{1.12}
\setlength{\aboverulesep}{0pt}
\setlength{\belowrulesep}{0pt}
\makebox[\textwidth][c]{%
\begin{tabular}{@{}lrrrlrrrlrrr@{}}
\toprule
\multicolumn{4}{c}{\textbf{Scenarios S1--S8}} & \multicolumn{4}{c}{\textbf{Scenarios S9--S16}} & \multicolumn{4}{c@{}}{\textbf{Scenarios S17--S22}} \\
\cmidrule(lr){1-4}\cmidrule(lr){5-8}\cmidrule(l){9-12}
\textbf{Sc.} & \textbf{\(\mathrm{CVA}_{\mathrm{tab}}\)} & \textbf{\(\mathrm{CVA}_{\mathrm{SV}}\)} & \textbf{\(\varepsilon_{\mathrm{rel}}\)(\%)} & \textbf{Sc.} & \textbf{\(\mathrm{CVA}_{\mathrm{tab}}\)} & \textbf{\(\mathrm{CVA}_{\mathrm{SV}}\)} & \textbf{\(\varepsilon_{\mathrm{rel}}\)(\%)} & \textbf{Sc.} & \textbf{\(\mathrm{CVA}_{\mathrm{tab}}\)} & \textbf{\(\mathrm{CVA}_{\mathrm{SV}}\)} & \multicolumn{1}{r@{}}{\textbf{\(\varepsilon_{\mathrm{rel}}\)(\%)}} \\
\midrule
S1 & 0.522 & 0.670 & +28.4 & S9 & 0.547 & 0.669 & +22.3 & S17 & 0.469 & 0.604 & +28.7 \\
S2 & 0.812 & 0.916 & +12.9 & S10 & 0.501 & 0.681 & +35.9 & S18 & 0.574 & 0.738 & +28.7 \\
S3 & 0.312 & 0.484 & +55.4 & S11 & 0.484 & 0.692 & +43.2 & S19 & 0.652 & 0.839 & +28.7 \\
S4 & 0.574 & 0.686 & +19.6 & S12 & 0.519 & 0.667 & +28.6 & S20 & 0.571 & 0.675 & +18.3 \\
S5 & 0.414 & 0.611 & +47.7 & S13 & 0.520 & 0.669 & +28.5 & S21 & 0.470 & 0.654 & +39.2 \\
S6 & 0.960 & 0.997 & +3.9 & S14 & 0.522 & 0.673 & +28.9 & S22 & 0.419 & 0.646 & +54.4 \\
S7 & 0.232 & 0.430 & +85.4 & S15 & 0.523 & 0.675 & +29.1 &  &  &  &  \\
S8 & 0.604 & 0.685 & +13.5 & S16 & 0.391 & 0.504 & +28.7 &  &  &  &  \\
\bottomrule
\end{tabular}%
}
\end{table}

\begin{table}[H]
\normalsize
\centering
\caption{Scenario definitions for the ideal robustness analysis.}
\label{tab:ideal_cva_scenario_correspondence}
\setlength{\tabcolsep}{4pt}
\renewcommand{\arraystretch}{1.15}
\begin{tabular*}{\textwidth}{@{\extracolsep{\fill}} l l l l l l @{}}
\toprule
\textbf{Sc.} & \textbf{Perturbation} & \textbf{Sc.} & \textbf{Perturbation} & \textbf{Sc.} & \textbf{Perturbation} \\
\midrule
S1 & Baseline configuration & S8  & $\sigma\mapsto 0.80\sigma$ & S15 & $r(t)\mapsto r(t)+100\,\mathrm{bp}$ \\
S2 & $K_1\mapsto 0.90K_1$ & S9  & $\sigma\mapsto 0.90\sigma$ & S16 & $s_{\mathrm{CDS}}\mapsto 0.75s_{\mathrm{CDS}}$ \\
S3 & $K_1\mapsto 1.10K_1$ & S10 & $\sigma\mapsto 1.10\sigma$ & S17 & $s_{\mathrm{CDS}}\mapsto 0.90s_{\mathrm{CDS}}$ \\
S4 & $K_2\mapsto 0.90K_2$ & S11 & $\sigma\mapsto 1.20\sigma$ & S18 & $s_{\mathrm{CDS}}\mapsto 1.10s_{\mathrm{CDS}}$ \\
S5 & $K_2\mapsto 1.10K_2$ & S12 & $r(t)\mapsto r(t)-100\,\mathrm{bp}$ & S19 & $s_{\mathrm{CDS}}\mapsto 1.25s_{\mathrm{CDS}}$ \\
S6 & $K_1\mapsto 0.90K_1$; $K_2\mapsto 0.90K_2$ & S13 & $r(t)\mapsto r(t)-50\,\mathrm{bp}$ & & \\
S7 & $K_1\mapsto 1.10K_1$; $K_2\mapsto 1.10K_2$ & S14 & $r(t)\mapsto r(t)+50\,\mathrm{bp}$ & & \\
\midrule
S20 & \multicolumn{5}{l}{$K_1\mapsto 0.95K_1$; $K_2\mapsto 1.05K_2$; $\sigma\mapsto 0.90\sigma$; $r(t)\mapsto r(t)+50\,\mathrm{bp}$; $s_{\mathrm{CDS}}\mapsto 0.90s_{\mathrm{CDS}}$} \\
S21 & \multicolumn{5}{l}{$K_1\mapsto 1.05K_1$; $K_2\mapsto 0.95K_2$; $\sigma\mapsto 1.10\sigma$; $r(t)\mapsto r(t)-50\,\mathrm{bp}$; $s_{\mathrm{CDS}}\mapsto 1.10s_{\mathrm{CDS}}$} \\
S22 & \multicolumn{5}{l}{$K_1\mapsto 1.10K_1$; $K_2\mapsto 0.90K_2$; $\sigma\mapsto 1.20\sigma$; $r(t)\mapsto r(t)-100\,\mathrm{bp}$; $s_{\mathrm{CDS}}\mapsto 1.25s_{\mathrm{CDS}}$} \\
\bottomrule
\end{tabular*}
\end{table}

\subsection{Finite-grid convergence analysis}
\label{app:convergence_analysis}
\begin{table}[H]
\centering
\normalsize
\caption{Finite-grid CVA convergence and circuit-resource scaling as a function
of the market-register size \(n\), while keeping \(m=2\) fixed. Logical and ISA
resources are reported for the amplified block \(Q^{1}A\). The ISA resources
correspond to hardware-aware optimised transpilation to the backend used in
later CVA amplitude estimation experiments, \texttt{ibm\_basquecountry}. For
\(n=4\), the reported depths and two-qubit gate counts correspond to the
actual trained CVA circuit. For \(n\neq4\), no trained circuits are used; the
values are structural projections obtained by scaling the \(n=4\) hardware
metrics with same-family \(Q^{1}A\) proxy circuits.}
\label{tab:grid_cva_resource_scaling}
\setlength{\tabcolsep}{6pt}
\renewcommand{\arraystretch}{1.2}
\setlength{\aboverulesep}{0pt}
\setlength{\belowrulesep}{0pt}

\makebox[\textwidth][c]{%
\begin{tabular}{@{}
    c
    r
    r
    r
    r
    r
    r
@{}}
\toprule

\multicolumn{1}{@{}c}{}
& \multicolumn{2}{c}{\textbf{CVA Approx.}}
& \multicolumn{2}{c}{\textbf{Logical Resources}}
& \multicolumn{2}{c@{}}{\textbf{ISA Resources}} \\

\cmidrule{2-3} \cmidrule{4-5} \cmidrule{6-7}

\textbf{\(n\)}
& \textbf{\(\mathrm{CVA}_n\)}
& \textbf{$\varepsilon_{\mathrm{grid}}(n)(\%)$}
& \textbf{Depth}
& \textbf{2Q Gates}
& \textbf{Depth}
& \multicolumn{1}{r@{}}{\textbf{2Q Gates}} \\
\midrule

2  & 0.399 & 63.4 & 557   & 335  & 1430  & 598  \\
4  & 0.522 & 52.1 & 1226  & 702  & 2691  & 1122 \\
8  & 0.929 & 14.8 & 2508  & 1449 & 6015  & 2503 \\
12  & 1.046 & 4.0  & 4600  & 2517 & 10794 & 4293 \\
20 & 1.082 & 0.7  & 10676 & 5421 & 24474 & 9101 \\

\bottomrule
\end{tabular}%
}
\end{table}

\subsection{Amplified-circuit resource tables}
\label{appsec:amplified_resource_tables}

The ISA depths and two-qubit gate counts reported for the CVA circuits (Table \ref{tab:cva_amplified_logical_transpiled_resources}) should
be interpreted as hardware-aware proxy metrics. They are obtained by
transpiling the amplified circuits to \texttt{ibm\_basquecountry} using the same
logical circuit structure and a nine-qubit SWAP-free mapping. They should not
be read as exact resource reports for the Q-CTRL executions, since Q-CTRL's
Performance Management Qiskit Function does not expose the final
internally optimised circuits submitted to the device.

\begin{table}[H]
\centering
\normalsize
\caption{Circuit resource scaling for the amplified validation circuits $Q^k A_{\mathrm{test}}$, where $K = 2k+1$. \textit{Logical resources} (left) report gate counts after decomposing controlled rotations and $R_{ZZ}$ interactions into elementary one-qubit and CNOT gates. \textit{ISA resources} (right) result from hardware-aware fixed-layout transpilation with multi-seed selection at optimisation level~3, targeting \texttt{ibm\_basquecountry} with no SWAP insertions. The lower block details the individual Grover constituents.}
\label{tab:atest_amplified_logical_transpiled_resources}
\setlength{\tabcolsep}{5pt}
\renewcommand{\arraystretch}{1.2}
\setlength{\aboverulesep}{0pt}
\setlength{\belowrulesep}{0pt}
\begin{tabular}{@{\hspace{6pt}} l c c c c c c c c @{\hspace{6pt}}}
\toprule
\multicolumn{1}{@{\hspace{6pt}} l}{}
& &
& \multicolumn{3}{c}{\textbf{Logical Resources}}
& \multicolumn{3}{c @{\hspace{6pt}}}{\textbf{ISA Resources}} \\
\cmidrule{4-6} \cmidrule{7-9}
\multicolumn{1}{@{\hspace{6pt}} l}{\textbf{Circuit}}
& \(\boldsymbol{k}\)
& \(\boldsymbol{K}\)
& \textbf{1q Gates}
& \textbf{CNOTs}
& \textbf{Depth}
& \textbf{1q Gates}
& \textbf{CZs}
& \multicolumn{1}{c @{\hspace{6pt}}}{\textbf{Depth}} \\
\midrule
$Q^{0}A_{\mathrm{test}}$  & $0$  & $1$  & $17$   & $12$  & $26$   & $79$   & $16$   & $64$   \\
$Q^{1}A_{\mathrm{test}}$  & $1$  & $3$  & $69$   & $42$  & $89$   & $259$  & $60$   & $209$  \\
$Q^{2}A_{\mathrm{test}}$  & $2$  & $5$  & $121$  & $72$  & $152$  & $426$  & $100$  & $348$  \\
$Q^{3}A_{\mathrm{test}}$  & $3$  & $7$  & $173$  & $102$ & $215$  & $613$  & $146$  & $493$  \\
$Q^{4}A_{\mathrm{test}}$  & $4$  & $9$  & $225$  & $132$ & $278$  & $785$  & $190$  & $644$  \\
$Q^{10}A_{\mathrm{test}}$ & $10$ & $21$ & $537$  & $312$ & $656$  & $1822$ & $437$  & $1492$ \\
$Q^{28}A_{\mathrm{test}}$ & $28$ & $57$ & $1473$ & $852$ & $1790$ & $5008$ & $1185$ & $4056$ \\
\midrule
\multicolumn{9}{@{\hspace{6pt}} l}{\textbf{Grover constituents}} \\
$A$           & -- & -- & $17$ & $12$ & $25$ & $77$ & $16$ & $61$ \\
$A^{\dagger}$ & -- & -- & $17$ & $12$ & $23$ & $84$ & $18$ & $62$ \\
$S_{0}$       & -- & -- & $17$ & $6$  & $14$ & $37$ & $9$  & $33$ \\
$S_{f}$       & -- & -- & $1$  & $0$  & $1$  & $1$  & $0$  & $1$  \\
\bottomrule
\end{tabular}
\end{table}

The logical resources are reported after decomposing two-qubit gates
into a standard elementary circuit basis consisting of arbitrary one-qubit
gates and CNOTs. This convention follows the standard decomposition framework
of \cite{Barenco1995ElementaryGates}, where one-qubit gates
together with the two-qubit XOR/CNOT gate are shown to form a universal
elementary gate set. The ISA resources are then obtained by transpiling the
same circuits to the backend-native basis, where the elementary entangling
operation may be represented as CZ rather than CNOT.

\begin{table}[H]
\centering
\normalsize
\caption{Circuit resource scaling for amplified CVA circuits $Q^k A_{\Theta}$. \textit{Logical resources} (left) report gate counts after decomposing all controlled operations into elementary gates. \textit{ISA resources} (right) result from hardware-aware transpilation targeting \texttt{ibm\_basquecountry} on a nine-qubit mapping that requires no SWAP insertions. The lower block details the individual Grover ($Q=-AS_0A^{\dagger}S_f$) constituents metrics.}
\label{tab:cva_amplified_logical_transpiled_resources}
\setlength{\tabcolsep}{5pt}
\renewcommand{\arraystretch}{1.2}
\setlength{\aboverulesep}{0pt}
\setlength{\belowrulesep}{0pt}
\begin{tabular}{@{\hspace{6pt}} l c c c c c c c c @{\hspace{6pt}}}
\toprule
\multicolumn{1}{@{\hspace{6pt}} l}{}
& &
& \multicolumn{3}{c}{\textbf{Logical Resources}}
& \multicolumn{3}{c @{\hspace{6pt}}}{\textbf{ISA Resources}} \\
\cmidrule{4-6} \cmidrule{7-9}
\multicolumn{1}{@{\hspace{6pt}} l}{\textbf{Circuit}}
& \(\boldsymbol{k}\)
& \(\boldsymbol{K}\)
& \textbf{1q Gates}
& \textbf{CNOTs}
& \textbf{Depth}
& \textbf{1q Gates}
& \textbf{CZs}
& \multicolumn{1}{c @{\hspace{6pt}}}{\textbf{Depth}} \\
\midrule
$Q^{0}A_{\Theta}$ & $0$ & $1$ & $161$  & $124$  & $213$  & $496$   & $117$  & $322$  \\
$Q^{1}A_{\Theta}$ & $1$ & $3$ & $1007$ & $702$  & $1226$ & $3883$  & $1122$ & $2691$ \\
$Q^{2}A_{\Theta}$ & $2$ & $5$ & $1853$ & $1280$ & $2239$ & $7201$  & $2093$ & $5082$ \\
$Q^{3}A_{\Theta}$ & $3$ & $7$ & $2699$ & $1858$ & $3252$ & $10647$ & $3101$ & $7611$ \\
$Q^{4}A_{\Theta}$ & $4$ & $9$ & $3545$ & $2436$ & $4265$ & $14029$ & $4091$ & $9815$ \\
\midrule
\multicolumn{9}{@{\hspace{6pt}} l}{\textbf{Grover constituents}} \\
$A$           & -- & -- & $161$ & $124$ & $212$ & $496$  & $117$ & $322$  \\
$A^{\dagger}$ & -- & -- & $161$ & $124$ & $222$ & $563$  & $145$ & $358$  \\
$S_{0}$       & -- & -- & $513$ & $324$ & $583$ & $2253$ & $665$ & $1713$ \\
$S_{f}$       & -- & -- & $11$  & $6$   & $12$  & $56$   & $18$  & $43$   \\
\bottomrule
\end{tabular}
\end{table}

\subsection{Calibration results}
\label{appsec:calibration}
\begin{table}[H]
\normalsize
\centering
\caption{Summary of the single-amplitude hardware and contrast-model calibration in the validation experiment with the reduced-depth circuit $A_{\mathrm{test}}$, executed on \texttt{ibm\_basquecountry}. The noise model is the one in \eqref{eq:cabiqae_observation_model}.}
\label{tab:combined_calibration_summary}
\begin{tabular}{l cccc}
\toprule
\multicolumn{5}{c}{\textbf{A. Hardware Calibration Parameters}} \\
\midrule
\multicolumn{3}{l}{Parameter} & \multicolumn{2}{c}{Value} \\
\cmidrule(r){1-3} \cmidrule(l){4-5}
\multicolumn{3}{l}{Hardware-replay target amplitude $a_{\mathrm{true}}$} & \multicolumn{2}{c}{$0.3602728053$} \\
\multicolumn{3}{l}{Grover-power range} & \multicolumn{2}{c}{$k=0,\ldots,40$ (except $k=36$)} \\
\multicolumn{3}{l}{Calibration points (Aggregated / Fit / Excluded)} & \multicolumn{2}{c}{$40$ / $24$ / $16$} \\
\multicolumn{3}{l}{Amplification shots (Total / Typical per point)} & \multicolumn{2}{c}{$430\,080$ / $12\,288$} \\
\multicolumn{3}{l}{$P(\mathrm{obs}=1\mid \mathrm{prep}=0)$} & \multicolumn{2}{c}{$0.00061 \pm 0.00027$} \\
\multicolumn{3}{l}{$P(\mathrm{obs}=1\mid \mathrm{prep}=1)$} & \multicolumn{2}{c}{$0.99865 \pm 0.00041$} \\
\addlinespace
\multicolumn{5}{c}{\textbf{B. Empirical Contrast-Model Calibration}} \\
\midrule
\multicolumn{1}{c}{$c_0$}
    & \multicolumn{2}{c}{Slope of $\log c_k$}
    & \multicolumn{1}{c}{$\tau_c$}
    & \multicolumn{1}{c}{$R^2$} \\
\midrule
\multicolumn{1}{c}{$1.16 \pm 0.10$}
    & \multicolumn{2}{c}{$-0.0294 \pm 0.0024$}
    & \multicolumn{1}{c}{$33.96 \pm 2.80$}
    & \multicolumn{1}{c}{$0.87$} \\
\bottomrule
\end{tabular}
\end{table}

\begin{table}[H]
\normalsize
\centering
\caption{Summary of the hardware and contrast-model calibrations in the
six-qubit CVA amplitude estimation experiment executed on
\texttt{ibm\_basquecountry} with Q-CTRL's Performance Management. Again the noise model is \eqref{eq:cabiqae_observation_model}.}
\label{tab:cva_hardware_combined_calibration_summary}
\begin{tabular}{l cccc}
\toprule
\multicolumn{5}{c}{\textbf{A. Hardware Calibration Parameters}} \\
\midrule
\multicolumn{3}{l}{Parameter} & \multicolumn{2}{c}{Value} \\
\cmidrule(r){1-3} \cmidrule(l){4-5}
\multicolumn{3}{l}{Initial ideal amplitude $a_0$}
    & \multicolumn{2}{c}{$0.152$} \\
\multicolumn{3}{l}{Objective good state}
    & \multicolumn{2}{c}{\texttt{111}} \\
\multicolumn{3}{l}{Grover-power range}
    & \multicolumn{2}{c}{$k=0,\ldots,4$} \\
\multicolumn{3}{l}{Calibration points (Aggregated / Fit / Excluded)}
    & \multicolumn{2}{c}{$5$ / $4$ / $1$} \\
\multicolumn{3}{l}{Amplification shots (Total / Typical per point)}
    & \multicolumn{2}{c}{$330000$ / $66000$} \\
\multicolumn{3}{l}{$P(\mathrm{obs}=111\mid \mathrm{prep}=000)$}
    & \multicolumn{2}{c}{$0.00 \pm 0.00$} \\
\multicolumn{3}{l}{$P(\mathrm{obs}=111\mid \mathrm{prep}=111)$}
    & \multicolumn{2}{c}{$0.99500 \pm 0.00078$} \\
\multicolumn{3}{l}{Fitted contrast baseline $b$}
    & \multicolumn{2}{c}{$0.15$} \\
\addlinespace
\multicolumn{5}{c}{\textbf{B. Empirical Contrast-Model Calibration}} \\
\midrule
\multicolumn{1}{c}{$c_0$}
    & \multicolumn{2}{c}{Slope of $\log c_k$}
    & \multicolumn{1}{c}{$\tau_c$}
    & \multicolumn{1}{c}{$R^2$} \\
\midrule
\multicolumn{1}{c}{$0.67 \pm 0.31$}
    & \multicolumn{2}{c}{$-0.46 \pm 0.13$}
    & \multicolumn{1}{c}{$2.20 \pm 0.60$}
    & \multicolumn{1}{c}{$0.86$} \\
\bottomrule
\end{tabular}
\end{table}

\begin{figure}[H]
    \centering
    \includegraphics[width=1.00\textwidth]{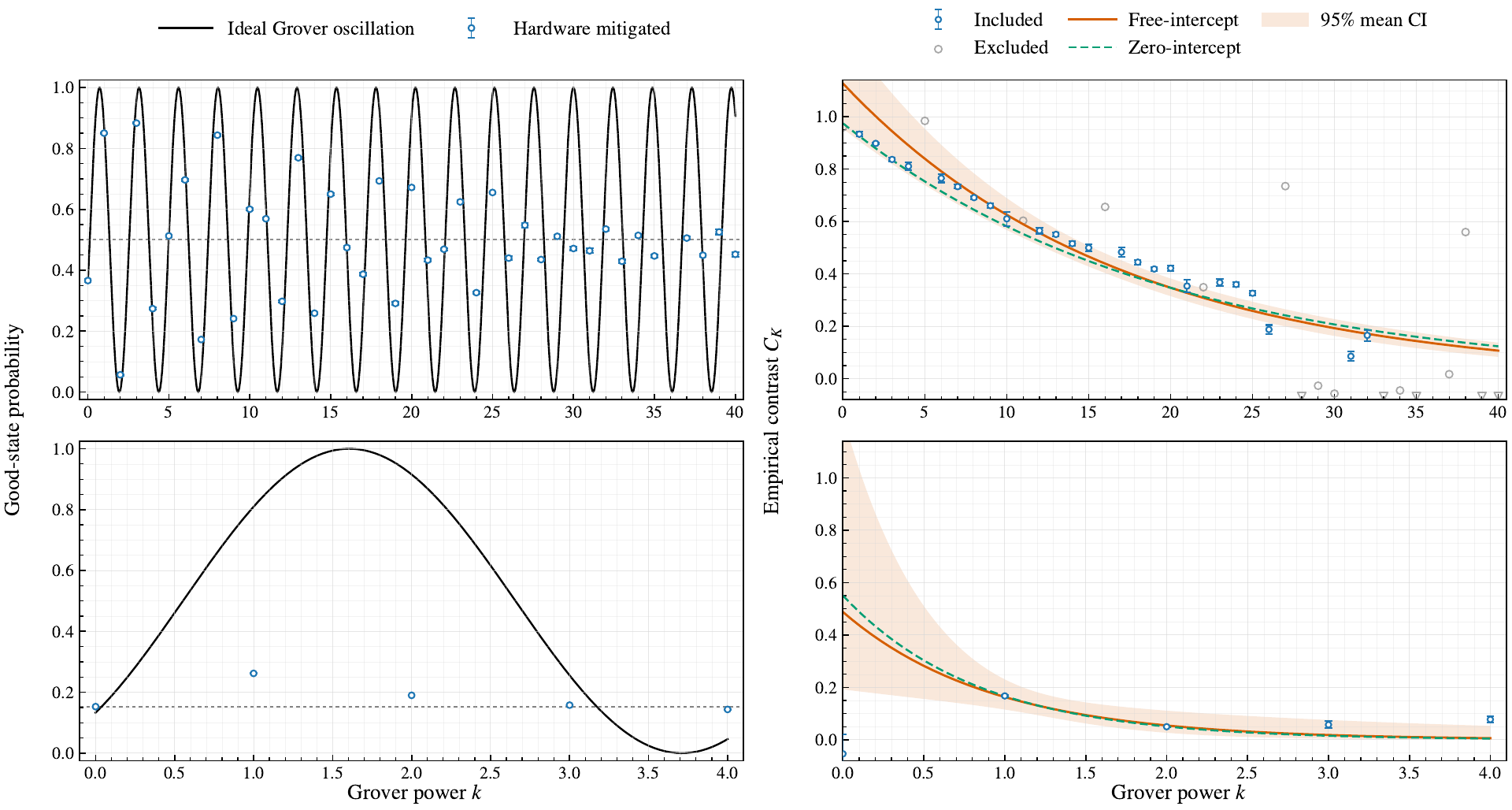}
    \caption{
Experimental characterisation of Grover amplification for the three-qubit toy model $A_{\mathrm{test}}$ (top row) and the CVA circuit $A_{\Theta}$ (bottom row). The $A_{\mathrm{test}}$ hardware calibration uses one target amplitude. Left: mitigated good-state probability and ideal Grover oscillation. Right: empirical contrast, exponential fits, and $95\%$ confidence interval. Error bars denote standard errors; triangles indicate values outside the displayed range.
}
    \label{fig:calibration}
\end{figure}

\subsection{Amplitude estimation}
\label{appsec:ae_extended_results}

\begin{figure}[H]
    \centering
    \includegraphics[width=1.00\textwidth]{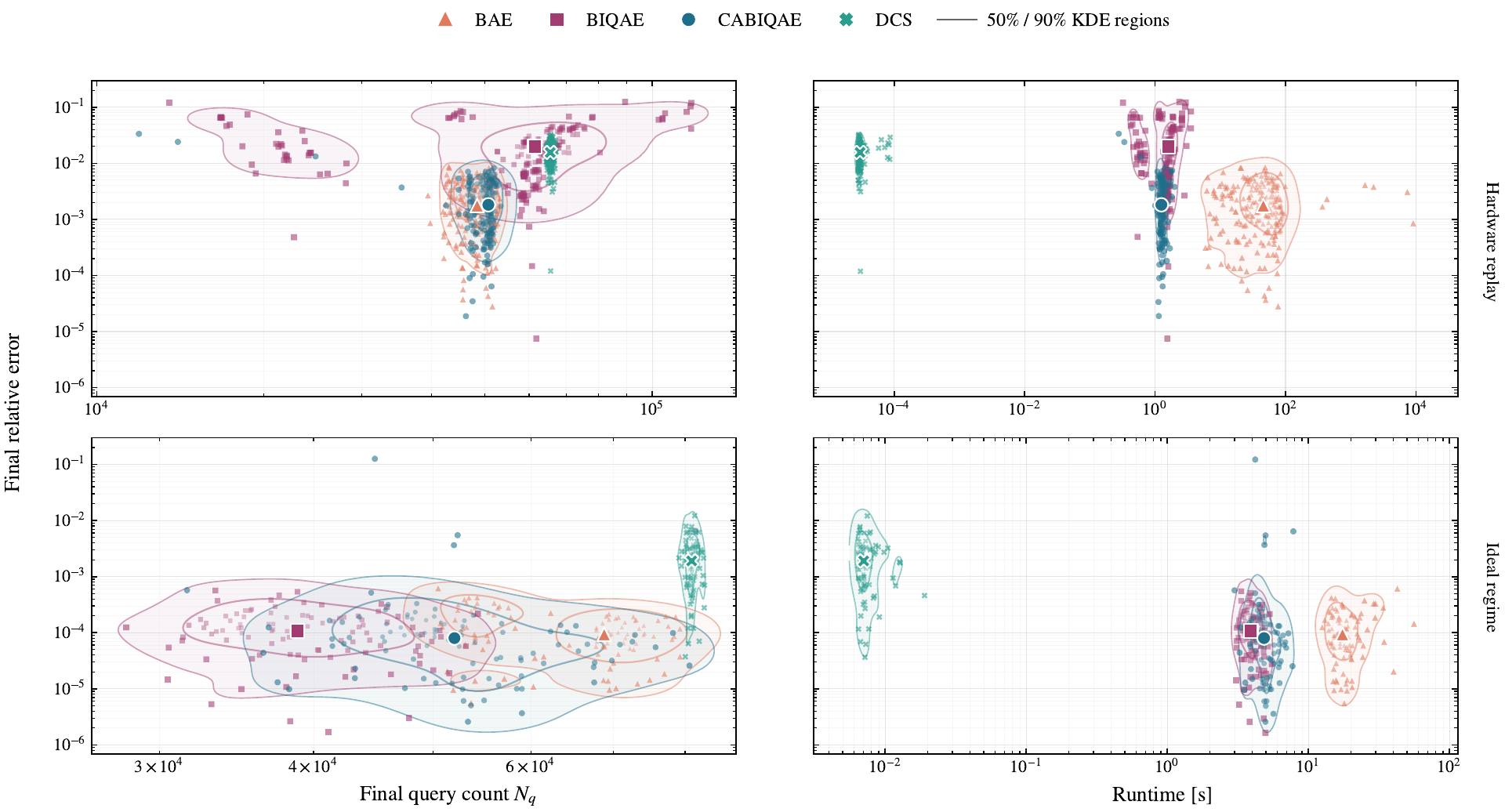}
   \caption{Final relative error versus query cost (left) and runtime (right) for the CABIQAE validation experiment under hardware-replay (top) and the ideal regime (bottom). Markers show individual runs, enlarged symbols indicate medians, and shaded contours delimit the $50\%$ and $90\%$ KDE regions.}
    \label{fig:ae_validation_final_error_density}
\end{figure}

\begin{figure}[H]
    \centering
    \includegraphics[width=1.00\textwidth]{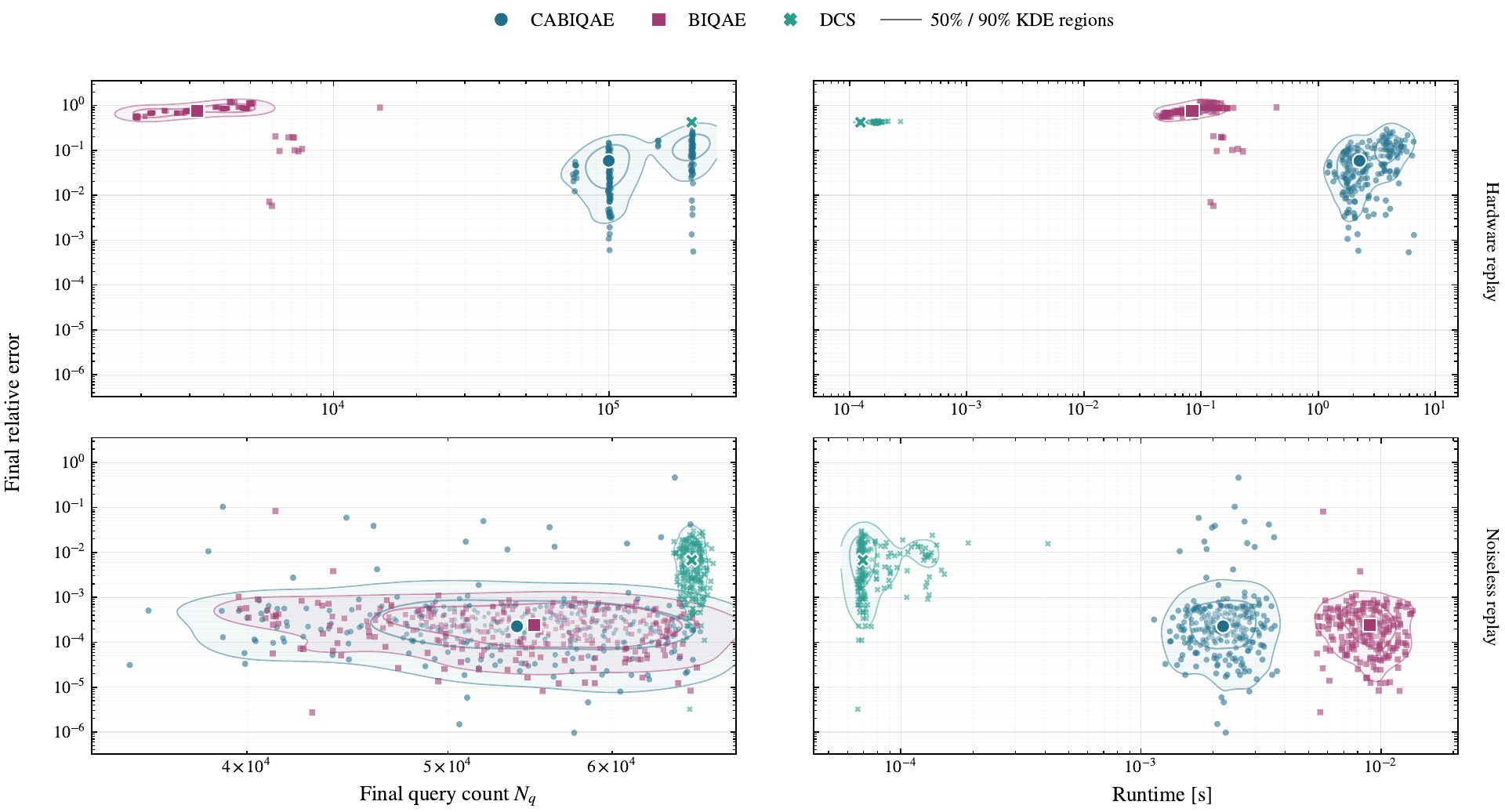}
    \caption{Final relative CVA error versus query cost (left) and runtime (right) for the CVA amplitude-estimation experiment under hardware-replay (top) and the ideal regime (bottom). Markers show individual runs, enlarged symbols indicate medians, and shaded contours delimit the $50\%$ and $90\%$ KDE regions.}
    \label{fig:cva_final_error_density}
\end{figure}
\end{document}